\documentclass[12pt,a4paper,openright]{book}
\usepackage{amsthm,amssymb,latexsym,amsmath,subfigure,amsxtra}
\usepackage[margin=2.5cm]{geometry}
\usepackage{rotating, booktabs}  
\usepackage[titletoc,toc,title]{appendix}
\usepackage{amsfonts,euscript,mathrsfs}
\usepackage{graphicx,color,xcolor}
\usepackage{mathtools}
\usepackage{yhmath}
\usepackage{arcs}
\usepackage{slashed}
\usepackage{cancel}
\usepackage{paralist}
\usepackage{hyperref}
\usepackage{empheq}
\usepackage{booktabs}
\usepackage{array}
\usepackage[all,cmtip]{xy}
\usepackage[affil-it]{authblk}
\usepackage{epsf}
\usepackage{enumerate}
\usepackage{hhline}
\usepackage{array}
\usepackage{tabularx}
\usepackage{hangcaption}

\usepackage{lipsum} 
\usepackage{microtype} 
\usepackage{lettrine} 

\newtheorem{theorem}{Theorem}
\newtheorem{proposition}{Proposition}

\newtheorem{corollary}{Corollary}
\newtheorem{definition}{Definition}

\newtheorem{remark}{Remark}
\newtheorem{ex}{Example}

\newtheorem{lemma}{Lemma}

\def\a{\alpha}
\def\b{\beta}
\def\e{\eta}
\def\g{\gamma}
\def\k{\kappa}
\def\s{\sigma}

\def\vt{\vartheta}

\def\m{\mu}
\def\n{\nu}

\def\veps{\varepsilon}

\def\bvt{\bar{\vartheta}}
\def\o{\omega}
\def\tdo{\widetilde{\omega}}
\def\O{\Omega}

\def\lr{\rfloor}

\def\p{\partial}
\def\G{\Gamma}

\def\d{\delta}
\theoremstyle{plain}

\makeatletter
\renewcommand{\@chapapp}{}
\newenvironment{chapquote}[2][2em]
  {\setlength{\@tempdima}{#1}%
   \def\chapquote@author{#2}%
   \parshape 1 \@tempdima \dimexpr\textwidth-2\@tempdima\relax   \itshape}
 {\par\normalfont\hfill--\ \chapquote@author\hspace*{\@tempdima}\par\bigskip}
 \makeatother

\begin{document}


\begin{titlepage}  
\begin{center}
\Large Department of Physics\\
\Large National Tsing Hua University\\
\Large Doctoral Dissertation\\
\vspace*{15ex}
\huge Gravitational Theories with Torsion\\
\vspace*{15ex}
\Large Huan-Hsin Tseng\\

\Large Advisor\\
\Large Prof. Dr. Chao-Qiang Geng\\
\vspace*{2ex} \Large June, 2015
\end{center}
\end{titlepage}


\pagenumbering{roman}  

\chapter*{Abstract}
\addcontentsline{toc}{chapter}{Abstract}

We give a complete formulation of Poincar\'{e} gauge theory, starting from the fibre bundle formulation to the resultant Riemann-Cartan spacetime. We also introduce several diverse gravity theories descendent from the Poincar\'{e} gauge theory. Especially, the cosmological effect of the simple scalar-torsion ($0^+$) mode in Poincar\'{e} gauge theory of gravity is studied. In the theory, we treat the geometric effect of torsion as an effective quantity,
which behaves like dark energy, and study the effective equation of state (EoS) of the model.

We concentrate on the two cases of the constant curvature solution and positive kinetic energy. In the former, we find that the torsion EoS has different values corresponding to the
stages of the universe. For example, it behaves like the radiation (matter) EoS of $w_r = 1/3$ $(w_m =0)$ in the radiation (matter) dominant epoch, while in the late time the torsion density is supportive for the accelerating universe. In the latter case of positive kinetic energy, we find the (affine) curvature is not constant in general and hence requires numerical solution. Our numerical analysis shows that the EoS in general has an asymptotic behavior in the high redshift regime, while it could cross the phantom divide line in the low redshift regime. By further analysis of the Laurent series expansion, we find that the early evolution of the torsion density
$\rho_T$ has a radiation-like asymptotic behavior of $O(a^{-4})$ where $a(t)$ denotes the scale factor, along
with a stable point of the torsion pressure $(P_T)$ and a density ratio $P_T/\rho_T \to 1/3 $ in the high redshift regime $(z \gg 0)$, this is different from the previous result in the literature. Some numerical illustrations are also demonstrated.

We construct the extra dimension theory of teleparallel gravity by using differential forms. In particular, we discuss the Kaluza-Klein and braneworld scenarios by direct dimensional reduction and specifying the shape of fibre. The FLRW cosmological scenario of the braneworld theory in teleparallel gravity demonstrates its equivalence to general relativity (GR) in the field equations, namely they possess the same Friedmann equation.

\chapter*{Acknowledgements}
\addcontentsline{toc}{chapter}{Acknowledgements}
This Ph.D. thesis cannot be done without Prof.~Chao-Qiang Geng, for he provided an excellent opportunity and an enthusiastic environment for me to delve into gravity theory and cosmology. Most of all, Prof.~Chao-Qiang Geng was kind to allow me pursuing my own way of understanding gravitational physics and gauge theories via fundamental mathematics. 

I am mostly in deep debt to Prof.~Friedrich W. Hehl in University of K\"{o}ln, for my complete understanding towards the Poincar\'e gauge theory was based on his careful guidance and patience via countless electronic communication. I am also in great honor to thank Prof.~James M. Nester in National Central University for he guided me studying mathematical physics and axiomatized gravitational theory in early Ph.D stage.

Much of my published research on cosmology was under collaboration with my excellent coworkers Chung-Chi Lee and Ling-Wei Luo. I am also grateful for National Center for Theoretical Sciences (NCTS) for travel sponsorship of several seminars and conferences. 

My deep gratitude also goes to Prof.~Larry Ford in Tufts University for his hospitality of hosting my visit in Boston and providing many interesting ideas for collaborated investigation on fluctuation of quantum fields.

Finally, my sincere appreciation goes to those who indirectly help me on accomplishing understanding of general relativity, such as books and researches by David Bleecker, Andrzej Trautman, J\"{u}rgen Jost, and Theodore Frankel, for you have shown me the beauty that never fades.

This Ph.D. thesis is dedicated to my parents and my best friend Hsin-Yi Lin.


\newpage

\frontmatter
\tableofcontents  
\newpage

\pagenumbering{arabic}  


\mainmatter

\chapter{Introduction} \label{intro}

The recent cosmological observations, such as those from type Ia
supernovae~\cite{obs1, obs11}, cosmic microwave background
radiation~\cite{obs12, arXiv:1001.4538}, large scale
structure~\cite{astro-ph/0501171, obs13} and weak
lensing~\cite{astro-ph/0306046}, reveal that our universe is subject
to a period of accelerated expansion.

Although general relativity (GR) developed in the last century has been successful in many ways
of explaining various experimental results in gravity, the nature of the accelerating universe now
rises as a small cloud shrouding it. We thereby look for a more
general theory that comprises GR yet being able to explain the
accelerating problem referred to as dark energy~\cite{DE}.

In general, there are two ways to resolve the phenomenon of the late-time accelerated universe~\cite{DE} either
by \emph{modified gravity} or by \emph{modified matter theories}. Modified gravity asserts that considering alternative geometry may be responsible for forces that we are not able to explain, usually by modifying geometric Lagrangians or changing the geometric framework of spacetime. Modified matter theories include some negative pressure matter that could result in the expanding effect. In this thesis, we adopt the viewpoint of
an alternative gravity theory by considering the so-called
Poincar\'{e} gauge theory (PGT)~\cite{Hehl:1976kj, Obukhov:1987tz,Hehl:1994ue}, which integrates the gauge covariant idea into spacetime.

PGT starts with the consideration of gauging the
Poincar\'{e} group $\mathcal{P} = \mathbb{R}^{1,3} \rtimes SO(1,3)$, where $\mathbb{R}^{1,3}$ denotes the Minkowski spacetime $(\mathbb{R}^4, \langle \cdot, \cdot \rangle_{\mathbb{R}^{1,3}} )$ and $ \langle e_\m, e_\n \rangle_{\mathbb{R}^{1,3}} = \text{diag}(-1,+1,+1,+1)$, into
gravity and ends up in effect as a Riemann-Cartan spacetime $(M,g,\nabla)$,
where $M$ is the spacetime manifold, $g$ is a metric and $\nabla$ is a general metric-compatible connection on $M$.
Such a general connection can be decomposed into
$\nabla = \overline{\nabla} - K$ such that $\overline{\nabla}$ is the Riemannian one and $K$ is the contortion tensor related to torsion tensor $T$ of $\nabla$.
As a result, PGT is in general a gravitational theory with torsion~\cite{Hehl:1976kj,Obukhov:1987tz}
that couples to the spin of the matter field. Gauge theory with the Poincar\'{e} group can be considered as a natural extension of GR, in the sense that it
contains GR as a degenerate case. In fact, it comprises a large class including GR, the Einstein-Cartan theory~\cite{Trautman:2006fp}, teleparallel gravity, and quadratic Poincar\'{e} gauge theory.

Historically, Einstein effectively assumed vanishing torsion ad hoc in 1915, later in 1928 he attempted to utilize the teleparallel (purely torsional) theory to unify gravitation and electromagnetism \cite{Einstein:ap}. Around the same time, \'{E}lie Cartan, as a mathematician who constructed the idea of torsion in 1922, communicated with Einstein about his work. In their sequence of communications \cite{ECletter}, they set up a large portion of the foundation for the teleparallel gravity theory of nowadays.

The birth of gauge theory was innovated in the hand of Hermann Weyl in 1918 \cite{Weyl1918}, while in 1929 he achieved the concept of $U(1)$-gauge theory \cite{Weyl1929} we know nowadays and introduced the vierbein (orthonormal basis, tetrad) into general relativity. The success of local gauge theory in 1950s brought new life into the gravity with torsion. Utiyama gave a first attempt in gauging $SO(1,3)$ into spacetime without success \cite{Utiyama:1956sy}, mainly due to the Riemannian connection used. On this track, Sciama then introduced torsion and related it to spin \cite{Sciama}, and later Kibble showed how to describe gravity with torsion as a local gauge theory of the Poincar\'{e} group \cite{Kibble}. Later in 1976 Hehl \emph{et al} formulated a complete gravitation theory that demonstrates Poincar\'{e} gauge invariance and eventually results in gravity with torsion \cite{Hehl:1976kj}. The great success of gauge theory in fundamental physics leads us to believe that gravity should also belong to the roll of gauge theories, since all the other fundamental interactions like the electroweak and the strong are beautifully formulated by such rules. In this sense, this provides a best guiding principle to follow in searching for an alternative gravity theory.

The framework of PGT is based on the gauge principles of Yang-Mill's theory of non-Abelian group. Through the use of principal fibre bundles in mathematics, one derives a more clear vision of gauge structures and its essence which is eventually beneficial for the transition between different Lie groups.

We set out from the formulation of PGT in fibre bundle language \cite{McInnes:1984sm},\cite{Trautman:1970cy} (Chapter 2) which is also general for all gauge theories, and try to address the story of PGT as complete as possible in a united and compact way with the minimal offering of bundle materials for essential study. Such construction will then provide an integrated and clear view for the Poincar\'{e} gauge gravity and Riemann-Cartan spacetime, which in the end leads to several diverse alternative gravity theories, such as the Einstein-Cartan theory, GR, teleparallel gravity, and quadratic PGT.

In particular, we shall investigate a specific quadratic theory of the \emph{scalar-torsion mode} in PGT \cite{Tseng:2012hn},\cite{Geng:2013hp},\cite{Geng:2013hp}(Chapter 3), that possesses dynamical torsion field. This particular mode is also called \emph{simple $0^+$ mode or SNY-model,\cite{Shie:2008ms}}, which is
one of the six modes: $0^{\pm},1^{\pm}$ and $2^{\pm}$ labeled by spin and parity, based on the linearlized
theory of PGT~\cite{Hayashi:1981mm,Sezgin:1979zf}. In the theory, the $0^+$ mode is known to have no interaction with any fundamental source~\cite{Kopczynski} and thus it could have a significant magnitude without being much noticed within the current universe. In particular, it has been studied by Shie, Nester and Yo (SNY) that the spin-$0^+$ mode is divided into two classes: one with negative energy density but exhibiting late time de-Sitter universe served as dark energy; the other with normal positive energy condition that is also responsible for the late-time acceleration but demonstrates the early oscillation in various physical quantities. In this thesis, we concentrate on these two classes and present the numerical solutions of the late-time acceleration behavior and their corresponding equation of state (EoS), defined by $w=p/\rho$, where $\rho$ and $p$ are the energy density and pressure of the relevant component of the universe respectively. This type of cosmology provides some realistic and interesting features so that it has been explored in numerous discussions~\cite{Chen:2009at,Li:2009zzc,Li:2009gj,
Ho:2011qn,Ho:2011xf,Ao:2010mg,Baekler:2010fr,Ao:2011kc,Xi:2011uz}. Consequently, this mode naturally becomes a subject to study~\cite{Hehl:2012pi}.

In PGT, there is another interesting degenerate case called \textit{teleparallelism} (in contrast to GR of zero torsion), where curvature vanishes identically on the spacetime and torsion is the only responsiblility for the gravitational force. The name is so dubbed simply because without curvature every vector field is parallel. Such construction is possible if we adopt the so-called Weitzenb\"{o}ck connection $\nabla^W$ on a Riemann-Cartan space.  As indicated early, such a spacetime was considered by Einstein \cite{Einstein:ap}, who had unified theory concerns, and later it was developed into a type called \emph{teleparallel equivalent to general relativity} (TEGR). The Lagrangian is in a special form such that it is \textit{almost equivalent} to GR in every aspect. TEGR is equivalent to GR in the field equations and the matter evolution so that they cannot be told from the dynamics. One distinction between TEGR and GR is the local Lorentz violation at the Lagrangian level \cite{Li:2010cg}, i.e, TEGR does not respect local Lorentz transformation. However such local Lorentz violation terms appear in the form of an exact differential such that it can be regarded as a boundary term were we in a closed manifold. In any case this term does not affect the field equations such that it gives the same action as GR, which is what we meant by ``almost equivalent''. We also give an account from bundle formalism why such violation occurs. In Chapter 4, we shall provide more discussions.

As is known, there exists another type of modified gravity theory called \textit{extra dimension theory} that consists of a higher dimensional spacetime (called \emph{bulk}, generally higher than four) and a 4-dimensional submanifold as our living space(-time). The five dimension gravity generally induces gravity on the four-dimensional spacetime along with a type of force. First extra dimension theory that unifies electromagnetism and gravitation was initiated by Nordstr\"{o}m \cite{Nordstrom} around 1914 as well as Kaluza \cite{Kaluza} and Klein \cite{Klein}, known as the KK theory. In KK theory, electromagnetic field is from the projection of five-dimensional spacetime whose fibre is a small circle $S^1 \cong U(1)$. It is usually used to explain the hierarchy problems with the effective Planck scale in 4-dimension by dimensional reduction. There is another type of extra dimension theory called \emph{large extra dimension} or ADD model, proposed by Arkani-Hamed, Dimopoulos and Dvali \cite{ADD} in 1998. It was proposed to explain why gravity is so weak compared to other forces. The ADD theory assumes that the fields of the Standard Model are confined on the 4-dimensional \emph{membrane}, with only gravity being able to propagate through the large extra dimension that is spatial. Thus it is also referred to as the \textit{braneworld} theory.

We construct the extra dimension theory for TEGR gravity \cite{Geng:2014yya},\cite{Geng:2014nfa}. In order to build such theory, it is necessary to search for the torsion relations between the brane and the bulk mimicking the Gauss-Codacci equation. In general, such relations could be complicated in component form. Thus we adopt differential forms
to reduce the large amount of computation and to serve as a rigorous tool. In the construction, we keep our geometric setting as general as possible to contain branworld theory and KK theory under the same geometric framework. Some of the aspects have been explored in the literature~\cite{deAndrade:1999vq, Barbosa:2002mg,
Fiorini:2013hva, Bamba:2013fta, Nozari:2012qi}, which can be compared with our results. In the end, we utilize our extra dimension theory for TEGR for FLRW cosmology as an application and derive a result consistent with GR.


\chapter{Poincar\'{e} Gauge Gravity Theory} \label{PGT}


\begin{chapquote}{\textit{A. Trautman}}
\ldots \textit{It is
possible that ECT will prove to be a better classical
limit of a future quantum theory of gravitation than
the theory without torsion.}
\end{chapquote}

Poincar\'{e} Gauge Theory for gravity (PGT) is a theory that
incorporates gravity as a gauge theory of the Poincar\'{e} group $\mathcal{P} = \mathbb{R}^{1,3} \rtimes SO(1,3)$. To
develop such a theory, one needs to find an ambient space where both theories
cooperate.

Recall that to describe (local) gauge symmetry, one requires gauge invariance of an internal group $G$, that is a Lie group, and the covariant transition induced by an external group of diffeomorphisms Diff$(M)$ between two
observers of spacetime, and hence such a theory must be based on a 4-dimensional Lorentzian
manifold. For a gauge theory modelled on a spacetime, a principal
fibre bundle is then a natural candidate. For example,
electromagnetism of $U(1)$ symmetry can be formulated on a
$U(1)$-principal bundle.

Therefore one finds that the best suited mathematical theory that depicts PGT
is the principal fibre bundle theory. Below we mainly follow the treatment of \cite{Trautman:1970cy}, \cite{Steenrod}, \cite{Kobayashi}, \cite{Bleecker}, \cite{Jost1}, and \cite{Trautman:1979cq} in principal fibre bundle theory to provide essential bundle material to clearly address PGT.

\section{Preliminaries and Notations}

\subsection{Geometric construction of PGT}

\begin{definition}{(Principal $G$-bundle)}\label{Def:PFB}

Let $G$ be a Lie group. A principal $G$-bundle consists of a pair of differentiable manifolds $P$ called the \textbf{total
space} and $M$ called \textbf{base manifold} with a differentiable (surjective) projection $\pi:P \to M$ and an action of $G$ on $P$ such that
\begin{enumerate}
\item
For every $g\in G$, there exists a diffeomorphism $R_g : P \to P$
such that $R_{g_1 g_2}(p) = R_{g_2} \circ R_{g_1}(p)$ for all $g_1$,
$g_2 \in G$ and $p\in P$. And if $e \in G$ is the identity element,
then $R_e (p) = p $ for all $p\in P$. We also require the group action of $G$ acts freely on $P$, $
(p,g) \in P \times G \mapsto R_g(P) \in P$ such that $ R_g(p)
\neq p $ for all $g\neq e$. We also write $R_g (p) = p g$.

The action of $G$ on $P$ then defines an equivalence relation. Define
$p\sim q$ for $p$, $q\in P$ $\Leftrightarrow$ if there exists $g\in
G$ such that $p=qg$.

\item
$M$ is the quotient space of the equivalence relation $\sim$ induced
by $G$, $M=P/G$. Hence $\pi^{-1}(\pi(p))= \{p g| g\in G\}$ (the
orbit of $G$ through $p$). If $x\in M$, then $\pi^{-1}(x)$ is called the
\textbf{fibre} above $x$.

\item
$P$ is locally trivial. For each $x\in M$, there exists an open set
$U$ containing $x$ and a diffeomorphism $T_U:\pi^{-1}(U) \to U
\times G$ such that $T_U(p)=(\pi(p), \psi_U(p))$ satisfying $
\psi(p g) = \psi(p) g$ for all $g\in G$. The map $T_U$ is
called a \textbf{a local trivialization} or a \textbf{choice of
gauge} (in physics language).

\end{enumerate}

\end{definition}

In practice the base manifold $M$ corresponds to
the 4-dimensional spacetime in consideration, while for all $p \in \pi^{-1}(x)$ there exists a map $ G \to \pi^{-1} (x)$ by
$g \mapsto pg$, which is a diffeomorphism depending on $p$. Thus all
fibres $\pi^{-1}(x)$ for $x \in M$ are isomorphic to $G$ called the
\textbf{internal (symmetry) group} or gauge group.

A principal $G$-bundle provides a space for phase factors, see \cite{Wu:1975es}, while a
connection on $P$ yields a \textbf{gauge potential} in physics
language, which we denote

\begin{definition}
Let $\mathfrak{g}$ be the Lie algebra of $G$. A connection is a
$\mathfrak{g}$-valued 1-form  $\omega $ on $P$, denoted by $\Lambda^1(P;
\mathfrak{g}) := \Lambda^1(P)\otimes \mathfrak{g}$, such that

\begin{enumerate}

\item
Let $A\in \mathfrak{g}$, define the \textbf{fundamental vector
field} corresponding to $A$ on $P$ by
\begin{equation}\label{E:fundamental vector field}
A_p^* := \frac{d}{dt} \left( p \exp (t A) \right)\Bigr|_{t=0}
\end{equation}
then we require $\omega_p (A_p^*) = A$.

\item
For $g\in G$, let $\mathfrak{ad}_g : \mathfrak{g} \to \mathfrak{g} $
be the associated adjoint map\footnote{The associated adjoint map $\mathfrak{ad}_g : \mathfrak{g} \to \mathfrak{g} $ is
defined as the differential map of $Ad_g : G \to G$ at $e \in G$, where $Ad_g(h)
:= ghg^{-1}$. Specifically, let $a\in \mathfrak{g}$ find a local
curve $t\mapsto c(t)$ on $G$ such that $c(0)=e$, $c'(0) = a $, then
define $\mathfrak{ad}_g (a) = \frac{d}{dt} \left( g \, c(t)
g^{-1}\right)\Bigr|_{t=0} $ }. We require $\omega_{pg}(R_{g*} X_p) =
\mathfrak{ad}_{g^{-1}} \omega_p(X_p)$ for all $g\in G$, $p\in P$ and
$X \in TP$ a vector field on $P$. i.e, $R_g^* \omega =
\mathfrak{ad}_{g^{-1}} \omega$
\end{enumerate}

\end{definition}

On a principal fibre bundle $P$, the notion of \emph{vertical} is already distinguished by the projection $\pi: P \to M$, specifically if $w\in T_pP$, a vector on $P$, is \textbf{vertical} if $\pi_*(w) = 0$. In contrast, for the notion of \emph{horizontal} we need extra structure to specify it. A connection $\omega$ on $P$ essentially tells us what vector is \textit{horizontal}, defined by $\o(v) =0$ if $v\in T_pP$. Therefore given a connection 1-form $\o$ on $(P,\pi,M,G)$, we may decompose a vector field $w\in T_pP$ as $w = w^V+ w^H$ such that $w^V $ is vertical $\pi_*(w^V) =0$ and $w^H$ is horizontal $\o(w^H) = 0$.

With the technical definitions above, we can then describe gravitational gauge theories. We recall that Einstein's general relativity can be reformulated on the following principal bundle whose structure group is $SO(1,3)$,

\begin{definition}{(Linear frame bundle)}\label{Def:frame bundle}

Let $M$ be a 4-dim Lorentz manifold (spacetime). Define a
\textbf{generalized frame (a generalized observer)} at $x\in M$ to
be a linear isomorphism $u: \mathbb{R}^{1,3} \to T_xM $, and thus
$\{u(e_0), \ldots, u(e_3)\} \in T_xM $ is a basis for $T_xM$ (not necessarily orthonormal), where $e_0,\ldots,e_3 $ is the
canonical basis of the Minkowski spacetime $\mathbb{R}^{1,3}$. Let $L_x(M):= \{ u_x:
\mathbb{R}^{1,3} \to T_xM  \, \mbox{a isomorphism}\}$ be the set of all generalized observers at $x \in
M$. Define the set
\begin{equation}\label{L(M)}
L(M) := \bigcup_{x\in M} L_x(M)
\end{equation}
with $\pi(u_x):= x$ for $u_x\in L_x(M)$,, and the action of
$GL(\mathbb{R}^{1,3})$ on $L(M)$ by $R_A: L(M) \to
L(M)$ with $R_A(u_x) := u_x
\circ A$ for $A\in GL(\mathbb{R}^{1,3})$. Equipped with a suitable differentiable structure on $L(M)$, one
verifies that $(L(M),\pi, M, GL(\mathbb{R}^{1,3})$ is a
principal fibre bundle called the \textbf{linear frame bundle}.

\end{definition}

One sees that if $u_x$ is an observer at $x$, then the right action
$R_A$, $A\in GL(\mathbb{R}^{1,3})$, brings it from one
frame to another. Furthermore, on a semi-Riemannian\footnote{The Semi-Riemannian geometry denotes a differentiable manifold $M$ and a Lorentz metric $g$ that has one negative signature. } spacetime
$(M,g)$ we can consider an \textbf{orthonormal frame} (or \textbf{tetrad}) $u_x \in
L_x(M)$ at $x\in M$ such that the $g(u_x (v), u_x (w)) = \langle v,w \rangle_{\mathbb{R}^{1,3}}$ for
all $v, w \in \mathbb{R}^{1,3} $. An orthonormal frame then corresponds to an
\textbf{observer} in the usual physical sense. We can further define an
\textbf{orthonormal frame bundle}.

\begin{definition}{(Orthonormal frame bundle)}

Let $F_x(M):= \{ u_x \in L_x(M) \, | \, g\circ u_x = \langle \cdot,\cdot \rangle_{\mathbb{R}^{1,3}} \}$ be the
set of all orthonormal frames. Define the set
\begin{equation}\label{F(M)}
F(M) := \bigcup_{x\in M} F_x(M)
\end{equation}
with $\pi: F(M) \to M$ by $\pi(u_x) = x$ for $u_x \in F_x(M)$, and
the right action of $G=O(1,3)$ by $R_A (u_x) = u_x \circ A $ for $A
\in O(1,3)$. With the differential structure induced from $L(M)$,
one can verify that $F(M)$ is also a principal fibre bundle called the
\textbf{orthonormal frame bundle}.

\end{definition}

Since we know that the hyper-rotation group
\[
O(1,3):=\{ A\in
GL(\mathbb{R}^{1,3} ) \, | \, \langle Av, Aw \rangle_{\mathbb{R}^{1,3}}  =
\langle v, w \rangle_{\mathbb{R}^{1,3}} \}
\]
where $\langle \cdot,\cdot \rangle_{\mathbb{R}^{1,3}} $ is the
Lorentzian inner product of Minkowski spacetime, has 4 components:
\begin{equation}
\begin{aligned}
L_+^{\uparrow} &: = \{ B \in O(1,3)\, | \, \det B = 1, B^0_0 \geq 1
\}, \\
L_-^{\uparrow}  &:= \{ B \in O(1,3)\, | \, \det B = -1,
B^0_0 \geq
1 \}\\
L_+^{\downarrow} &: = \{ B \in O(1,3)\, | \, \det B = 1, B^0_0 \leq
- 1 \}, \\
L_-^{\downarrow} &:= \{ B \in O(1,3)\, | \, \det B =
-1, B^0_0 \leq - 1 \}
\end{aligned}
\end{equation}
where the connected component $L_+^{\uparrow}$ is usually referred
as the Lorentz group $SO(1,3)$. It follows that $F(M)$ can
contain up to 4 components in general. If we assume $F(M)$ has 4
components for simplicity, then the base manifold (spacetime) $M$ is
then called \textbf{space and time orientable}. A choice of one
component for $F(M)$ corresponds to a \textbf{space and time
orientation}. Let $F_0(M)$ be such a choice, then the principal
fibre bundle $(F_0(M),\pi, M, SO(1,3))$ describes a spacetime with
Lorentz gauge covariance. Since the fibre bundle $(F_0(M),\pi, M,
SO(1,3))$ construction is canonical whenever $M$ exists, we conclude
that Einstein's general relativity has local Lorentz gauge freedom.

So far, we have established the living space for gauge gravitational
theories, yet we do not have notion of curvature or shape for $P$ at this moment. Curvature is the core of general relativity in the sense that it is responsible
for the gravitational force one perceives but one should emphasize that curvature is not a property pertaining to a spacetime. In the theory
of General Relativity in 1915, Einstein adopted the unique
Levi-Civita (Riemannian) connection for defining curvature of a spacetime. The following reformulation of the Levi-Civita connection on a linear frame bundle equivalent to the one in general GR textbooks is helpful for PGT generalization later. We begin with defining curvature

\begin{definition}\label{Def:Exterior covariant derivative}

If $\psi \in \Lambda^k(P, V)$, then we define $\psi^H \in \Lambda^k(P, V)$ by $\psi^H  (X_1, \ldots ,X_k) := \psi (X_1^H, \ldots , X_k^H)$, where $X_i \in TP$.
\end{definition}

\begin{definition}{(Exterior covariant derivative)}

If $\psi \in \Lambda^k(P, V)$, then we define $ D^\o \psi := (d\psi)^H \in \Lambda^{k+1}(P,V)$ with respect to the connection $\o$.
\end{definition}

With the exterior covariant derivative in particular for $V= \mathfrak{g}$ of $G$, we can define the notion of curvature on a principal fibre bundle.
\begin{definition}{(Curvature of a connection)}\label{Def:curvature}

If $\o$ is a connection 1-form on $P$, the \textbf{curvature} of the connection is defined as the 2-form $\O^\o := D^\o \o \in \Lambda^2(P, \mathfrak{g})$.
\end{definition}

In terms of physics context, $\o$ is referred to as the \textbf{gauge potential} and $\O^\o $ is the \textbf{field strength} (corresponding to $\o$). The following form is usually more familiar in the physics literature than Definition (\ref{Def:curvature}).
\begin{theorem}

If $G$ is a matrix group, then the curvature 2-form can be expressed by
\begin{equation}\label{E:connection curvature}
\O^\o = d\o + \o \wedge \o
\end{equation}
\end{theorem}
where $\o$ is regarded as a matrix in $\mathfrak{g}$ with each entry a real-valued 1-form $ \o_\m{}^\n \in \Lambda^1(P, \mathbb{R})$, write $\o = \begin{pmatrix}
  \o_1{}^1 & \o_1{}^2 & \cdots & \o_1{}^n \\
 \o_2{}^1 & \o_2{}^2 & \cdots & \o_2{}^n \\
  \vdots  & \vdots  & \ddots & \vdots  \\
  \o_n{}^1 & \o_n{}^2 & \cdots &\o_n{}^n
 \end{pmatrix}$ and the wedge $\o \wedge \o$ is defined as
\begin{equation}\label{E:Lie alg exp}
 \begin{pmatrix}
  \o_1{}^1 & \o_1{}^2 & \cdots & \o_1{}^n \\
 \o_2{}^1 & \o_2{}^2 & \cdots & \o_2{}^n \\
  \vdots  & \vdots  & \ddots & \vdots  \\
  \o_n{}^1 & \o_n{}^2 & \cdots &\o_n{}^n
 \end{pmatrix} \wedge
  \begin{pmatrix}
  \o_1{}^1 & \o_1{}^2 & \cdots & \o_1{}^n \\
 \o_2{}^1 & \o_2{}^2 & \cdots & \o_2{}^n \\
  \vdots  & \vdots  & \ddots & \vdots  \\
  \o_n{}^1 & \o_n{}^2 & \cdots &\o_n{}^n
 \end{pmatrix}
 =
   \begin{pmatrix}
  \o_1{}^k \wedge \o_k{}^1 & \o_1{}^k \wedge \o_k{}^2 & \cdots & \o_1{}^k \wedge \o_k{}^n \\
  \o_2{}^k \wedge \o_k{}^1 & \o_2{}^k \wedge \o_k{}^2 & \cdots & \o_2{}^k \wedge \o_k{}^n \\
  \vdots  & \vdots  & \ddots & \vdots  \\
  \o_n{}^k \wedge \o_k{}^1 & \o_n{}^k \wedge \o_k{}^2 & \cdots & \o_n{}^k \wedge \o_k{}^n
 \end{pmatrix}
\end{equation}

In the case of a linear frame bundle $L(M)$, (\ref{E:connection curvature}) leads to the usual definition of the curvature tensor on a Riemannian manifold $(M,g,\nabla)$ by pullback of local sections $\s_{U_i} : U_i \subseteq M \to L(M)$, given by
\begin{equation}\label{E:curvature tensor}
R(X,Y) Z := \nabla_X \nabla_Y Z - \nabla_Y \nabla_X Z - \nabla_{[X,Y]} Z
\end{equation}
for $X, Y, Z \in \mathfrak{X}(M)$, the collection of all $C^{\infty}$-vector fields. In particular, the curvature 2-from (\ref{E:connection curvature}) reduces to the Faraday tensor $F_{\m\n} = A_{\m,\n} - A_{\m,\n} $ for $G= U(1)$ an Abelian group,
\[
\O^\o_U = d\o_U :=\frac{1}{i} \, d (A_\m \, dx^\m ) =\frac{1}{i} \, dA_\m \wedge dx^\m = \frac{1}{i} \, ( A_{\n,\m} -  A_{\m,\n} ) dx^\m \wedge dx^\n
\]
where $\o_U := \s^*_U \o : = \frac{1}{i} \, A_\m \, dx^\m \in \Lambda(U,\mathfrak{u}(1) \cong i \mathbb{R})$ is a local 1-form on $U \subseteq M$ pulled back by a local section $\s_U: U\subset M \to L(M)$. As for the particle force where $G=SU(2)$ is non-Abelian, the field strength (\ref{E:connection curvature}) applies to see non-linear dependence on the gauge potential $\o_\m$ as in the case of PGT. We are ready for introducing torsion, which appears as another field strength in the PGT as we explain later.

\begin{definition}{(The \textbf{canonical 1-form}, the \textbf{soldering form})}\label{Def:canonical 1-form}

For a generalized observer $u\in L(M)$, define the \textbf{canonical
1-form} $\varphi \in \Lambda^1(L(M),\mathbb{R}^{1,3})$ by $\varphi(X_u) := u^{-1}
(\pi_*(X_u))$, where $X \in TL(M)$ is a vector field on $L(M)$ and $u^{-1}:T_xM \to \mathbb{R}^{1,3}$ is the inverse map of the linear isomorphism $u$. The
restricted canonical 1-form on $F(M)$ is defined similarly on
$F(M)$, which is denoted by the same notation $\varphi$.

\end{definition}

In order to describe the interaction between the external symmetry
of the spacetime and the internal symmetry of the gauge, a representation of $G$ is needed. Recall that a
representation of a group $G$ on a vector space $V$ is a group
homomorphism $\rho : G \to GL(V)$ such that $\rho(g h) = \rho(g) \cdot \rho(h)$ for all $g,h \in G$. With a representation of a Lie group $G$, a
fundamental object is defined:

\begin{definition}{(Basic differential forms)}\label{Def:basic differential form}

Let $V$ be a vector space and $\overline{\Lambda}^k(P,V)$ be a space
of \textbf{basic differential $k$-forms} on $P$ such that, for a
given representation $\rho: G \to GL(V)$, it is both
\textbf{$G$-invariant} and \textbf{horizontal}:
\begin{enumerate}
\item
($G$-invariant)  For $\a \in \overline{\Lambda}^k(P,V)$,
$X_1,\ldots, X_k \in T_pP$, $p\in P$, we have,
\begin{equation}\label{E:basic differential form}
\alpha (R_{g_*} X_1, \ldots, R_{g_*} X_k  ) = \rho(g^{-1}) \cdot
\alpha (X_1, \ldots, X_k) \quad \mbox{or} \quad R^*_g \a =
\rho(g^{-1}) \cdot \a,
\end{equation}

\item
(horizontal) If one of $X_1,\ldots, X_k$ is vertical, then $\alpha
(X_1, \ldots, X_k) =0$.
\end{enumerate}

\end{definition}
with the definitions above, immediately one finds a fact that the canonical
1-form $\varphi$ is a basic differential 1-form with respect to the representation $\rho: O(1,3) \to
GL(\mathbb{R}^{1,3})$. The exterior covariant derivative of basic differential forms are defined similarly:
\begin{definition}{(Exterior covariant derivative of $\overline{\Lambda}^k(P,V)$ ) }\label{Def:Exterior covariant derivative2}

Given a connection $\o$ on $P$, with respect to the adjoint representation $V=\mathfrak{g}$, $\rho: G \to GL(V)$ by $g \mapsto \rho(g)=\mathfrak{ad}_g$. We define
\begin{equation}
D^\o : \overline{\Lambda}^k(P,V) \to \overline{\Lambda}^{k+1}(P,V), \,\, \mbox{by} \quad D^\o \psi :=(d\psi)^H
\end{equation}
\end{definition}
one can easily see that $D^\o \psi$ satisfies (\ref{E:basic differential form}) since
\[
R^*_g D^\o \psi = R_g^* (d\psi)^H = ( R_g^* d\psi)^H = ( d R_g^* \psi)^H = \mathfrak{ad}_{g^{-1}} \cdot(d\psi)^H =  \mathfrak{ad}_{g^{-1}} \cdot D^\o \psi
\]
Similarly, the exterior covariant derivative of $\overline{\Lambda}^k(P,V)$ has another form.
\begin{theorem}\label{Thm:ext derivative}

For $\psi \in \overline{\Lambda}^k(P,V)$, we have $ D^\o \psi =  d\psi + \o \dot{\wedge} \psi$.
\end{theorem}

where the notation $\psi \in \overline{\Lambda}^{k}(P,V)$ is defined as
\[
\a \dot{\wedge} \psi (X_1, \ldots, X_j, X_{j+1}, \ldots, X_{j+k}) := \frac{1}{j! k!} \sum_{\s}  \rho_* \left( \a \left( X_{\s(1)}, \ldots , X_{\s(j)} \right) \right) \cdot  \a \left( X_{\s(j+1)}, \ldots , X_{\s(j+k)}\right)
\]
for $\a\in \overline{\Lambda}^{j}(P,\mathfrak{g})$ and $\s$ is a permutation of $\{ 1, \ldots, j+k \}$. Finally, we may present the notion of torsion stemming from the canonical 1-form $\varphi$ of the frame bundle $F(M)$.

\begin{definition}{(Torsion 2-form)}

The \textbf{torsion 2-form} of $\o$ is defined by $\Theta^\o := D^\o \varphi \in \overline{\Lambda}^2(F(M),\mathbb{R}^{1,3})$, where $\varphi $ is the canonical 1-form on $F(M)$ and $\o$ is a connection $\o$ given on $F(M)$.

Therefore by Theorem. (\ref{Thm:ext derivative}) one has
\begin{equation}\label{E:connection torsion}
\Theta^\o := D^\o \varphi = d\varphi + \o \dot{\wedge } \varphi
\end{equation}

\end{definition}
In particular, Riemannian geometry describing GR is of one special case

\begin{theorem}{(Levi-Civita connection)}

On the orthonormal frame bundle $F(M)$, there exists a unique
connection $\widetilde{\omega}\in \Lambda^1(F(M),\mathfrak{so}(1,3))$ such that
$\widetilde{\omega}$ has vanishing torsion, $D^{\widetilde{\omega}} \varphi =0 $, which is called the \textbf{Levi-Civita connection}.
\end{theorem}

The above bundle tools contain minimal geometric conception for us to understand PGT theory as a whole. From this viewpoint, GR can be reformulated on the frame bundle $\pi: F(M) \to M$. Below is a dictionary of terminology between general principal fibre bundles and gauge theories, see \cite{Bleecker}, \cite{Wu:1975es}, and \cite{Jost2}:

\begin{table}[h]
\centering
\begin{tabular}{c c c}
Principal fibre bundle theory & & Gauge theory \\
\hline \hline
total space & $P$ & space of phase factors \\
base space & $M$ & spacetime \\
structure group & $G$ & gauge group \\
local section & $\s$ & local gauge \\
connection 1-form & $\o$ & gauge potential \\
curvature 2-form & $\O$ & field strength \\
\hline \hline
\end{tabular}
\end{table}

With the bundle tools above, we can start to explore the Poincar\'{e} gauge theory of gravitation.



\section{Poincar\'{e} gauge theory on the affine frame bundle $\mathbb{A}(M)$}

PGT is, by definition, a gauge gravity theory of the Poincar\'{e} group $\mathcal{P} = \mathbb{R}^{1,3} \rtimes SO(1,3)$. Since we have seen that GR can be formulated by the frame bundle $(F(M),\pi,M,SO(1,3),\widetilde{\o})$, one may ask what is the suitable bundle description for PGT? That is we intend to find a principal bundle $(P, \pi ,M, G)$ with $G=\mathbb{R}^{1,3} \rtimes SO(1,3)$.

Of course, we may just assume that such a fibre bundle exists. However, we want to search for a natural construction like the frame bundle $L(M)$ in Definition (\ref{Def:frame bundle}). Recall that the Poincar\'{e} group $\mathcal{P}$ is the semidirect product of translations in Minkowski $\mathbb{R}^{1,3}$ and hyper-rotations $SO(1,3)$, defined by all rigid motion of the Minkowski spacetime $\mathbb{R}^{1,3}$. Analyzing the Poincar\'{e} group helps us to construct the PGT, \cite{TrautmanDG}.

\begin{definition}{(Affine space)}

If $E$ is an \textbf{affine space} with a vector space $V$ as the space of \emph{translations}. We require that $V$ acts \textbf{freely} and \textbf{transitively} on $E$. By this we mean:
\begin{itemize}
\item
(freely) for $p \in E$ and $v\in V$, if $v+ p = p $ $\Leftrightarrow$ $v=0$,

\item
(transitively) for all $p, q \in E$, there exists $v\in V$ such that $p+v=q$.

\end{itemize}

\end{definition}

If $E_1$, $E_2$ are two affine spaces, $V_1$ and $V_2$ are their groups of translations respectively, then the map $f:E_1 \to E_2$ is called an \textbf{affine map} if there exists a linear map $\b : V_1 \to V_2$ such that $f(v+p ) = \b_f (v) + f(p)$ for all $v\in V$, $p \in E$. In particular, if $E_1 =E_2 =E$ we define the \textbf{affine group} $GA(E)$ of $E$ by
\[
GA(E): = \{ f: E \to E | \,\, \mbox{$f$ is a bijective affine map} \}
\]
then we have an \emph{exact sequence} from the above,
\begin{equation}\label{E:exact seq}
\xymatrix{ 0 \ar[r]  & V \ar[r]^\a  & GA(E) \ar[r]^\b & GL(V) \ar[r] & 1}
\end{equation}\
in the sense that $Im \, \a = Ker \, \b$, which indicates that $GA(E)$ splits into a semi-direct product $GA(E) = V \rtimes GL(V)$. Thus if we take $V= \mathbb{R}^{1,3}$, one obtains $GA(\mathbb{R}^{1,3}) =  \mathbb{R}^{1,3}\rtimes GL(\mathbb{R}^{1,3}) \supset \mathcal{P}$, all rigid motion of the Minkowski space $\mathbb{R}^{1,3}$, including Lorentz group $SO(1,3)$, parity transformation, time reversal and translation.

Let $M$ be the spacetime. If we regard $\mathbb{R}^{1,3}$ as an affine space $\mathbb{A}^{1,3}$ and also the tangent space $T_x M $ at $x\in M$ as another affine space $A_x M$ (\textbf{tangent affine space}). Then elements in $A_x M $ are of the form $\bar{u} = (p, u(e_0), \ldots, u(e_3)) \in A_x M $, where $p\in A_xM$, $\{ e_0, \ldots, e_3 \} $ is the standard basis for $\mathbb{R}^{1,3}$, with $u: \mathbb{R}^{1,3} \to T_xM $ a linear isomorphism such that $\{ u(e_0), \ldots, u(e_3) \}$ forms a vector basis for $T_xM$. One can then identify an element $\bar{u} \in A_xM$ as an affine transformation $\widetilde{u}: \mathbb{A}^{1,3} \to A_xM$ by
\[
\widetilde{u} (0;e_0, \ldots, e_3) := (p, u(e_0), \ldots, u(e_3)) \in A_xM , \qquad \left( (0;e_0, \ldots, e_3) \in \mathbb{A}^{1,3} \right)
\]
Denote $\mathbb{A}_x M:= \{  \widetilde{u}= (p,u )  \, | \,\, \widetilde{u}: \mathbb{A}^{1,3}  \to A_x M , \,\, \mbox{an affine transformation} \} $. Then there is a 1-1 correspondence of $\bar{u} \in A_x M \,\, \overset{1-1}{\longleftrightarrow} \,\, \widetilde{u} \in \mathbb{A}_x M$ such that we can identity $A_x M  \cong \mathbb{A}_x M$. As a comparison,
\[
\begin{aligned}
\mbox{a linear frame $u\in F_xM$} \quad &\Leftrightarrow \quad u:\mathbb{R}^{1,3} \to T_x M \quad \mbox{( $u$ linear isomorphism)} \\
\mbox{an affine frame $\widetilde{u}\in \mathbb{A}_x M $} \quad &\Leftrightarrow \quad \widetilde{u}:\mathbb{A}^{1,3} \to A_x M \quad \mbox{( $\widetilde{u}$ affine transformation)}
\end{aligned}
\]
If we define the set $\mathbb{A}(M):= \bigcup_{x\in M} \mathbb{A}_x M$ and $R_g: \mathbb{A}(M) \to \mathbb{A}(M) $ the right action of the gauge group $ GA(\mathbb{R}^{1,3})= \{ g= (A, \xi)| \,\, A \in GL(\mathbb{R}^{1,3}) , \,\, \xi \in \mathbb{R}^{1,3} \} $ on $\mathbb{A}(M)$ by $R_{(A,\xi) } (p,u) := (p + u \cdot \xi, u \circ A)$, where $(p,u) \in \mathbb{A}_x M $ and $ (A,\xi) \in GA(\mathbb{R}^{1,3})$, and define the projection map $\widetilde{\pi} : \mathbb{A}(M)\to M $ by $\widetilde{\pi} (p,u) := x$ for all affine frames $(p,u)$ at $x \in M$. Then with proper differential structure on $\mathbb{A}(M)$, one can prove that $(\mathbb{A}(M), \widetilde{\pi} , M, GA(\mathbb{R}^{1,3}) )$ forms a principal fibre bundle called the \textbf{affine frame bundle}, \cite{McInnes:1984sm}, \cite{Kobayashi}. Such is the living space of the PGT. Since we want to study the relationship between the PGT and GR, we should investigate the connection between $\mathbb{A}(M)$ and $L(M)$. Recall that we have the exact sequence (\ref{E:exact seq}) where
\[
\begin{aligned}
\a : \mathbb{R}^{1,3} \to GA(\mathbb{R}^{1,3}) \quad  &\mbox{by}  \quad \a(\xi) := \begin{pmatrix}
                                                                                    I_{4\times 4} & \xi \\
                                                                                    0             & 1
                                                                                   \end{pmatrix} \\
 \b : GA(\mathbb{R}^{1,3}) \to GL(\mathbb{R}^{1,3}) \quad  &\mbox{by}  \quad \b\left( \begin{pmatrix}
                                                                                    A & \xi \\
                                                                                    0             & 1
                                                                                   \end{pmatrix} \right) := A           \\
\g : GL(\mathbb{R}^{1,3}) \to GA(\mathbb{R}^{1,3}) \quad  &\mbox{by}  \quad \g(A) := \begin{pmatrix}
                                                                                    A & 0 \\
                                                                                    0             & 1
                                                                                   \end{pmatrix}
                                                                                   \end{aligned}
\]
and thus $\b \circ \g = Id$ on $GL(\mathbb{R}^{1,3})$. Corresponding to the homomorphisms $\a$, $\b$, $\g$ , we define a natural projection $\b : \mathbb{A}(M) \to L(M)$ by $\b (p,u) := u $ and $\g : L(M) \to \mathbb{A}(M) $ by $\g (u) := (0,u) $ so that $\b \circ \g =Id $ on $L(M)$. A connection $\widetilde{\o}$ defined on the affine frame bundle $\widetilde{\pi} : \mathbb{A}(M) \to M$ is called a \textbf{generalized affine connection}. Taking $V = \mathbb{R}^{1,3}$ in the exact sequence (\ref{E:exact seq}) results in the splitting exact sequence of Lie algebras:
\begin{equation}\label{E:exact seq 2}
\xymatrix{ 0 \ar[r]  & \mathbb{R}^{1,3} \ar[r]  & \mathfrak{ga}(\mathbb{R}^{1,3}) \ar[r] & \mathfrak{gl}(\mathbb{R}^{1,3}) \ar[r] & 1}
\end{equation}
so that we obtain a decomposition of the Lie algebra $\mathfrak{ga}(\mathbb{R}^{1,3}) = \mathfrak{gl}(\mathbb{R}^{1,3}) \oplus \mathbb{R}^{1,3}$.

Thus for a generalized affine connection $\widetilde{\o}$, the pull-back $\g^* \widetilde{\o} \in \Lambda^1(L(M),\mathfrak{ga}(\mathbb{R}^{1,3}))$ can be decomposed into two 1-forms
\begin{equation}\label{E:GAC split}
\g^* \widetilde{\o} = \omega \oplus \psi
\end{equation}
according to their Lie-algebra value such that $\o \in \Lambda^1(L(M),\mathfrak{gl}(\mathbb{R}^{1,3}))$ and $\psi \in \overline{\Lambda}^1(L(M),\mathbb{R}^{1,3})$. In fact, by a theorem in \cite{Kobayashi}, it turns out that $\o$ defines a connection on $L(M)$. While from another theorem of McInnes, there is a $1:1$ correspondence between an affine connection and a linear connection
\begin{equation}
\widetilde{\o} \overset{1-1}{\longleftrightarrow} (\o,\psi).
\end{equation}
In summary, one has
\begin{theorem}{(Kobayashi\& Nomizu, \cite{Kobayashi})}

Let $\widetilde{\o}$ be a generalized affine connection on $(\mathbb{A}(M), \widetilde{\pi}, M, GA(\mathbb{R}^{1,3}))$ with the decomposition (\ref{E:GAC split}), then the affine curvature 2-form $\widetilde{\O}^\o \in \overline{\Lambda}(\mathbb{A}(M), \mathfrak{ga}(\mathbb{R}^{1,3}))$ has the decomposition
\begin{equation}\label{E:GAC curvature}
\g^* \widetilde{\O}^{\widetilde{\o}} = \O^\o + D^\o \psi
\end{equation}
where the covariant derivative $D^\o$ is defined by $\o $ on $L(M)$.
\end{theorem}

In particular, if a generalized affine connection $\widetilde{\o}$ coincidentally has the form
\begin{equation}\label{E:AC split}
\g^* \widetilde{\o} = \omega \oplus \varphi
\end{equation}
such that the $\mathbb{R}^{1,3}$ Lie-algebra valued component $\varphi$ is the canonical 1-form in Definition.~\ref{Def:canonical 1-form}, then we call such a connection $\widetilde{\o}$ an \textbf{affine connection} (c.f. generalized affine connection). Thus by (\ref{E:GAC curvature}), we see that an affine connection $\widetilde{\o}$ on $\mathbb{A}(M)$ yields
\begin{equation}\label{E:AC curvature}
\g^* \widetilde{\O}^\o = \O^\o + \Theta^\o
\end{equation}
\[
\mbox{\textit{ affine curvature $=$ spacetime curvature $+$ torsion} }
\]
where $\Theta^\o $ is the torsion form on $L(M)$ as given by $\Theta^\o: = D^\o \varphi$ in (\ref{E:connection torsion}). This shows that on the affine frame bundle $(\mathbb{A}(M), \widetilde{\pi}, M,  GA(\mathbb{R}^{1,3}) , \widetilde{\o} ) $ where $\widetilde{\o}$ is an affine connection, the curvature and torsion of spacetime on $M$ are united as one affine curvature $\widetilde{\O}^\o$. In other words, the affine curvature $\widetilde{\O}^\o$ splits into curvature $\O$ and torsion $\Theta^\o$ on the (base) spacetime $M$. This is reminiscent of a quotation in \cite{McInnes:1984sm}:

\textit{``This removes the objection that torsion is a feature of space-time structure which has no analogue in gauge theory.''}

To understand the above statement more clearly, we remark that, for example, one can define the curvature 2-form $\O^\o$ of the $U(1)$-bundle $(P,\pi, M , U(1),\o = \frac{1}{i} \, A_\m \, dx^\m)$ for Maxwell's electrodynamics, called the electromagnetic field strength $F = d\o$ . But one cannot define the \emph{torsion} of the electromagnetic potential $A_{\mu}$. Similarly, one cannot speak of the torsion of the $SU(n)$-gauge field. What makes gravity special is exactly the existence of the canonical 1-form $\varphi$ in $L(M)$, Definition \ref{Def:canonical 1-form}. However, from the viewpoint of bundles torsion is nothing more than a byproduct of the specific group $\mathcal{P}$ used.

Here, one can raise a na\"{\i}ve question as a summary for this section:

\textbf{Why is Poincar\'{e} gauge gravity theory (PGT) naturally related to gravity with torsion?}

The answer is the following:

Since the gauging of the Poincar\'{e} group $\mathcal{P}$ into gravity theory necessarily leads to the affine frame bundle $(\mathbb{A}(M), \widetilde{\pi}, M,  GA(\mathbb{R}^{1,3}) , \widetilde{\o} ) $, and an affine connection $\widetilde{\o}$ on $\mathbb{A}(M)$ is pulled back naturally on the frame bundle $L(M)$ by $\g$ such that it results in the splitting (\ref{E:AC split}) according to the Lie-algebra, and thus the affine curvature splits accordingly by (\ref{E:AC curvature}), hence the torsion and the curvature on the spacetime.

In fact, if we require the PGT to be the bundle defined by $(\mathcal{P}(M),\widetilde{\pi}, M, \mathcal{P}, \widetilde{\o})$ with structure group, called the Poincar\'{e} bundle as $\mathcal{P}$, then we will need to reduce from $(\mathbb{A}(M), \widetilde{\pi}, M, GA(\mathbb{R}^{1,3}))$ to the Poincar\'{e} bundle. However such a reduction is a delicate issue, which is beyond our scope, see \cite{McInnes:1984sm}. Thus in the following discussion, we assume the Poincar\'{e} bundle as $\mathcal{P}$ exists as a sub-bundle of $(\mathbb{A}(M), \widetilde{\pi}, M, GA(\mathbb{R}^{1,3}))$

The above bundle construction clearly describes the gauge scenario for PGT. However, when we project from the total bundle space $\mathcal{P}(M) $ to the spacetime $M$, the resultant field strength we observe on $M$ is simply curvature and torsion. Thus for most PGT theorists, their actual computation does not involve the bundle formulation, except that it provides a clear account of gauging gravity. Rather, one can only regard the projected version of the spacetime with torsion and curvature, namely $(M, g, \nabla)$ with $\nabla $ a general connection compatible to $g$ on $M$, and such a local version on $M$ is what we introduce next.


\section{Riemann-Cartan geometry}

Given a semi-Riemannian spacetime $(M,g)$, GR is formulated on the tuple $(M,g ,\widetilde{\nabla} )$ where $\widetilde{\nabla} $ is the Levi-Civita connection. If one considers a general connection $\nabla$, then $(M,g,\nabla )$ is called a \textbf{metric-affine spacetime}, which can be regarded as the most general spacetime under the framework of differentiable geometry. Here a general connection $\nabla: \mathfrak{X}(M) \times \mathfrak{X}(M) \to \mathfrak{X}(M)$  on $M$ is an operator such that \cite{Jost1},
\begin{equation}
\begin{aligned}
\nabla_{f X +g Y} Z &= f \nabla_{X} Z + g \nabla_{ Y} Z  \\
\nabla_{X} (Y+Z) &=  \nabla_{X} Y + \nabla_{X} Z\\
\nabla_{X} (fY) &= f \nabla_{X} Y + X(f ) Y
\end{aligned}
\end{equation}
where $X, Y, Z \in  \mathfrak{X}(M)$, $f, g \in C^{\infty}(M)$. With a connection, one can then define the torsion tensor and the curvature tensor by \cite{Jost1},
\begin{definition}{(Curvature tensor and torsion tensor)}

The \textbf{Riemann curvature tensor} $R(\cdot, \cdot): \mathfrak{X}(M)  \times \mathfrak{X}(M) \times \mathfrak{X}(M)  \to  \mathfrak{X}(M) $ and the \textbf{torsion tensor} $T(\cdot, \cdot):\mathfrak{X}(M) \times \mathfrak{X}(M)  \to  \mathfrak{X}(M) $ of $\nabla$ are defined by
\begin{equation}\label{E:R and T}
\begin{aligned}
R(X,Y) Z &= \nabla_X  \nabla_Y Z - \nabla_Y  \nabla_X Z - \nabla_{[X,Y]} Z \\
T(X,Y) Z &= \nabla_X  Y - \nabla_Y  X - [X,Y]
\end{aligned}
\end{equation}
for all $X,Y,Z \in \mathfrak{X}(M) $.
\end{definition}
Notice that on a metric-affine spacetime $(M,g,\nabla)$ one does not have to place the condition of \textbf{metric compatibility} (or \textbf{metricity}) constrained by
\begin{equation}\label{E:metricity}
d\left( g(X,Y) \right) = g\left( \nabla X, Y \right) + g\left( X, \nabla  Y \right)
\end{equation}
which is equivalent to a usual coordinate form $g_{\a\b,\m} = \G_{\m\a}{}{}^\n \, g_{\n\b} + \G_{\m\b}{}{}^\n \, g_{\n\a} $ by taking coordinate vector fields $X= \frac{\partial}{\partial x^\a}$, $ Y = \frac{\partial}{\partial x^\b }$, and $Z = \frac{\partial}{\partial x^\m}$ in (\ref{E:metricity}).

\begin{definition}{(Riemann-Cartan spacetime, underlying space of PGT)}\label{Def:RC space}

The tuple $(M,g,\nabla)$ with a metric compatible (\ref{E:metricity}) connection $\nabla$ on $M$ is called the \textbf{Riemann-Cartan spacetime}, which is the underlying space of PGT, \cite{Trautman:2006fp}.
\end{definition}

The Definition (\ref{Def:RC space}) of PGT is equivalent to the version of bundle formulation. That is $(M,g, \nabla)$ $\Leftrightarrow$ $(\mathcal{P}(M), \widetilde{\pi}, M, \mathcal{P},\widetilde{\o})$. However, this requires some steps which we now explain.

Starting from the Poincar\'{e} bundle $(\mathcal{P}(M), \widetilde{\pi}, M, \mathcal{P},\widetilde{\o})$, to obtain curvature and torsion defined on $M$, one requires a local section.

Let $U \subseteq M$ be an open neighborhood in $M$ and $\s_U: U \to F(M)$ be a local (smooth) section on $U$, a $C^{\infty}$-map such that $\pi \circ \s_U = Id_U$. With the given local section $\s_U$ on $U \subseteq M$, equivalent to choosing an observer, we derive a locally defined curvature 2-form on $U$ from the Poincar\'{e} bundle $(\mathcal{P}(M), \widetilde{\pi}, M, \mathcal{P},\widetilde{\o})$ via
\begin{equation}\label{E:pull back R and T}
\begin{aligned}
\s_U^* \Omega^\o &\overset{(\ref{E:connection curvature})}{=} \s_U^* \left( d\o + \o \dot{\wedge} \o \right)  =  d (\s_U^* \o)  + (\s_U^*\o ) \dot{\wedge} (\s_U^*\o ) \\
\s_U^* \Theta^\o &\overset{(\ref{E:connection torsion})}{=} \s_U^* \left( d\varphi + \o \dot{\wedge} \varphi \right) = d( \s_U^*\varphi) + (\s_U^*\o) \dot{\wedge} (\s_U^* \varphi)
\end{aligned}
\end{equation}
if we denote $\O_U:= \s_U^* \O^\o \in \Lambda^2(U, \mathfrak{so}(1,3)) $, $T_U:= \s_U^* \Theta^\o \in \Lambda^2(U, \mathbb{R}^{1,3}) $, $\o_U:= \s_U^* \o \in \Lambda^1(U, \mathfrak{so}(1,3))$, and $\vt_U:= \s_U^* \varphi \in \Lambda^1(U, \mathbb{R}^{1,3})$, then we can rewrite the above as
\begin{equation}\label{E:local R and T}
\begin{aligned}
\O_U &= d\o_U + \o_U \wedge \o_U \,\, \in \Lambda^2(U, \mathfrak{so}(1,3))\\
T_U  &= d\vt_U + \o_U \wedge \vt_U \,\, \in \Lambda^2(U, \mathbb{R}^{1,3})
\end{aligned}
\end{equation}
so that we return to the case in (\ref{E:Lie alg exp}), where the Lie group is a matrix group. If we omit the sub-index $U$ for simplicity without confusion and write
\begin{equation}\label{E:Lie basis exp}
\begin{aligned}
\O &= \O^{\b\a} \, E_{\a\b} ,\quad  T = T^\a \, \widetilde{e}_\a, \\
\o &= \o^{\a\b} \, E_{\a\b} ,\quad \vt = \vt^\a{} \, \widetilde{e}_\a,
 \end{aligned}
\end{equation}
where $\{\widetilde{e}_\a\in \mathbb{R}^{1,3}\}$ are the standard basis of the Minkowski spacetime and $\{ E_{\a\b} \in \mathfrak{so}(1,3) \}$ is a basis of the Lie algebra $\mathfrak{so}(1,3)$ with $\a, \b, \ldots = 0,1,2,3 $. Notice that since the Lie-algebra basis $E_{\a\b} = -E_{\b\a}$ is anti-symmetric so that $\O^{\b\a} = - \O^{\a\b}$ and $\o^{\b\a} = - \o^{\a\b}$ is anti-symmetric in the indices $\a$ and $\b$. One also remarks that the indices of $\O_{\a\b}$, $\o_\a{}^\b$, $T^\b$, \ldots , etc. have nothing to do with the spacetime index or coordinate. i.e, one cannot say $T^\a $ is a 1-tensor or $\o_\a{}^\b$ is a 2-tensor simply by judging from the number of indices.

With the basis expansion (\ref{E:Lie basis exp}), one may further rewrite (\ref{E:local R and T}) with (\ref{E:Lie alg exp}) in a usual form recognizable in several references \cite{Hehl:1994ue}, \cite{Trautman:2006fp}
\begin{equation}\label{E:T and R}
\begin{aligned}
\O^\b{}_\a &= \nabla \, \o_\a{}^\b = d \o_\a{}^\b + \o_\g{}^\b \wedge \o_\a{}^\g =
\frac{1}{2} R^\b{}_{\a\mu\nu} \, \vt^{\mu} \wedge \vt^{\nu} ,\\
T^\a &= \nabla \, \vt^\a = d\vt^\a + \o_\b{}^\a \wedge \vt^\b =
\frac{1}{2} T_{\mu\nu}{}{}^\a \, \vt^{\mu} \wedge \vt^{\nu},
\end{aligned}
\end{equation}
both the last equalities may be verified via the relationships:
\begin{equation}
\begin{aligned}
\O^\b{}_\a (X,Y) &= \vt^\b \left( R(X,Y) e_\a \right), \qquad T^\a(X,Y) = \vt^\a \left( T( X,Y )\right) ,\quad \mbox{or}\\
\O_{\b\a} (X,Y) &= g\left( R(X,Y) e_\a , e_\b \right), \qquad T_\a(X,Y) =  g( \left( T( X,Y ) , e_\a \right)
\end{aligned}
\end{equation}
where $e_\a = \s(x)(\widetilde{e}_\a) \in \mathfrak{X}(M)$ is an orthonormal basis (tetrad) induced by the section $\s$, $X, Y \in \mathfrak{X}(M)$ and $R(e_{\mu}, e_\n)e_\g := R^\s{}_{\g\m\n} \, e_\s $, $T(e_\m,e_\n) := T_{\m\n}{}{}^\s \, e_\s$ are defined by (\ref{E:R and T}).
One can also verify the following useful identities by direct computations, \cite{Kobayashi}:
\begin{equation}
\nabla T^\a = \O^\a{}_\b \wedge \vt^\b  , \quad
\mbox{($1^{st}$ Bianchi)}, \qquad \nabla \O^{\a\b} = 0, \quad \mbox{($2^{nd}$ Bianchi)}
\end{equation}

Also the requirement of the metricity condition (\ref{E:metricity}) in a tetrad $e_\a$ is written as
\begin{equation}\label{E:metricity2}
 0 = \o_\a{}^\m \, \eta_{\m\b} + \o_\b{}^{\m} \, \eta_{\m\a}, \qquad \o_{\a\b} = - \o_{\b\a}
\end{equation}
where $\eta_{\a\b}:= g(e_\a , e_\b)$ and in the last line we denote $\o_{\a\b} := \o_\a{}^\m \, \eta_{\m\b}$ for simplicity. Here we remark that the connection 1-form $\o_\a{}^\b \in \Lambda^1(M, \mathbb{R})$ is as equivalently defined by $\nabla e_\a := \o_\a{}^\b \, e_\b$.

In summary, if we pullback from $(\mathbb{A}(M), \widetilde{\pi}, GA(\mathbb{R}^{1,3}), \widetilde{\o})$, the underlying space of PGT, to the spacetime $M$, the resultant elements are $(M,g,\nabla, \vt^\a )$ with $\varphi \Leftrightarrow  \vt^\a $ a tetrad as the $\mathbb{R}^{1,3}$-gauge potential for translation and $\nabla  \Leftrightarrow \o_\a{}^\b $ as the $\mathfrak{so}(1,3)$-gauge potential for rotations satisfying (\ref{E:metricity2}), called Riemann-Cartan spacetime,\cite{Hehl:1976kj}. The two field strength torsion and curvature $(\O^\a{}_\b, T^\b)$ of $\nabla$ are given by (\ref{E:T and R}). Below we write the geometric correspondence of the PGT in gauge theoretic language (c.f. \cite{Hehl:1994ue}, \cite{Trautman:2006fp}):

\begin{table}[h]
\centering
\begin{tabular}{c | c  c}
PGT & translational $\mathbb{R}^{1,3}$ &  rotational $SO(1,3)$ \\
\hline \hline
gauge potentials (on $L(M)$) & $\varphi$                      & $\o$ \\
gauge potentials (on $M$)             & $\vt^\a$                       & $\o_\a{}^\b$ \\
field strengths (on $L(M)$)  & $\Theta^\o =D^\o \varphi$       & $\O^\o = D^\o \o$ \\
field strengths (on $M$)              &  $T^\a = \nabla \vt^\a $             & $\O^\b{}_\a = \nabla  \o_\a{}^\b $ \\
\end{tabular}
\end{table}




\subsection{Changes of frames}\label{Sec:chaning frame}

It is important to study changes of frames under different observers in spacetime. Some gauge quantities may change correspondingly under the transition and here we wish to give a clear view.

In the definition of a principal fibre bundle $P$, it is defined to have a local trivialization on $T_U: \pi^{-1}(U) \subseteq P \to U \times G $ such that $T_U(p)=(\pi(p), \psi_U(p))$ with $ \psi_U:\pi^{-1}(U) \to G$ satisfying $ \psi_U(p g) =  \psi_U(p) \cdot g$ for all $p\in\pi^{-1}(U)$, $g\in G$.

One can then define the transition function on the overlap
\begin{definition}{(Transition function)}\label{Def:transition function}

Let $U$, $V$ be two open sets in $M$ such that $U \cap V \neq \phi$, and let $T_U : \pi^{-1}(U)  \to U \times G$ and $T_V : \pi^{-1}(V)  \to V \times G$ be two local trivializations of a principal bundle $(P,\pi,M,G)$. Then the \textbf{transition function} from $T_U$ to $T_V$ is defined as the map $\Psi_{UV} : U \cap V \to P$ by
\begin{equation}
\Psi_{UV}(x) := \psi_U(p) \, \psi_V(p)^{-1},
\end{equation}
where $p\in P$, $x\in M$ such that $\pi(p)=x$ and $\psi_V(p)^{-1}$ denotes the group inverse, not the functional inverse.
\end{definition}
The definition of the transition function is well-defined since if $p,q\in \pi^{-1}(U)$ ($ p =q g$ for some $g \in G$), then $\psi_U(qg) \, \psi_V(qg)^{-1} = \psi_U(q) \, \psi_V(q)^{-1} $. For the transition function it is easy to verify that it possesses the following properties
\begin{proposition}
The transition function $\Psi_{UV}$ satisfies,
\begin{equation}
\begin{aligned}
\Psi_{UU}(y) &= e, \qquad \qquad \quad ( \forall y\in U)\\
\Psi_{UV}(y) &= \Psi_{VU}^{-1}(y), \qquad   (\forall y\in U\cap V) \\
\Psi_{UV}(y) \cdot \Psi_{VW}(y) \cdot \Psi_{WU}(y) &= e   \qquad \qquad\quad  (\forall y\in U\cap V \cap W)
\end{aligned}
\end{equation}
\end{proposition}

The change of the two frames is equivalent to the transition between two sections.
\begin{proposition}

Let $T_U : \pi^{-1}(U)  \to U \times G$ and $T_V : \pi^{-1}(V)  \to V \times G$ be two local trivializations of $P$ as given above, one can define a section $\s_U: U \to P$ associated to the trivialization $T_U$ given by $\s_U(x) := T_U^{-1}(x,e)$, along with $\s_V: V \to P$ similarly defined. For $Y \in T_xM$, we have
\begin{equation}\label{E:change section}
\s_{V*} (Y_x) = \left[ L_{\Psi^{-1}_{UV}(x)*} \left( \Psi_{UV*}(Y_x)  \right) \right]^*_{\s_V(x)} + R_{ \Psi_{UV}(x) *} \circ \s_{U*}(Y_x)
\end{equation}

\end{proposition}

\begin{proof}
First we notice that since $\s_U(x) = T_U^{-1}(x,e)$ we have $T_U(\s_U(x) \, g) = (x,g) $, in particular $T_U(\s_U(x) ) = (x,e)  = (\pi(\s_U(x)), \psi_U(\s_U(x)))$. Thus we find $\psi_U(\s_U(x)) = e$ for all $x\in \pi^{-1}(U)$.
Now let $\g: \mathbb{R} \to M $ be a curve such that $\g(0) =x$ and $\g'(0) =Y$, then
\begin{equation}
\begin{aligned}
\s_{V*}(Y_x) &:= \frac{d}{dt}\Bigr|_{t=0} \, \s_{V} (\g(t)) = \frac{d}{dt} \, \left[ \s_U (\g(t)) \, \Psi_{UV}(\g(t)) \right]\\
             &= \frac{d}{dt}\Bigr|_{t=0} \, \left[ \s_U (x) \, \Psi_{UV}(\g(t))  \right] +  \frac{d}{dt}\Bigr|_{t=0} \, \left[ \s_U (\g(t)) \, \Psi_{UV}(x) \right]\\
             &= \frac{d}{dt}\Bigr|_{t=0} \, \left[ \s_V (x) \, \Psi^{-1}_{UV}(x) \, \Psi_{UV}(\g(t))  \right] +  R_{\Psi_{UV}(x)*} \s_{Ux}(Y_x)
\end{aligned}
\end{equation}
where we have used $\Psi_{VU} (x)= \Psi^{-1}_{UV}(x) = \psi_V(x) \cdot \psi_U^{-1}(x)$. If we define a curve $\a$ in $G$ by $\a(t):= \Psi^{-1}_{UV}(x) \, \Psi_{UV}(\g(t))$, we find $\a(0) = e$ and
\[
\a'(0) := \frac{d}{dt}\Bigr|_{t=0} \left( \Psi^{-1}_{UV}(x) \, \Psi_{UV}(\g(t)) \right) =  L_{\Psi^{-1}_{UV}(x)*} \left( \Psi_{UV*}(Y_x) \right) \in \mathfrak{g}
\]
Then we may write $\a(t) = e^{tA} $ with $A:= L_{\Psi^{-1}_{UV}(x)*} \left( \Psi_{UV*}(Y_x) \right)\in \mathfrak{g} $ and the first term in the last equality of (\ref{E:change section}) becomes
\[
\frac{d}{dt}\Bigr|_{t=0} \, \left( \s_V (x) \, \Psi^{-1}_{UV}(x) \, \Psi_{UV}(\g(t))  \right) = \frac{d}{dt}\Bigr|_{t=0}  \left( \s_V (x) \cdot e^{tA} \right) : = \left[ L_{\Psi^{-1}_{UV}(x)*} \left( \Psi_{UV*}(Y_x)  \right) \right]^*_{\s_V(x)}
\]
where the last equality is due to (\ref{E:fundamental vector field}).
\end{proof}
With the proposition above, we have
\begin{corollary}\label{Cor:LM change section}

On the frame bundle $(L(M),\pi,M, GL(\mathbb{R}^{1,3}),\o)$, let $\s_U: U \to L(M)$, $\s_V: V \to L(M)$ be two local sections corresponding to the local trivializations $T_U$ and $T_V$ respectively. For the canonical 1-form $\varphi \in \overline{\Lambda}^1(L(M), \mathbb{R}^{1,3})$ and torsion 2-form $\Theta^\o \in \overline{\Lambda}^2(L(M), \mathbb{R}^{1,3})$ with canonical representation $\rho : GL(\mathbb{R}^{1,3}) \to GL(4,\mathbb{R})$ and $Y_x \in T_xM$ one has following local forms $\s_U^*\varphi \in \Lambda^1(U, \mathbb{R}^{1,3}) $, $\s_U^*\Theta^\o \in \Lambda^2(U, \mathbb{R}^{1,3})$, $\s_V^*\varphi \in \Lambda^1(V, \mathbb{R}^{1,3}) $, and $\s_V^*\Theta^\o \in \Lambda^2(V, \mathbb{R}^{1,3})$ and they are related by
\begin{equation}\label{E:LM change section}
\begin{aligned}
(\s_V^*\varphi )(Y_x) = \varphi (\s_{V*} Y_x)  &= \left(\Psi_{UV}(x)\right)^{-1} \cdot (\s_U^* \varphi) (Y_x)\\
(\s_V^*\Theta^\o )(Y_x) = \varphi (\s_{V*} Y_x)  &= \left(\Psi_{UV}(x)\right)^{-1} \cdot (\s_U^* \Theta^\o) (Y_x)
\end{aligned}
\end{equation}
\begin{proof}
By applying the canonical 1-form $\varphi$ directly on (\ref{E:change section})
\[
(\s_V^*\varphi )(Y_x) = \varphi (\s_{V*} Y_x) = \varphi \left( R_{\Psi_{UV}(x)*} \s_{Ux*}(Y_x) \right) = \left(\Psi_{UV}(x)\right)^{-1} \cdot (\s_U^* \varphi) (Y_x)
\]
where the two equalities are true because $\varphi$ is a basic differential 1-form in Def. (\ref{Def:basic differential form}) and similarly for the torsion 2-form.
\end{proof}

\end{corollary}

We explain why the local form $\s^* \varphi \in \Lambda^1(U, \mathbb{R}^{1,3})$ is important, simply because it gives a coframe $\vt^\a \in TM$ on $M$.
\begin{remark}{(Coframe on $M$)}

Let $\s : U\subseteq M \to L(M)$ be a local section on $U$ and $(E_\a ) \in \mathbb{R}^{1,3}$, $\a = 0,1,2,3$, be the standard basis. Since $\s(x) \in L(M)$ for $x\in U$, denoting $\s(x)(E_\a) := e_\a(x) \in T_xU$, then $(e_\a) \in \mathfrak{X}(U)$ form a basis on $U$ and let the coframe $(\vt^\a \in T^*U)$ dual to $(e_\a)$ be such that $\vt^\a (e_\b) := \d^\a_\b$. Therefore
\begin{equation}
(\s^* \varphi) (e_\a ) := \varphi(\s_* e_\a) := (\s(x))^{-1} \cdot \pi_*(\s_* \, e_\a) = (\s(x))^{-1} (e_\a) = E_\a = (\vt^\b \otimes E_\b)  (e_\a)
\end{equation}
and hence
\begin{equation}\label{E:local canonical form}
\s^* \varphi = \vt^\b \otimes E_\b, \qquad (\s^* \varphi)^\b = \vt^\b
\end{equation}
This indicates that the pull back of the canonical 1-form $\varphi$ gives a coframe $\vt^\a $ corresponding to the given section $\s:U\to L(M)$.
\end{remark}

By the remark, Corollary \ref{Cor:LM change section} then implies the transition of the two frames on $M$.
\begin{remark}{(Local transition)}\label{Rmk:local transition}

Let $(e_\a) \in \mathfrak{X}(U) $ and $(\vt^\a := (\s_U^* \varphi)^\a) \in \Lambda^1(U)$ be a set of frame and coframe on $U \subset M $ as defined above by a local section $\s_U : U \to L(M)$ and $(\widetilde{e}_\a) \in \mathfrak{X}(V) $ and $(\widetilde{\vt}^\a := (\s_V^* \varphi)^\a) \in \Lambda^1(V)$ be another set of frame and coframe on $V \subset M $ similarly defined by a section $\s_V : V \to L(M)$. Also denote the transition function $\Psi_{UV}: U\cap V \to GL(1,3)$ in a matrix form $A^\b_\a (x):= (\Psi_{UV}(x))^\b_\a$ for $x\in U \cap V$, then (\ref{E:LM change section}) states that if $\widetilde{e}_\a(x) = A_\a^\b(x) \, e_\b$, then the local tetrad and torsion 2-form follow the transformation
\begin{equation}\label{E:frame change}
\widetilde{\vt}^\b(x) = (A^{-1})_\a^\b(x) \, \vt^\a (x) , \qquad \widetilde{T}^\b(x) = (A^{-1})_\a^\b(x) \, T^\a (x) ,
\end{equation}
where $T^\a \in \Lambda^2(U,\mathbb{R})$ are as defined in (\ref{E:pull back R and T}) and (\ref{E:local R and T}).
\end{remark}

In fact, by (\ref{E:change section}) we also obtain a transformation law for a connection 1-form $\o$ and curvature 2-form $\O^\o$ on a general principal bundle $(P,\pi, M,G,\o)$.

\begin{proposition}\label{Prop:connection transform}

With the same notations above and define $\o_U:= \s_U^* \o \in \Lambda^1(U, \mathfrak{g})$, $\o_V:= \s_V^* \o \in \Lambda^1(V,\mathfrak{g})$, we have the transformation behavior between two local trivializations
\begin{equation}\label{E:connection transition}
\o_V (Y_x) = L^{-1}_{\Psi_{UV}(x)*} \left( \Psi_{UV*}(Y_x) \right) + \mathfrak{ad}_{\Psi_{UV}^{-1}(x)} (\o_U(Y_x)), \qquad \forall Y_x \in T_xM
\end{equation}
in particular, when the Lie group $G$ is a matrix group we obtain a familiar form
\begin{equation}\label{E:connection transition2}
\o_V =  \Psi_{UV}^{-1} \cdot d  \Psi_{UV} +  \Psi_{UV}^{-1} \cdot \o_U  \cdot \Psi_{UV}
\end{equation}
therefore one says that ``a connection does not transform like a tensor" due to the extra piece $\Psi_{UV}^{-1} \cdot d  \Psi_{UV}$.

The transformation of the local curvature 2-forms $\O^\o_U:= \s_U^*\O^\o \in \overline{\Lambda}^2(U, \mathfrak{g})$ and $\O^\o_V:= \s_V^*\O^\o \in \overline{\Lambda}^2(V, \mathfrak{g})$ is
\begin{equation}\label{E:curvature transition}
\O^\o_V (x) =  \mathfrak{ad}_{\Psi_{UV}^{-1}(x)} (\O^\o_U(x)) , \O^\o_V  =   \Psi_{UV}^{-1} \cdot \O^\o_U  \cdot \Psi_{UV} ,\qquad (\mbox{if $G$ is a matrix group})
\end{equation}

\end{proposition}

One remarks that Proposition (\ref{Prop:connection transform}) is true for all principal bundles, while (\ref{Rmk:local transition}) is only defined on the frame bundle $L(M)$ over $M$.


\subsection{Computational Aspects}\label{Sec:Computational Aspects}

So far we only setup the essential equipment for the PGT, but we have not yet specified the gravitational dynamics to follow, a gravitational Lagrangian. Thus a specific gravity Lagrangian shall give us the unique evolution of the PGT spacetime. Before we turn to imposing meaningful Lagrangians on $(M,g,\nabla)$, first we demonstrate some computational skills.

Given a RC-spacetime $(M,g,\nabla)$ with a coframe $\vt^\a$ such that the connection 1-form $\nabla e_\a := \o_\a{}^\b \, e_\b$. Then we can define the volume 4-form on $M$
\begin{equation}
\eta:= \star 1 =  \vt^0 \wedge \vt^1 \wedge \vt^2
\wedge \vt^3 = \frac{1}{4!} \,
\epsilon_{\a\b\mu\nu} \, \vt^\a \wedge \vt^\b \wedge \vt^\m \wedge \vt^\n
\end{equation}
by the \textbf{star operator} (Hodge dual) $\star:\Lambda^k(M) \to \Lambda^{4-k}(M)$. Using
the Hodge dual, we define a convenient basis called \textbf{$\eta$-basis} for $\Lambda(M) := \bigoplus_k \Lambda^k(M)$, see \cite{Trautman-EC1}-\cite{Trautman-EC4}, \cite{Hehl:1994ue}, which is a vector space.
\begin{equation}\label{E:eta basis}
\begin{aligned}
\eta_\a &:= \star \vt_\a, \quad \qquad  \qquad  \, \qquad \qquad \mbox{(3-form)}\\
\eta_{\a\b} &:=  \star (\vt_\a \wedge \vt_\b ), \,\,\,\,\,\, \qquad \quad \qquad \mbox{(2-form)}\\
\eta_{\a\b\g} &:=   \star (\vt_\a \wedge \vt_\b \wedge \vt_\g ), \, \,\, \quad \,\, \qquad \mbox{(1-form)}\\
\eta_{\a\b\g\delta} &:= \star(\vt_\a \wedge \vt_\b \wedge \vt_\g \wedge \vt_{\delta} ), \qquad \mbox{(0-form)}
\end{aligned}
\end{equation}
where $\vt_\a:= g_{\a\b} \, \vt^\b$, and each of which has $dim \{\eta \} = \binom{4}{0}$, $dim \{\eta_\a \} = \binom{4}{1}$, $dim \{\eta_{\a\b} \} = \binom{4}{2}$, $dim \{\eta_{\a\b\g} \} = \binom{4}{3}$, $dim \{\eta_{\a\b\g\d} \} = \binom{4}{4}$. This is to be compared to the $\vt$-basis, that also forms a basis for each $\Lambda^k(M) $.
\[ \underbrace{1}_{\mathclap{\text{0-form}}} ,\quad \underbrace{\vt^\a}_{\mathclap{\text{1-form}}},\quad \underbrace{\vt^\a \wedge
\vt^\b}_{\mathclap{\text{2-form}}}, \quad \underbrace{\vt^\a \wedge
\vt^\b \wedge \vt^\g}_{\mathclap{\text{3-form}}}, \quad \underbrace{
\vt^\a \wedge \vt^\b \wedge \vt^\g \wedge
\vt^{\delta}}_{\mathclap{\text{4-form}}}
\]
In fact, one can verify that the $\eta$-basis can be written in another form,
\begin{equation}\label{E:eta basis2}
\begin{aligned}
\eta_\a &:= i_{e_\a}  \eta , \qquad \eta_{\a\b} :=  i_{e_\b}   \eta_\a ,\\
\eta_{\a\b\g} &:=  i_{e_\g}   \eta_{\a\b}, \qquad  \eta_{\a\b\g\delta} :=  i_{e_\d}  \eta_{\a\b\g}
\end{aligned}
\end{equation}
where $i_{V} :\Lambda^k(M) \to \Lambda^{k-1}(M)$ denotes the \textbf{interior product} by $(i_{V} \a) (X_1, \cdots, X_{k-1}) := \a(V, X_1, \ldots, X_{k-1} )$ for a given vector (field) $V\in T_M$ with $\a \in \Lambda^k(M)$, and $X_1,\ldots ,X_{k-1} \in\mathfrak{X}(M)$. In some of the literature, it is also denoted by symbol $V\lr \a := i_V \a$, where we take both interchangeably.

The $\eta$-basis is useful because of the following identities that help to reduce computation.
\begin{proposition}

\begin{equation}\label{E:eta basis id}
\begin{aligned}
\vt^\a \wedge \eta_{\b} &= \delta^\a_\b \, \eta\\
\vt^\a \wedge \eta_{\b\g} &= \delta^\a_\g \, \eta_{\b} -
\delta^\a_\b \, \eta_{\g}\\
\vt^\a \wedge \eta_{\b\g\sigma} &= \delta^\a_\b \, \eta_{\g\sigma} +
\delta^\a_\g \, \eta_{\sigma\b} + \delta^\a_{\sigma} \, \eta_{\b\g}\\
\vt^\a \wedge \eta_{\b\g\m\n} &= \d^\a_\n \, \eta_{\b\g\m} - \d^\a_\m \, \eta_{\n\b\g} + \d^\a_\g \, \eta_{\m\n\b} - \d^\a_\b \, \eta_{\g\m\n}
\end{aligned}
\end{equation}
\end{proposition}
\begin{proof}
There are at least three ways to compute the first identity: the first is to utilize
\begin{lemma}
For $\psi, \phi \in \Lambda^k(M)$, we have $\langle \psi, \phi \rangle = - \star \left( \phi \wedge \star \psi \right)$.
\end{lemma}
A second way is compute it directly
\[
\vt^\a \wedge \star \vt_\b = \vt^\a \wedge \left( \frac{1}{3!} \veps_{\b\m\n\g} \, \vt^\mu \wedge \vt^\n \wedge \vt^\g \right) = \frac{1}{3!} \veps_{\b\m\n\g} \, \veps^{\a\m\n\g} \, \eta = \delta^\a_\b \, \eta
\]
or a third way by considering
\[
 0= i_{e_\b} (\vt^\a \wedge \eta) = i_{e_\b} (\vt^\a ) \wedge \eta - \vt^\a \wedge i_{e_\b} (\eta)
\]
due to $\vt^\a \wedge \eta$ being a 5-form, which vanishes identically on 4-dimensional $M$. The second identity can be proved by considering
\[
i_{e_\g} (\vt^\a \wedge \eta_\b) = \d^\a_\g \, \eta_\b - \vt^\a \wedge \left( i_{r_\g} (\eta_\b) \right)
\]
where $i_{e_\g} (\eta_\b) = \eta_{\b\g}$ by using (\ref{E:eta basis2}). The third then follows a similar iteration trick.
\end{proof}
The following covariant derivatives of the $\eta$-basis are also useful in computation.
\begin{proposition}
\begin{equation}\label{E:eta basis id2}
D \eta_\a = \eta_{\a\g} \wedge T^\g, \quad D \eta_{\a\b} = \eta_{\a\b\g} \wedge T^\g, \quad D\eta_{\a\b\m} = \eta_{\a\b\m\g} \, T^\g
\end{equation}
\end{proposition}
with $D\eta =0$ and $D\eta_{\a\b\m\n} =0$.
\begin{proof}
We only prove the second, the rest are similar. Since
\[
D \eta_{\a\b} = \frac{1}{2!} D \left( \veps_{\a\b\m\n} \, \vt^\m \wedge \vt^\n \right) = \veps_{\a\b\m\n} \, T^\m \wedge \vt^\n = T^\m \wedge \star (\vt_\a \wedge \vt_\b \wedge \vt_\m)
\]
\end{proof}
With the convenient notations above, next we introduce the decomposition of the field strengths.

Under the local Lorentz group $SO(1,3)$, the torsion and the curvature are decomposed into 3 and 6 irreducible pieces respectively,
see \cite{Hehl:1994ue}, \cite{McCrea:1992wa}. The torsion tensor of 24 independent components $T_{\mu\nu}{}{}^\a$ is decomposed into
\begin{equation}\label{E:torsion decomp}
\begin{aligned}
T^\a &= {}^{(1)}T^\a + {}^{(2)}T^\a + {}^{(3)}T^\a  \quad \in \Lambda^2(M)\\
(24) &= \,\, (16)  \,\, +(4)  \,\, \, +(4)
\end{aligned}
\end{equation}
where
\begin{equation}\label{E:T irr pieces}
\begin{aligned}
{}^{(1)}T^\a &:= T^\a - {}^{(2)}T^\a - {}^{(3)}T^\a , \qquad
(\mbox{tensor} )\\
{}^{(2)}T^\a &:= \frac{1}{3} \vt^\a \wedge i_{e\b} (T^\b), \qquad
(\mbox{vector} )\\
{}^{(3)}T^\a &:= \frac{1}{3} \star( \star(\vt_\a \wedge T^\a) \wedge \vt^\a),
\qquad (\mbox{axial vector})
\end{aligned}
\end{equation}
and the curvature of 36 independent components $R^\g{}_{\a\b\m}$ is decomposed according
to
\begin{equation}\label{E:Riem irr1}
\O^{\a\b} = \sum_{I=1}^6 {}^{(I)}\O^{\a\b} \quad \in \Lambda^2(M)
\end{equation}
where
\begin{equation}\label{E:Riem irr2}
\begin{aligned}
{}^{(1)}\O_{\a\b} &:= \O_{\a\b} - \sum_{I=2}^6 {}^{(I)}\O_{\a\b} \\
{}^{(2)}\O_{\a\b} &:= (-1) \star \left( \vt_{[\a} \wedge \Psi_{\b]} \right)\\
{}^{(3)}\O_{\a\b} &:= -\frac{1}{12} \star \left( X \wedge \vt_\a \wedge \vt_\b \right)\\
{}^{(4)}\O_{\a\b} &:= (-1) \, \vt_{[\a} \wedge \Phi_{\b]}\\
{}^{(5)}\O_{\a\b} &:= -\frac{1}{2} \vt_{[\a} \wedge e_{\b]} \lr \left( \vt^\g \wedge Ric_\g \right) \\
{}^{(6)}\O_{\a\b} &:= -\frac{1}{12} R \, \vt_{\a} \wedge \vt_{\b}
\end{aligned}
\end{equation}
with
\begin{equation}\label{E:Riem irr3}
\begin{aligned}
         Ric_\a  &:=  e_\b \lr \O^\b{}_\a, \qquad  R:= e_\a \lr Ric^\a, \qquad
         X^\a    :=\star \left( \O^{\a\b}\wedge \vt_\b \right), \qquad  X := e_\a \lr X^\a,\\
         \Psi_\a  &:= X_\a - \frac{1}{4} \vt_\a \wedge X -
         \frac{1}{2} e_\a \lr \left( \vt^\b \wedge X_\b \right),
         \qquad
         \Phi_\a := Ric_\a - \frac{1}{4} R \, \vt_\a - \frac{1}{2}
         e_\a \lr \left( \vt^\b \wedge Ric_\b \right)
\end{aligned}
\end{equation}
where the first term, called the Ricci 1-form is recognized as $Ric_\a  :=  e_\b \lr \O^\b{}_\a = R_{\a\m} \, \vt^\m$ with Ricci tensor coefficients. In general $R_{\a\b} \neq R_{\b\a}$ in the PGT and more generally in Riemann-Cartan geometry (independent of any physics) due to the presence of torsion, and thus it has 16 independent components. The rest of the 5 irreducible pieces have their own properties, see \cite{Hehl:1994ue}. Thus we see from above that if $T^\a = 0$, then ${}^{(2)}\O^{\a\b}=
{}^{(3)}\O^{\a\b} ={}^ {(5)}\O^{\a\b} =0$, which recovers the Riemannian space $V_4$ describing GR.

We also define a convenient quantity to compare the connection of PGT $\o_\a{}^\b $ to the Riemannian one $\widetilde{\o}_\a{}^\b$, the \textbf{contortion 1-form} $K_\a{}^\b := \widetilde{\o}_\a{}^\b - \o_\a{}^\b $. Then we can write $T^\a = K^\a{}_\b \wedge \vt^\b$. In fact, the metricity condition (\ref{E:metricity2}) imposes that $K_{(\a\b)} \equiv 0$. It turns out that one can solve the contortion 1-form in terms of torsion
\begin{equation}
K_{\a\b} = i_{e_[\a}  T_{\b]} - \frac{1}{2} i_{e_\a} \left(i_{e_\b} T_\g \right) \vt^\g
\end{equation}

The differential form formalism, compared to the components formalism, usually can be avoid much of the complicated computation. We conclude this section with an important example:
\begin{ex}{(Gravity Lagrangian)}

Let the gravity Lagrangian be of the 4-form $\mathcal{L}_G := \frac{1}{2\k} \, \O^{\a\b} \wedge \eta_{\a\b}$, we want to translate it into a component form. Expand $\O^{\a\b} \wedge \eta_{\a\b} = \left( \frac{1}{2} R^{\a\b}{}{}_{\m\n} \, \vt^\m \wedge \vt^\n \right) \wedge \eta_{\a\b}$ and using (\ref{E:eta basis id}) repeatedly one obtains $\frac{1}{2} R^{\a\b}{}{}_{\m\n} \, \left(   \d^\n_\b \, \d^\m_\a -  \d^\n_\a \, \d^\m_\b  \right) \, \eta = R \, \eta$. Hence we conclude that the Hilbert-Einstein Lagrangian of GR can be written into the 4-form on $M$, given by \cite{Trautman-EC1}.
\begin{equation}\label{E:HE Lagrangian}
\mathcal{L}_G := \frac{1}{2\k} \, \O^{\a\b} \wedge \eta_{\a\b} = \frac{1}{2\k} \, R \, \sqrt{-g} \, d^4x
\end{equation}
where $\k = 8\pi G / c^4$.
\end{ex}



\section{the Einstein-Cartan theory}

The Einstein-Cartan theory is a special case of Poincar\'{e} gauge gravity theory, and can be considered as a generalization of GR with torsion.
\begin{definition}{(Einstein-Cartan Theory)}\label{Def:Einstein-Cartan}

The \textbf{Einstein-Cartan theory} is defined by the tuple $(M,g,\nabla)$ with the gravity action $\mathcal{L}_{EC} := \frac{1}{2\k} \, \O^{\a\b} \wedge \eta_{\a\b}$, where $\nabla$ is a metric-compatible connection.
\end{definition}

In fact, the connection $\o_\a{}^\b$ in (\ref{E:HE Lagrangian}) need not be restricted to that of Levi-Civita. Hence although the Einstein-Cartan theory and GR share the same form of the action, they have a different geometry in general, especially when spin matter appears. Recall that in GR the metric is coupled to the energy-momentum tensor $\mathcal{T}_{ij}$. In the Einstin-Cartan theory the metric (or coframe) is coupled to the canonical energy-momentum tensor (of the Noether type), i.e,
\begin{equation}
\mathcal{T}_{ij} := g_{ij} \, \mathcal{L}_{\text{M}} - \frac{\partial \mathcal{L}_{\text{M}}}{\partial ( \phi^{a;i}) } \, \phi^a_{;j}
\end{equation}
where $\mathcal{L}_{\text{M}} = \mathcal{L}_{\text{M}}(\phi^a, D_i \phi^a)$ is the matter Lagrangian of the matter field $\phi^a$, the index $a=1, \ldots , k $ denotes the vector components of the matter field, and the indices $i,j,\ldots = 0,1,2,3$ denote the coordinate indices of the spacetime. Note that in general $\mathcal{T}_{ij} \neq \mathcal{T}_{ji} $. The reason for such an asymmetry is due to the presence of torsion, which shall be explained later.

The torsion tensor in the Einstein-Cartan theory is coupled to the \textbf{spin current} tensor of a matter field $\phi $ defined by \cite{Hehl:1994ue}
\begin{equation}\label{E:spin current}
\mathcal{S}_{ij}{}{}^k = \left( \frac{\p \mathcal{L}_{\text{M}}}{\p \phi^a_{;k}} \right) \,  \rho_{[ij]}{}{}_b{}{}^a \cdot \phi^b = - \mathcal{S}_{ji}{}{}^k
\end{equation}
where $\rho_{[ij]}{}{}_b{}{}^a := \left( \rho_* E_{ij} \right)^a_b $
is a Lie-algebra representation of $\rho: SO(1,3) \to GL(V)$ and $ \{ E_{ij} \in \mathfrak{so}(1,3) \}$ is a basis of $\mathfrak{so}(1,3)$ (for details see (\ref{E:spin current2})). The spin current (\ref{E:spin current}) comes from the conserved current of the internal $SO(1,3)$ symmetry, see Appendix. \ref{APP:spin current} and couples to torsion. Thus in PGT the spin does play a role in the spacetime evolution, unlike GR. However, in terms of differential forms the canonical energy-momentum tensor and the spin-current of matter are found to be simpler expressed by the 3-forms, see \cite{Hehl:1994ue}, \cite{Trautman-EC1}, \cite{Hehl:2014eja}:
\begin{equation}\label{E:matter current}
\mathcal{T}_{\a} = \frac{\d \mathcal{L}_{\text{M}}}{\d \vt^\a}, \qquad
\mathcal{S}_{\a\b} = \frac{\d \mathcal{L}_{\text{M}}}{\d \o^{\a\b}}
\end{equation}
where more specific explanation will be given later.

The field equations for the Einstein-Cartan theory are given in component form by \cite{Sciama}, \cite{Kibble},
\begin{equation}\label{E:Sciama-Kibble eq}
\begin{aligned}
R_{ij} - \frac{1}{2} g_{ij} \, R &= \k \mathcal{T}_{ij}\\
T_{ij}{}{}^k + \d^k_i \, T_{j l}{}{}^l -  \d^k_j \, T_{i l}{}{}^l &= \k \mathcal{S}_{ij}{}{}^k
\end{aligned}
\end{equation}

One can also write the (Einstein-Cartan)-Sciama-Kibble equation in a simpler form. Before the derivation, we first observe a property that

\begin{theorem}\label{Thm:EC to GR}

If the spin current of matter vanishes, $\mathcal{S}_{ij}{}{}^k \equiv 0$, then the Einstein-Cartan theory reduces to GR. Hence GR is a degenerate case of the Einstein-Cartan.
\begin{proof}
If $\mathcal{S}_{ij}{}{}^k \equiv 0$, then from $(\ref{E:Sciama-Kibble eq})_2$ contracting the indices $j$ and $k$, one has $T_{il}{}{}^l \equiv 0$, which indicates that in fact the whole torsion tensor $T_{ij}{}{}^k \equiv 0$ again by $(\ref{E:Sciama-Kibble eq})_2$, and thus recovers GR.
\end{proof}

\end{theorem}
Thus in the Einstein-Cartan theory without spin matter the second equation of (\ref{E:Sciama-Kibble eq}) is simply null. Now we derive (\ref{E:Sciama-Kibble eq}) by using the differential forms formalism.

\begin{theorem}{(Sciama-Kibble equation)}

In differential forms, the Einstein-Cartan field equations are written as
\begin{equation}\label{E:EC EOM}
\begin{aligned}
&\frac{1}{2} \eta_{\a\b\g} \wedge \O^{\g\b} = \k \,\mathcal{T}_\a \\
&\frac{1}{2} \eta_{a\b\g}\wedge T^{\g}     = \k \, \mathcal{S}_{\a\b}
\end{aligned}
\end{equation}
\end{theorem}

\begin{proof}
Since the Einstein-Cartan theory is a special case of PGT, the two gauge potentials $(\vt^a , \o_\a{}^\b)$ are considered as independent variables. Gravitational variation is of the form
\begin{equation}\label{E:EC var0}
\d \mathcal{L}_G = \d \vt^\a \wedge \frac{\d \mathcal{L}_G}{\d \vt^\a } + \d \o^{\a\b} \wedge \frac{\d \mathcal{L}_G}{\d \o^{\a\b} } + d (\cdots)
\end{equation}
where $d (\cdots)$ is some exact differential term which vanishes upon the integration over a closed 4-manifold by Stoke's theorem. Similarly,
\begin{equation}\label{E:EC mat var}
\d \mathcal{L}_{\text{M}} = \d \vt^\a \wedge \frac{\d \mathcal{L}_{\text{M}}}{\d \vt^\a } + \d \o^{\a\b} \wedge \frac{\d \mathcal{L}_{\text{M}}}{\d \o^{\a\b} } + d (\cdots)
\end{equation}

Consider the total variation
\begin{equation}\label{E:EC var1}
\d \left( \eta_{\a \b} \wedge \O^{\a \b} \right) = \delta \eta_{\a\b} \wedge \O^{\a\b} + \eta_{\a\b} \wedge \delta  \O^{\a\b},
\end{equation}
where the variational formula $ \d \left( \psi \wedge \phi \right) = \d \psi \wedge \phi + \psi \wedge \d \phi$ is used for all $\psi \in \Lambda^p(M)$, $\phi \in \Lambda^q(M)$ and since
\begin{equation}\label{E:variational formula}
\begin{aligned}
\delta \eta_{\a\b}& = \frac{1}{2!} \d \left(  \veps_{\a\b\m\n} \, \vt^\m \wedge \vt^\n \right)  = \d \vt^\m \wedge \eta_{\a\b\m} \\
\d \O^{\a\b} &= d \d \o^{\a\b} + \d \o^{\g\b} \wedge \o^\a{}_\g +
\o^{\g\b} \wedge \d \o^\a{}_\g\\
           &= d \d \o^{\a\b}+ \o_\g{}^\a \wedge \d
\o^{\g\b} + \o_\g{}^\b \wedge \d \o^{\a\g} = D\d \o^{\a\b}.
\end{aligned}
\end{equation}
then (\ref{E:EC var1}) reads
\begin{equation}
\begin{aligned}
\label{E:EC var2}
\d \left( \eta_{\a \b} \wedge \O^{\a \b} \right) &= \d \vt^{\mu} \wedge \left( \eta_{\a\b\mu} \wedge \O^{\a \b} \right)  + \eta_{\a\b} \wedge  D\d \o^{\a\b}\\
 &= \d \vt^{\mu} \wedge \left( \eta_{\a\b\mu} \wedge \O^{\a \b} \right) + D\left( \eta_{\a\b} \wedge \d \o^{\a\b} \right) - (D \eta^{\a\b} ) \wedge (\d \eta_{\a\b})\\
 &= \d \vt^{\mu} \wedge \left( \eta_{\a\b\mu} \wedge \O^{\a \b} \right) + \d \o^{\a\b} \wedge \left( T^\g \wedge \eta_{\a\b\g} \right)+  D \left( \eta_{\a\b} \wedge \d \o^{\a\b} \right)
\end{aligned}
\end{equation}
where we have used $D( \eta_{\a\b} \wedge \d \o^{\a\b}) = D \eta_{\a\b} \wedge \d \o^{\a\b} + \eta_{\a\b} \wedge D \d \o^{\a\b} $ in the second equality and the identity $D\eta_{\a\b} = T^\g \wedge \eta_{\a\b\g}$ from (\ref{E:eta basis id2}) is used in the last equality. Since the last term $D( \eta_{\a\b} \wedge \d \o^{\a\b}) = d ( \eta_{\a\b} \wedge \d \o^{\a\b})$ in (\ref{E:EC var2}) vanishes upon integration, by comparison to (\ref{E:EC var0}) one derives the field equations:
\begin{equation}
\begin{aligned}
\frac{\d \mathcal{L}_{EC}}{\d \vt^{\mu} } &= \frac{1}{2\k} \eta_{\mu\a\b} \wedge \O^{\a\b} =  -\mathcal{T}_{\mu}\\
\frac{\d \mathcal{L}_{EC}}{\d \o^{\a\b} } &= \frac{1}{2\k} \eta_{a\b\g}\wedge T^{\g}  = \mathcal{S}_{\a\b}
\end{aligned}
\end{equation}

\end{proof}

\begin{remark}
Here we remark that the variational formula $(\ref{E:variational formula})_1$ can be derived rigorously from the formula in Eq.(33) of \cite{Muench:1998ay}, namely
\begin{theorem}
For all $\phi \in \Lambda^p(M)$, and an arbitrary frame $(e_\a) \in TM$, one has
\begin{equation}
\left( \d \cdot \star - \star \cdot \d  \right) \, \phi  = \d \vt^\a \wedge (i_{e_\a} (\star \phi)) - \star [\d \vt^\g \wedge i_{e_\g} \phi ] + \d g_{\a\b} \left[ \vt^{(\a} \wedge \left( i_{e^{\b)}} (\star \phi) \right) - \frac{1}{2} g^{\a\b}\, \star \phi   \right]
\end{equation}
which indicates that the Hodge dual operator does not commute with the variational operator $\d$ in general. To obtain $(\ref{E:variational formula})_1$, one applies $\phi = \vt^\a \wedge \vt^\b$ with the orthonormal frame condition such that $\d g_{\a\b} \equiv 0$.
\end{theorem}

\end{remark}

Here we remark that the translation from (\ref{E:EC EOM}) to (\ref{E:Sciama-Kibble eq}) is direct as the following: for $(\ref{E:EC EOM})_1$
\begin{equation}
\frac{1}{2} \, \eta_{\a\b\g} \wedge \O^{\g\b} = \frac{1}{4} \,  \eta_{\a\b\g} \wedge \left( R^{\g\b}{}{}_{\mu\nu} \, \vt^{\m} \wedge \vt^\n \right) = \frac{1}{2} \left( R_{\b\a} - \frac{1}{2} g_{\a\b} \, R \right) \, \eta^\b = \k \, \mathcal{T}_{\a\b} \, \eta^\b
\end{equation}
where we first expand $\O^{\g\b}$ with (\ref{E:T and R}) and apply (\ref{E:eta basis id}) repeatedly. Notice in the last equality we have defined the coefficients of the 3-form expansion $\mathcal{T}_{\a}$ in terms of the 3-form basis $\eta^\a$ as the \textbf{energy-momentum tensor} in the usual sense. The coefficient $\mathcal{S}_{\a\b\g}$ called \textbf{spin current tensor} is similarly defined i.e,
\begin{equation}\label{E:source 3-form}
\mathcal{T}_{\a}  := \mathcal{T}_{\a\b} \, \eta^\b, \qquad \mathcal{S}_{\a\b} =  \mathcal{S}_{\a\b\g} \, \eta^{\g}
\end{equation}
Then we see from this expansion that in general $\mathcal{T}_{\a\b} \neq \mathcal{T}_{\b\a}$ and $\mathcal{S}_{\a\b\g} \neq \mathcal{S}_{\a\g\b}$. The translation for $(\ref{E:EC EOM})_2$ is similar.

As we have seen from the proof of Theorem (\ref{Thm:EC to GR}) in $(\ref{E:Sciama-Kibble eq})_2$ that the relationship between torsion and spin currents is \emph{algebraic}.\footnote{"Algebraic" means only operations of $+,-,\times, \div$ and $\sqrt[n]{\cdot}$ are involved, especially not differential or integral operators.} Basically, the Einstein-Cartan field equations (\ref{E:EC EOM}) are of $1^{st}$ order partial differential equations (PDEs) in $(\vt^\a, \o_\a{}^\b)$ if it is supported by a spin fluid source, \cite{Hehl:1976kj}, \cite{Gronwald:1995em}. With a detailed examination of the field equations, in the end one finds that the torsion field $T^\a = D \vt^\a $ is not a dynamical field, and hence not propagating. This is unlike the typical Yang-Mills gauge theory, $\mathcal{L}_{\text{YM}} \sim F \wedge F$, whose field equations are generally of $2^{nd}$ order PDEs in field variables. Thus from the gauge theoretical point of view, EC is a degenerate theory. Following the essence of Yang-Mills, Hehl, Nitsch, and Von der Heyde \cite{Hehl:1979gp} constructed a more general framework for PGT that in addition to the linear scalar curvature $R$ in the gravity action (Einstein-Cartan), quadratic terms of $(\O^\b{}_\a, T^\a)$ in the Lagrangian should be also in consideration, hence the introduction of quadratic PGT, \cite{Hayashi:1981mm}, \cite{Hehl:1979gp}.


\section{Quadratic PGT (qPG)}

\subsection{Field equations for general PGT}

On a Riemann-Cartan spacetime $(M,g, \nabla)$, we consider a general Lagrangian of the form
\begin{equation}\label{E:general Lagrangian}
\mathcal{L} = \mathcal{L}(g_{\a\b} , \vt^\a, d \vt^\a , \o_\a{}^\b, d\o_\a{}^\b , \phi^a, d\phi^a)
\end{equation}
which is a total Lagrangian containing both gravity and matter field of $\phi \in \Lambda^p (U \subseteq M ,V)$ and the index $a$ denotes the component in the vector space $V$. However, if we require $\mathcal{L} $ to have the Lorentz symmetry, i.e, under the transformation of two frames at $x\in M$
\begin{equation}\label{E:local Lorentz trans}
\widetilde{e}_\a (x) = \left( A(x) \right)_\a^\b \, \vt^\b (x), \qquad  \widetilde{\vt}^\a = \left( A^{-1}(x) \right)^\a_\b \, \vt^\b (x)
\end{equation}
where $A = (A_\a^\b) : U \subseteq M \to SO(1,3)$, the Lagrangian $\mathcal{L} $ has an invariant value. Then $d$ and $\o_\a{}^\b$ in $\mathcal{L} $ can only be involved with certain combination such that under (\ref{E:local Lorentz trans}) the Lorentz symmetry is preserved. In fact, it can only be the combination of the exterior covariant derivative given in Definition (\ref{Def:Exterior covariant derivative2}) or Theorem (\ref{Thm:ext derivative})
\[
D \phi = d \phi  + \o \dot{\wedge} \phi , \quad \Leftrightarrow \quad D = d \phi  + \rho_\b^\a \, \o_\a{}^\b \wedge \phi
\]
where $\rho: SO(1,3) \to GL(V)$. Therefore

\begin{proposition}{($SO(1,3)$-invariant Lagrangian)}

By imposing the Lorentz symmetry, $G = SO(1,3)$, the general Lagrangian (\ref{E:general Lagrangian}) is reduced down to the form
\begin{equation}\label{E:reduced Lagrangian}
\mathcal{L} = \mathcal{L}(g_{\a\b} , \vt^\a , \phi^a, D\phi^a, T^\a, \O^\b{}_\a)
\end{equation}
\end{proposition}
One defines the variation of the independent arguments
\begin{equation}\label{E:action variation 1}
\delta \mathcal{L} = \delta \vt^\a \wedge \frac{\p
\mathcal{L}}{\p\vt^\a} + \delta T^\a \wedge \frac{\p
\mathcal{L}}{\p T^\a} + \delta \O^\b{}_\a \wedge \frac{\p
\mathcal{L}}{\p \O^\b{}_\a} + \delta \phi^a \wedge
\frac{\p \mathcal{L}}{\partial \phi^a} +\delta (D\phi^a)
\wedge \frac{\p \mathcal{L}}{\partial (D\phi^a)}
\end{equation}
Moreover, if we replace the variations of $\d T^\a $ and $\delta \O^\b{}_\a$ by
\begin{equation}
\delta T^\a = D \left( \delta \vt^\a \right) + \delta \o_\b{}^\a \wedge \vt^\b, \qquad
\delta \O^\b{}_\a = D \delta \o_\a{}^\b
\end{equation}
where the identities are derived by applying the relation $[\delta,d]=0$. As a consequence, one has
\begin{multline}\label{E:action variation 2}
\delta \mathcal{L} = \delta \vt^\a \wedge \left[ \frac{\partial \mathcal{L}}{\partial \vt^\a} + D \left(
\frac{\partial \mathcal{L}}{\partial T^\a} \right) \right] + \delta
\o^{\a\b} \wedge \left[ \rho_{\a\b}{}{}^a{}_b \, \phi^b \wedge \frac{\partial
\mathcal{L}}{\partial (D \phi^a) } + \vt_{[\a} \wedge \frac{\partial
\mathcal{L}}{\partial T^{\b]}} + D \left( \frac{\partial
\mathcal{L}}{\partial \O^{\b\a}} \right) \right]\\
 + \delta \phi^a \wedge
\frac{\delta\mathcal{L}}{\delta \phi^a} + d\left( \delta \vt^\a
\wedge \frac{\partial\mathcal{L}}{\partial T^\a} + \delta
\o_\a{}^\b \wedge \frac{\partial\mathcal{L}}{\partial \O^\b{}_\a}
+ \delta \phi^a\wedge \frac{\partial\mathcal{L}}{\partial D \phi^a}
\right)
\end{multline}
where $\rho_{\a\b}{}{}^a{}_b$ is defined in (\ref{E:spin current2}) and $
D \d \vt^\a \wedge \frac{\partial \mathcal{L}}{\partial T^\a} = D \left(  \d \vt^\a \wedge \frac{\partial \mathcal{L}}{\partial T^\a} \right) +  \d \vt^\a \wedge D \left( \frac{\partial \mathcal{L}}{\partial T^\a} \right) $ is used. By the comparison of (\ref{E:action variation 1}) to (\ref{E:action variation 2}), we obtain \cite{Hehl:1994ue}
\begin{equation}
\begin{aligned}
\frac{\d \mathcal{L}}{\d\vt^\a} &=  \frac{\partial \mathcal{L}}{\partial \vt^\a} + D \left(
\frac{\partial \mathcal{L}}{\partial T^\a} \right), \\
 \frac{\d
\mathcal{L}}{\d\o^{\b\a}} &=  \rho_{\a\b}{}{}^a{}_b \, \phi^b \wedge \frac{\partial
\mathcal{L}}{\partial (D \phi^a) } + \vt_{[\a} \wedge \frac{\partial
\mathcal{L}}{\partial T^{\b]}} + D \left( \frac{\partial
\mathcal{L}}{\partial \O^{\b\a}} \right)
\end{aligned}
\end{equation}

In particular, if the general Lagrangian (\ref{E:general Lagrangian}) is decomposable into the form
\begin{equation}\label{E:decomposable Lagrangian}
\mathcal{L} = \mathcal{L}_G(\vt^\a, T^\a, \O_\a{}^\b) + \mathcal{L}_{M}(\vt^\a, \phi^a , D\phi^a)
\end{equation}
where $\mathcal{L}_G$ denotes the gravitational Lagrangian in PGT and $\mathcal{L}_{M}$ is the matter Lagrangian. (\ref{E:action variation 2}) simply tells us

\begin{theorem}{(Field equations for PGT)}

The field equations for the Lagrangian (\ref{E:decomposable Lagrangian}) is given by
\begin{equation}\label{E:PGT EOM}
\begin{aligned}
\frac{\partial \mathcal{L}_G}{\partial \vt^\a} + D \left(
\frac{\partial \mathcal{L}_G}{\partial T^\a} \right) &= \mathcal{T}_\a, \\
\vt_{[\a} \wedge \frac{\partial \mathcal{L}_G}{\partial T^{\b]}} + D \left( \frac{\partial
\mathcal{L}_G}{\partial \O^{\b\a}} \right) &= \mathcal{S}_{\a\b}
\end{aligned}
\end{equation}
where $\mathcal{T}_\a \in \Lambda^3(M)$ and $\mathcal{S}_{\a\b} \in \Lambda^3(M)$ are the \textbf{(canonical) energy-momentum current} and \textbf{spin-current of matter} defined by \cite{Hehl:1994ue}
\begin{equation}\label{E:matter current2}
\begin{aligned}
\mathcal{T}_\a &:= \frac{\d\mathcal{L}_M}{\d \vt^\a} = \frac{\partial \mathcal{L}_M}{\partial \vt^\a} + D \left(
\frac{\partial \mathcal{L}_M}{\partial T^\a} \right),\\
\mathcal{S}_{\a\b} &:= \frac{\d\mathcal{L}_M}{\d \o^{[\a\b]}} = \rho_{\a\b}{}{}^a{}_b \, \phi^b \wedge \frac{\partial
\mathcal{L}}{\partial (D \phi^a) } + \vt_{[\a} \wedge \frac{\partial \mathcal{L}_M}{\partial T^{\b]}} + D \left( \frac{\partial
\mathcal{L}_M}{\partial \O^{\b\a}} \right),
\end{aligned}
\end{equation}
with the evolution of the matter field given by
\begin{equation}\label{E:PGT matter EOM}
\frac{\d \mathcal{L}_M}{\d \phi} = \frac{\partial \mathcal{L}_M}{\partial \phi} - D \left( \frac{\partial \mathcal{L}_M}{\partial (D\phi)}\right) =0.
\end{equation}
Together (\ref{E:PGT EOM}) and (\ref{E:PGT matter EOM}) constitute the complete evolution of PGT in the presence of the matter field $\phi$.
\end{theorem}

In fact, computing terms like $\frac{\partial \mathcal{L}_G}{\partial \vt^\a}$, $\frac{\partial \mathcal{L}_G}{\partial T^\a}$, \ldots in (\ref{E:PGT EOM}) require more work. The following theorem can make computation easier, see \cite{Hehl:1994ue}
\begin{theorem}{(Field equations for PGT)}

The field equations for the Lagrangian (\ref{E:decomposable Lagrangian}) are given by
\begin{equation}\label{E:PGT EOM2}
\begin{aligned}
DH_\a - t_\a &= \mathcal{T}_\a ,\qquad DH_{\a\b} - s_{\a\b} &= \mathcal{S}_{\a\b}
\end{aligned}
\end{equation}
where $H_\a \in \Lambda^2(M)$ and $H_{\a\b} \in \Lambda^2(M)$ are called the \textbf{translational excitation} and the \textbf{Lorentz excitation} (field
momenta), defined by
\begin{equation}\label{E:excitations}
H_\a := - \frac{\partial \mathcal{L}_G}{\partial T^\a}, \quad H_{\a\b} := -
\frac{\partial \mathcal{L}_G}{\partial \O^{\a\b}}
\end{equation}
with $t_\a \in \Lambda^3(M)$ and $s_{\a\b} \in \Lambda^3(M)$ called the \textbf{gravitational energy-momentum} and the
\textbf{gravitational spin current}, respectively, given defined by
\begin{equation}\label{E:gravitational energy-momentum}
\begin{aligned}
t_\a &:= \frac{\partial \mathcal{L}_G}{\partial \vt^\a} = i_{e_\a}  \mathcal{L}_G +  \left( i_{e_\a} T^\b \right) \wedge H_\b +
\left( i_{e_\a} \O^{\g\b} \right) \wedge H_{\b\g}, \\
s_{\a\b} &:=
\frac{\partial \mathcal{L}_G}{\partial \o^{\a\b}} = - \vt_{[\a} \wedge
H_{\b]}
\end{aligned}
\end{equation}
\end{theorem}

Both the last equalities in (\ref{E:gravitational energy-momentum}) require some non-trivial procedures concerning the local diffeomorphisms and local Lorentz invariance, see \cite{Hehl:1994ue}. With this theorem, one easily derives the Einstein-Cartan theory (\ref{E:EC EOM}) alternatively.

Indeed, with the Einstein-Cartan Lagrangian (\ref{Def:Einstein-Cartan}), one computes that the gravitational excitations from (\ref{E:excitations}) and (\ref{E:gravitational energy-momentum})
\begin{equation}
H_\a := - \frac{\partial \mathcal{L}_G}{\partial T^\a} =0 , \quad H_{\a\b} := -
\frac{\partial \mathcal{L}_G}{\partial \O^{\a\b}} = \eta_{\a\b}
\end{equation}
and the gravitational currents
\begin{equation}
t_\a = \O^{\b\g} \wedge \eta_{\b\g\a}, \qquad s_{\a\b} = 0
\end{equation}
which recover (\ref{E:EC EOM}).


\subsection{Quadratic PGT Lagrangians}

In the last two sections, ones sees that the Einstein-Cartan theory as a Poincar\'{e} gauge theory is in many ways degenerate. Unlike Yang-Mills, there is a self-interaction of the form $F \wedge \star F$, quadratic in the field strength. Naturally, one would expect that gravity as a gauge theory may have analogous terms like $T \wedge \star T$ and $\O \wedge \star \O$. Using the irreducible decomposition of field strengths for gravity in PGT, (\ref{E:torsion decomp})--(\ref{E:Riem irr3}) one can consider a general Lagrangian quadratic in PGT field strengths,
\begin{multline}\label{E:qPG Lagrangian1}
\mathcal{L}_{qPG} = \frac{1}{2\k} \left[ a_0 \, \O^{\a\b} \wedge \eta_{\a\b} - 2\Lambda \eta + T^\a \wedge \left( \sum_{I=1}^3  a_I \star{}^{(I)} T_\a \right) \right]\\
-\frac{1}{2\varrho} \O^{\a\b} \wedge \left( \sum_{I=1}^6 b_I \star {}^{(I)}\O_{\a\b} \right)
\end{multline}
where the self-coupling constants $a_0, \ldots ,a_3$, $b_1,\ldots, b_6$ are defined to be dimensionless, $[\k]= T^2/M L$, and $[\varrho] = T/ML^2$. In \cite{Baekler:2010fr},  the above quadratic combinations in (\ref{E:qPG Lagrangian1}) are classified into 2 types,:
\begin{equation}\label{E:parity even}
\begin{aligned}
\mathcal{L}^+_{\text{weak}} &= \frac{1}{2\k} \left[ a_0 \, \O^{\a\b} \wedge \eta_{\a\b} - 2\Lambda \eta + T^\a \wedge \left( \sum_{I=1}^3  a_I \star{}^{(I)} T_\a \right) \right]\\
\mathcal{L}^+_{\text{strong}} &= -\frac{1}{2\varrho} \O^{\a\b} \wedge \left( \sum_{I=1}^6 b_I \star {}^{(I)}\O_{\a\b} \right)
\end{aligned}
\end{equation}
where both $\mathcal{L}^+_{\text{weak}}$ and $\mathcal{L}^+_{\text{strong}}$ belong to the \emph{parity even} pieces, which are easily observed by the presence of Hodge dual operator. In fact, in \cite{Baekler:2010fr} they considered a wider class of quadratic combinations including terms like
\begin{equation}\label{E:parity odd}
\begin{aligned}
\mathcal{L}^-_{\text{weak}} &=  \frac{b_0}{2\k} {}^{(3)}\O_{\a\b} \wedge \vt^\a \wedge \vt^\b + \frac{1}{\k} \left( \sigma_1 {}^{(1)}T^\a \wedge {}^{(1)}T_\a + \sigma_2 {}^{(2)}T^\a \wedge {}^{(3)}T_\a\right)\\
\mathcal{L}^-_{\text{strong}} &= -\frac{1}{2\varrho} \left( \m_1 {}^{(1)}\O^{\a\b} \wedge {}^{(1)}\O_{\a\b} + \m_2 {}^{(2)}\O^{\a\b} \wedge {}^{(4)}\O_{\a\b} + \m_3 {}^{(3)}\O^{\a\b} \wedge {}^{(6)}\O_{\a\b} + \m_4 {}^{(5)}\O^{\a\b} \wedge {}^{(5)}\O_{\a\b} \right)
\end{aligned}
\end{equation}
which are considered as \emph{parity odd} terms. Also it is shown in \cite{Baekler:2010fr} that the most general quadratic combinations can be written as
\begin{multline}\label{E:total qPG}
\mathcal{L}_{\text{qPG}} = \frac{1}{2\k} \left[ \left( a_0 R - 2\Lambda + b_0 X \right) \, \eta \right. \\
     \left. {} + \frac{a_2}{3} \mathcal{V} \wedge \star \mathcal{V}  - \frac{a_3}{3} \mathcal{A} \wedge \star \mathcal{A} - \frac{2\sigma_2 }{3} \mathcal{V} \wedge \star \mathcal{A} + a_1 {}^{(1)}T^\a \wedge \star {}^{(1)}T_\a \right] \\
     - \frac{1}{2\varrho} \left[ \left( \frac{b_6}{12}R^2 - \frac{b_3}{12}X^2 + \frac{\mu_3}{12} R \, X \right) \eta + b_4 {}^{(4)}\O^{\a\b} \wedge \star {}^{(4)}\O_{\a\b}\right.\\
      \left. {} + {}^{(2)}\O^{\a\b} \wedge \left( b_2 \star {}^{(2)}\O_{\a\b} + \mu_2 {}^{(4)}\O_{\a\b} \right) + {}^{(5)}\O^{\a\b} \wedge \left( b_5 \star {}^{(5)}\O_{\a\b} + \mu_4 {}^{(5)}\O_{\a\b} \right) \right]
\end{multline}
where $ \mathcal{V}:= i_{e_\g} (T^\g)$ and $\mathcal{A}: = \star (\vt_\a \wedge T^\a)$ denote the \textbf{vector} and \textbf{axial} parts of the torsion tensor, respectively. It has up to 8 independent variables in the Lagrangian, but where do the physics of these terms lie? This is a question to be thoroughly explored. However, recently, Shie, Nester, Yo (SNY, \cite{Shie:2008ms}) have chosen a small subclass out of (\ref{E:total qPG}) by the Hamiltonian analysis of Poincar\'{e} gauge gravity \cite{Yo:1999ex}, \cite{Yo:2001sy}. The SNY-model is simple but physically interesting, which we explore in the next chapter.

In this chapter we have gone through the basic formulation of PGT in principal fibre bundle language, which is suitable for the general gauge theory as well. Such a formulation helps the transition from the typical Yang-Mills theory to gravity gauge theory and shows their many similarities. Later, we project the PGT gauge potentials and strengths $(\varphi, \o , \O^\o, \Theta^\o)$ onto the spacetime $(\vt^\a, \o_\a{}^\b, \O^{\a\b}, T^\a)$ via a local section $\sigma: U\subset M \to P$. Such reduction addresses that the Poincar\'{e} gauge group invariance results in the field strength of torsion and curvature one observes on the spacetime $M$, and clearly address why we obtain torsion and curvature simultaneously when gauging the Poincar\'{e} group into the spacetime. The resultant spacetime is called the Riemann-Cartan spacetime, described by $(M,g,\nabla)$. Thus this is the version that appears most in the physics literature.

We also studied a simplest Lagrangian, $\mathcal{L}_G = \frac{1}{2\k} \, \O^{\a\b} \wedge \eta_{\a\b}$ (\ref{E:HE Lagrangian}) on a Riemann-Cartan spacetime, known as the Einstein-Cartan theory developed around 1961 by Sciama and Kibble. We also sketched the quadratic PGT as a natural extension mimicking the Yang-Mills theory. It can be seen that PGT has a more general structure than GR by its nature, and thus contains more variety and possibility.

Below we show a diagram, Fig.(\ref{fg:RC}), that may best represent the relationship between different geometries to conclude this chapter.

\begin{figure}[h]
\begin{center}
\scalebox{0.6}{\includegraphics{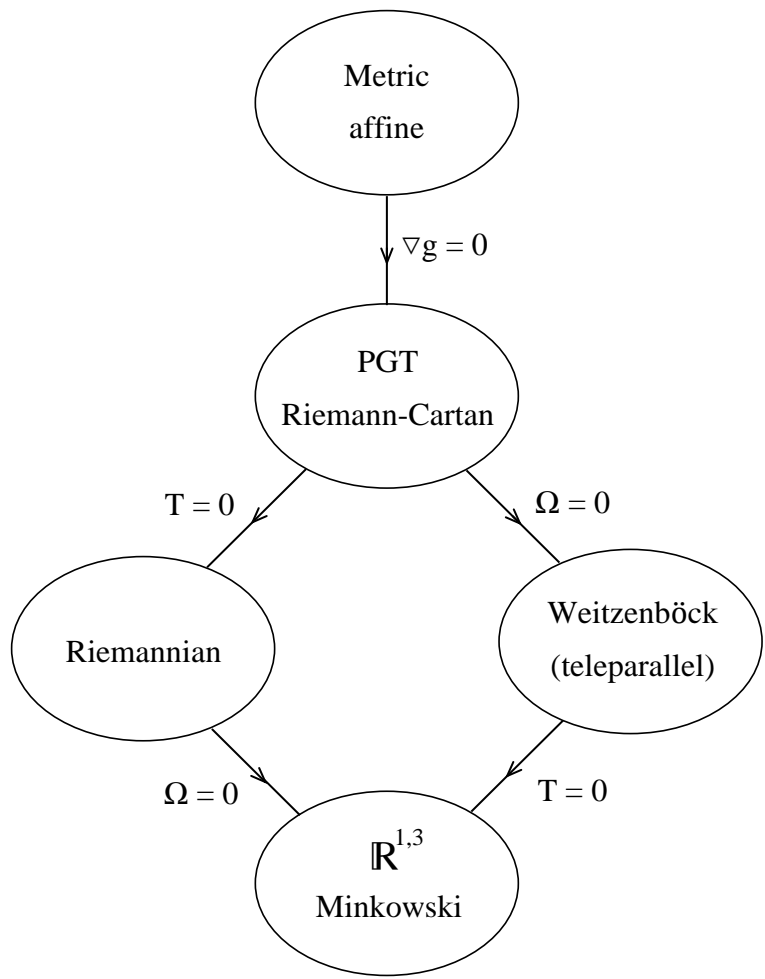}} \caption{Riemann-Cartan
(RC) space and its subcases.}
\end{center}
\label{fg:RC}
\end{figure}


\chapter{Scalar-torsion Mode of PGT}

\section{Lagrangian for the scalar-torsion mode}

As shown in Chater 2, qPG gravity consists of a Riemann-Cartan spacetime $(M,g,\nabla)$ and a Lagrangian 4-form
\begin{equation}\label{E:Lagrangian}
\mathcal{L}(g,\vt,\Gamma) = \mathcal{L}_G + \mathcal{L}_{M}\,,
\end{equation}
where $\vt^\a $ is a set of tetrad,
$\o_\a{}^\b \in \Lambda^1(M)$ is the connection 1-form with respect to $\vt^\a$,
$\mathcal{L}_M$ is the matter Lagrangian, and $\mathcal{L}_G$ is the
gravitational Lagrangian that can be made up by certain
combinations. In~\cite{Shie:2008ms}, SNY studied the spin $0^+$
mode, given by~\cite{Shie:2008ms,Hehl:2012pi}
\begin{equation}\label{E:SNY}
\mathcal{L}_{G} =  \frac{a_0}{2} R \eta + \frac{b}{24} R^2 \eta +
\frac{a_1}{8} T^\a \wedge\star \left(  {}^{(1)} T_\a -2 {}^{(2)} T_\a - \frac{1}{2} {}^{(3)} T_\a
\right)\,,
\end{equation}
where ${}^{(I)} T^\a$ are irreducible torsion pieces in (\ref{E:torsion decompose}), and the coefficients of $\mathcal{L}_G$ in (\ref{E:SNY}) are constrained by the positivity argument~\cite{Shie:2008ms} such that
\begin{equation}\label{E:condition}
a_1 > 0 ,\qquad b > 0.
\end{equation}


\section{Cosmology in Scalar-torsion Mode of PGT}

To study the cosmology of the scalar-torsion gravity in PGT, we
consider the homogeneous and isotropic FLRW universe, given by the
metric
\begin{equation}\label{E:FRW}
ds^2 = -dt^2 + a^2(t) \left( \frac{dr^2}{1-kr^2} + r^2
d\Omega^2\right),
\end{equation}
where $k$ is the sectional curvature of the spatial homogeneous universe and we set $k=0$ for simplicity.

For (\ref{E:SNY}) of  the SNY model  in the FRLW cosmology with
no spin source $S_{ijk} \equiv 0$, defined in (\ref{E:spin current}), the field equations of PGT (\ref{E:PGT EOM2}) lead to~\cite{Shie:2008ms}
\begin{eqnarray}
\label{E:main eq1}
\dot{H} &=& \frac{\mu}{6 a_1} R + \frac{1}{6a_1} \mathcal{T} -2H^2,\\
\label{E:main eq2} \dot{\Phi}(t) &=& \frac{a_0}{2a_1}R +
\frac{\mathcal{T}}{2 a_1} - 3 H
\Phi + \frac{1}{3} \Phi^2,\\
\label{E:main eq3} \dot{R} &=& -\frac{2}{3} \left( R +
\frac{6\mu}{b} \right) \Phi,
\end{eqnarray}
where $\mu = a_1 + a_0$, $H= \dot{a}(t)/a(t)$, and $\Phi(t)= T_t$,
which is the time component of the torsion trace, defined by
$T_i := T_{ij}{}^j$, and the coordinate indices $i,j,k,\ldots $ run from $ 0 , \ldots 3$.
Here, $R$ in (\ref{E:main eq1})-(\ref{E:main eq3}) denotes the
affine curvature in (\ref{E:Riem irr3}). In addition, we have the relation
\begin{equation}
R = \bar{R} + 2 T^j_{;j} - \frac{2}{3}
T_kT^k\,,
\end{equation}
where $\bar{R} = 6(\dot{H} + 2H^2)$ represents  the curvature of the
Levi-Civita connection induced by (\ref{E:FRW}). The energy-momentum
tensor $\mathcal{T}_{ij}$ is defined as (\ref{E:source 3-form}) and $\mathcal{T}$ stands for the trace
$\mathcal{T}_i{}^i$. Explicitly, one has
\begin{equation}\label{E:rho_T}
\begin{aligned}
\mathcal{T}_{tt} &= \rho_M = \frac{b}{18} \left( R + \frac{6\mu}{b}
\right) \left( 3H - \Phi \right)^2 - \frac{b}{24}R^2 - 3a_1 H^2,\\
\mathcal{T} &= 3p_M- \rho_M \,.
\end{aligned}
\end{equation}
with the subscript $M$ representing the ordinary matter including
both dust and radiation. To see the geometric effect of torsion, we
can write down the Friedmann equations as
\begin{eqnarray}\label{E:Friedmann s-t mode}
H^2 &=& \frac{\rho_c}{3a_0}, \qquad  \quad \quad \quad \rho_c =\rho_M + \rho_T,\nonumber \\
 \dot{H} &=& -\frac{\rho_c + p_{tot}}{3 a_0}, \qquad p_{tot} = p_M + p_T,
\label{E:FriedmannEq}
\end{eqnarray}
with $a_0=\left( 8 \pi G \right)^{-1}$ in GR, where $\rho_c$ and
$p_{tot}$ denote the critical energy density and total pressure of
the universe, while $\rho_T$ and $p_T$ correspond to the energy
density and pressure of some effective field, respectively. By
comparing the equation of motion of the scalar-torsion mode in
PGT~(\ref{E:rho_T}) to the Friedmann
equations~(\ref{E:FriedmannEq}), one obtains
\begin{eqnarray}\label{E:rho,p_T}
\rho_T &=&  3\mu H^2 - \frac{b}{18} \left( R + \frac{6\mu}{b}
\right)
(3H - \Phi )^2 + \frac{b}{24} R^2,\nonumber\\
p_T &=& \frac{1}{3} \left( \mu( R - \bar{R} ) + \rho_T \right),
\end{eqnarray}
which will be regarded as the torsion dark energy density and pressure,
respectively.
\begin{equation}\label{E:conservation_rho_c}
\dot{\rho}_c + 3H \left( \rho_c + p_{tot} \right) =0,
\end{equation}
which can also be derived by applying the identity
\[
\bar{\nabla}_j \bar{G}^{ij} = \bar{\nabla}_j \left( \bar{R}^{ij} -
\frac{1}{2} \bar{R} g^{ij} \right)= \bar{\nabla}_j \left(
\mathcal{T}^{ij}+\mathcal{T}_T^{ij} \right) =0,
\]
where $\bar{\nabla}$ is the covariant derivative with respect to the
Levi-Civita connection and $\mathcal{T}_{T \  j}^{\ i}=diag\left(
-\rho_T, p_T, p_T, p_T \right)$ is the effective energy-momentum
tensor of the torsion dark energy.

In addition, from (\ref{E:main eq1}) -- (\ref{E:main eq3}), one can
check that the continuity equation for the torsion field is also
valid, $i.e.$
\begin{equation}
\label{E:conservation_rho_T} \dot{\rho}_T + 3 H \left( \rho_T+ p_T
\right) =0.
\end{equation}
Consequently, we obtain the continuity equation for the ordinary
matter to be
\begin{equation}\label{E:conservation_rho_M}
\dot{\rho}_M + 3 H \left( \rho_M+ p_M \right) =0.
\end{equation}
By assuming no coupling between radiation and dust, the matter
densities of radiation ($w_r=1/3$) and dust ($w_m = 0$) in
scalar-torsion cosmology share the same evolution behaviors as in GR,
$i.e.$ $\rho_r \propto a^{-4}$ and $\rho_m \propto a^{-3}$,
respectively. In order to investigate the cosmological evolution, it
is natural to define the total EoS by~\cite{DE}
\begin{eqnarray}
w_{tot} = -1 - \frac{2 \dot{H}}{3 H^2} =\frac{p_{tot}}{\rho_c},
\end{eqnarray}
which leads to
\begin{eqnarray}
\label{E:w_T}
 w_{tot} =\Omega_M w_M + \Omega_T w_T,
\end{eqnarray}
where $\Omega_\alpha=\rho_\alpha/\rho_c$ and $w_\alpha =
p_\alpha/\rho_\alpha$ with $\alpha=M,T$, representing the energy
density ratios and EoSs of matter and torsion, respectively. Note
that the EoS in (\ref{E:w_T}), which is commonly used in the
literature, e.g.~\cite{DE}, can be determined from the cosmological observations in~\cite{obs1,obs11,obs12,obs13,obs14}. In particular,
it can be used to distinguish the modified gravity theories from the
$\Lambda$CDM~\cite{DE}.

Consequently, the evolution of the torsion dark energy can be
described solely in terms of $w_T$ by
\begin{eqnarray}
\rho_T(z) = \rho_T^{(0)} \exp \left\{ 3 \int_0^z dz^{\prime}
\frac{1+w_T(z^{\prime})}{1+z^{\prime}} \right\}.
\end{eqnarray}

In the following sections, we focus on this important quantity.


\section{Numerical Results of Torsion Cosmology}
The evolution of torsion cosmology is determined by (\ref{E:main
eq1}) -- (\ref{E:main eq3}). In general, one needs to solve the
dynamics of $R$, $\Phi$ and $H$ by the system of ordinary
differential equations. However, one easily sees that in
(\ref{E:main eq3}) there exists a special case: the
constant scalar affine curvature solution, $R=-6\mu/b$~\cite{Shie:2008ms}. Recall that in order to conform with
the positive kinetic energy argument, the condition
(\ref{E:condition}) is required. Such a condition leads to the negative curvature $R=-6\mu/b <0$ in this case with a negative matter
density $\rho<0$,
the condition of $a_1 < -a_0 <0$ is required~\cite{Shie:2008ms}.

We concentrate on the EoS of the scalar-torsion mode in both special
and normal cases which we define later.
  We also present the cosmological evolution of
the density ratio, defined by $\Omega=\rho/\rho_c$, from a high
redshift to the current stage.


\subsection{Special Case: $R= const.$}\label{sec:specialcase}
In this special case, we take the assumption of $a_1< 0$, $\mu<0$
and $a_0 > 0$ in~\cite{Shie:2008ms}. The evolution equations
(\ref{E:main eq1}) -- (\ref{E:main eq3}) reduce to
\begin{eqnarray}
\label{eq:sc01}
\rho_M &=& - 3 a_1 H^2 - \frac{3}{2}\frac{\mu^2}{b}, \\
\label{eq:sc02}
\rho_T &=& \frac{3}{2}\frac{\mu^2}{b} + 3 \mu H^2, \\
\label{eq:sc03} \dot{H} &=& -\left( 1+w_M \right)
\left(\frac{3}{4}\frac{\mu^2}{a_1 b}+\frac{3}{2}H^2\right).
\end{eqnarray}

To employ numerical calculation, we rescale the parameters as below:
\begin{eqnarray}
\label{eq:screscaling}
 m^2& =&\rho_m^{(0)}/3 a_0\,, \quad
 \widetilde{a}_0 = a_0/m^2 b, \quad \widetilde{a}_1 = -a_1/m^2 b, \nonumber \\
\widetilde{t}&= & m \cdot t, \quad \quad \widetilde{\mu}=
\widetilde{a}_1-\widetilde{a}_0, \quad
 \widetilde{H}^2 =H^2/m^2, \quad \widetilde{R}=R/m^2,
\end{eqnarray}
where $\rho_m^{(0)}$ is the matter density at $z=0$ and the scalar
affine curvature is $\widetilde{R}= 6 \widetilde{\mu} >0$. From
(\ref{eq:sc01}), (\ref{eq:sc02}) and (\ref{eq:sc03}), we obtain the
following dimensionless equations,
\begin{eqnarray}
\label{eq:scHm2}
&&\widetilde{H}^2= \frac{\widetilde{a}_0}{\widetilde{a}_1}\left( a^{-3}+\chi a^{-4} \right) + \frac{\widetilde{\mu}^2}{2\widetilde{a}_1}, \\
\label{eq:scrhoT}
&&\frac{\rho_T}{\rho_m^{(0)}}= \frac{ \widetilde{\mu}^2}{2 \widetilde{a}_0}- \frac{ \widetilde{\mu}}{\widetilde{a}_0} \widetilde{H}^2, \\
\label{eq:scdotH} &&\widetilde{H}{\widetilde{H}}^{\prime}= \left( 1+w_M
\right) \left(\frac{3}{4} \frac{\widetilde{\mu}^2}{\widetilde{a}_1}-
\frac{3}{2} \widetilde{H}^2\right)\,,
\end{eqnarray}
where the prime ``$ \prime$''.stands for  $d/d\ln a$ and
$\chi=\rho_r^{(0)}/\rho_m^{(0)}$. Using
(\ref{E:conservation_rho_T}),
(\ref{E:w_T}) and (\ref{eq:scrhoT}), we find that
\begin{eqnarray}
\label{eq:scTEOS} w_T = -1 - \frac{\dot{\rho}_{T}}{3H \rho_T}= -1
-\frac{4}{3} \frac{\dot{\widetilde{H}}}{2\widetilde{H}^2-  \widetilde{\mu}}.
\end{eqnarray}

\begin{center}
\begin{figure}[tbp]
\begin{tabular}{ll}
\begin{minipage}{80mm}
\begin{center}
\unitlength=1mm \resizebox{!}{6.5cm}{\includegraphics{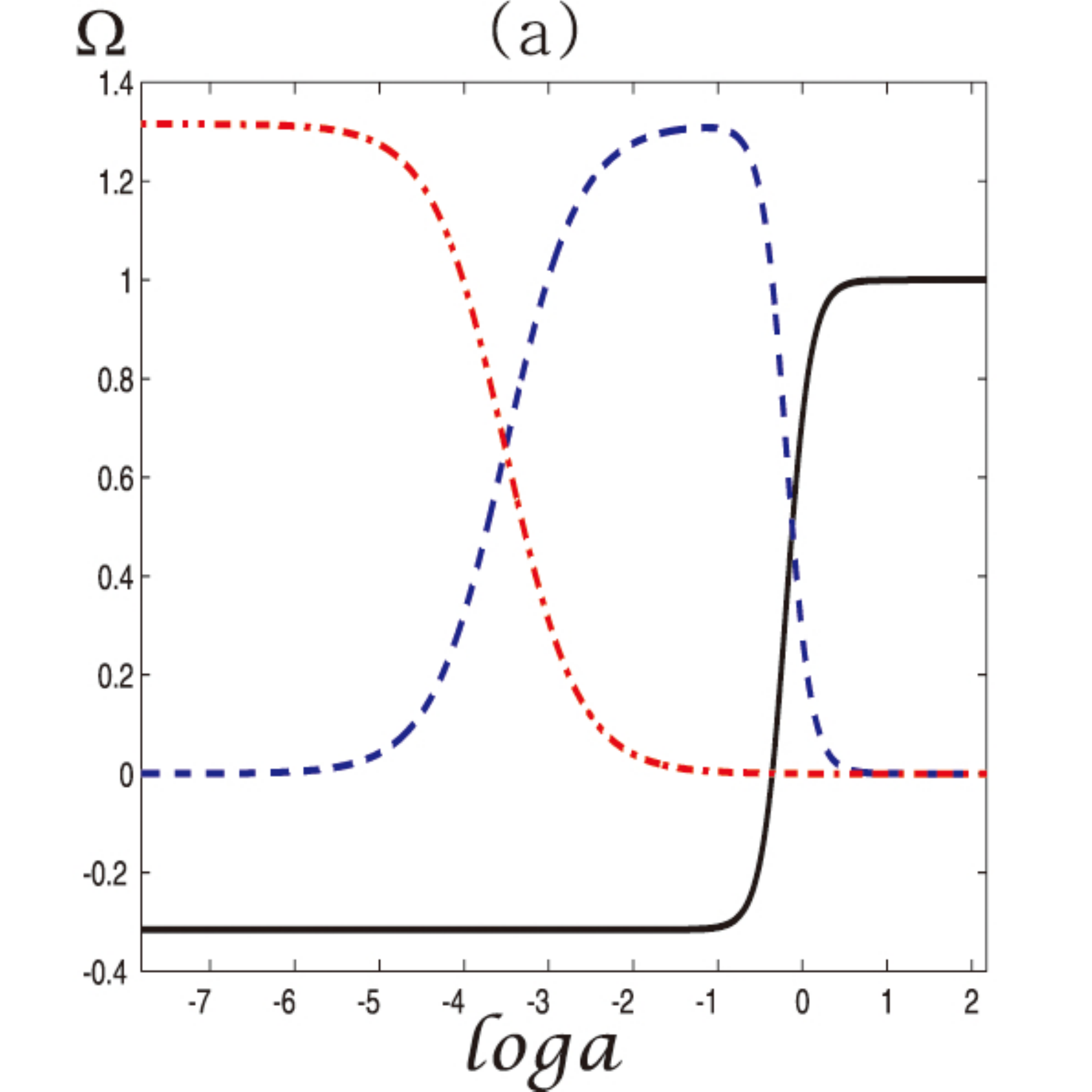}}
\end{center}
\end{minipage}
&
\begin{minipage}{80mm}
\begin{center}
\unitlength=1mm \resizebox{!}{6.5cm}{\includegraphics{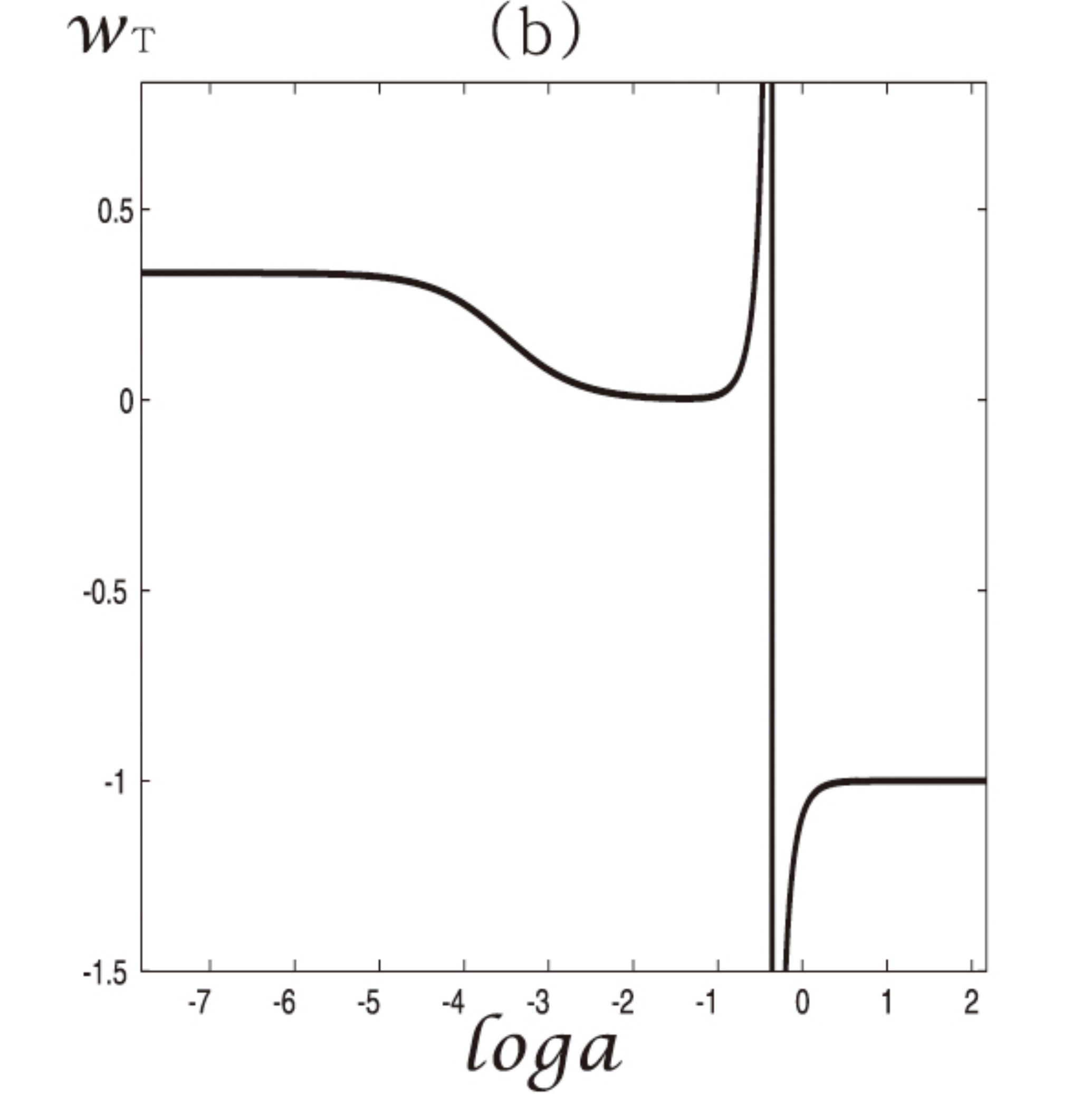}}
\end{center}
\end{minipage}\\[5mm]
\end{tabular}
\caption{ Evolutions of (a) the energy density ratio $\Omega$  and
(b) the torsion EoS $w_T$ with $\Omega_m^{(0)}=27.5\%$, where the
solid (black), dashed (blue), and dotted-dashed (red) lines stand
for torsion, matter and radiation, respectively. } \label{fg:1}
\end{figure}
\end{center}

From (\ref{eq:scHm2})--(\ref{eq:scTEOS}), it is easy to see that the
evolution of $\rho_T$ is automatically determined without solving
any differential equation for given values of $\widetilde{a}_0$ and
$\widetilde{a}_1$.
 The numerical results of this special case are shown in
Fig.~\ref{fg:1}, where we have chosen $\widetilde{a}_0=76$,
$\widetilde{a}_1=100$ and $\chi=3.07 \times 10^{-4}$ corresponding to
$\Omega_m^{(0)}=\widetilde{H}^{-2}_{z=0}\simeq 27.5\%$. In
Fig.~\ref{fg:1}a, we plot the energy density ratios of torsion,
matter and radiation, $\Omega_T$, $\Omega_m$ and $\Omega_r$,
respectively. Notice that $\rho_T$ depends on the parameters
$\widetilde{a}_0$ and $\widetilde{a}_1$, and there exists a late-time
de-Sitter solution when $\widetilde{H}^2= \widetilde{\mu}^2 / 2
\widetilde{a}_1$. In the high redshift regime, in which $\widetilde{H}^2 \gg
\widetilde{\mu}, \widetilde{\mu}^2/ \widetilde{a}_1$, we observe that the
torsion density ratio $\Omega_T$ is a constant which can also be
estimated from (\ref{eq:sc01}) and (\ref{eq:sc02}), namely
\begin{eqnarray}
\label{eq:scMoT} \frac{\rho_M}{\rho_T} = \frac{3 \widetilde{a}_1
\widetilde{H}^2 - 3\widetilde{\mu}^2/2}{3\widetilde{\mu}^2/2 - 3 \widetilde{\mu}
\widetilde{H}^2} \simeq -\frac{\widetilde{a}_1}{\widetilde{\mu}},
\end{eqnarray}
which manifests itself as a negative constant. In Fig.~\ref{fg:1}b, we
show that the torsion EoS $w_T$
  acts  as matter $w_m=0$ and radiation $w_r=1/3$
in the matter-dominant  ($\rho_m \gg \rho_r$) and
radiation-dominant  ($\rho_r \gg \rho_m$) stages, respectively,
which are interesting asymptotic behaviors. We also observe that in
the low redshift regime of $\log \, a \simeq 0$, $w_T$ is smaller
than unity, indicating the existence of a late-time acceleration
epoch.

\subsection{The Normal Case}\label{sec:normalcase}

The normal case corresponds to both the
kinetic energy and the matter density being positive, i.e, the parameters $a_0$, $a_1$
and $b$ are subject to the condition (\ref{E:condition}). It is also
convenient to rescale the parameters in the form
\begin{eqnarray}
\label{eq:norescaling} \widetilde{a}_0 &=& a_0/m^2b, \quad \widetilde{a}_1
=a_1/m^2b, \quad
\widetilde{t} =t \cdot m, \quad  \widetilde{\mu} =\widetilde{a}_0 + \widetilde{a}_1, \nonumber \\
\widetilde{H}^2& =&H^2/m^2, \quad \widetilde{\Phi} =\Phi/m, \quad
\widetilde{R}=R/m^2,
\end{eqnarray}
where $m^2=\rho_m^{(0)}/3 a_0$. Using the above rescaling
parameters, (\ref{E:main eq1}) -- (\ref{E:main eq3}) and
(\ref{E:rho_T}) are then rewritten as
\begin{eqnarray}\label{E:EOMrescaled}
\label{eq:no1}
&&\widetilde{H}{\widetilde{H}}^{\prime}= \frac{\widetilde{\mu}}{6\widetilde{a}_1} \widetilde{R} - \frac{\widetilde{a}_0}{2\widetilde{a}_1} a^{-3} -2\widetilde{H}^2, \\
\label{eq:no2}
&&\widetilde{H}{\widetilde{\Phi}}^{\prime}=\frac{\widetilde{a}_0}{2\widetilde{a}_1}\left( \widetilde{R} -3a^{-3} \right) -3\widetilde{H}\widetilde{\Phi} + \frac{1}{3}\widetilde{\Phi}^2,\\
\label{eq:no3}
&&\widetilde{H}{\widetilde{R}}^{\prime}= - \frac{2}{3} \left( \widetilde{R} + 6 \widetilde{\mu} \right) \widetilde{\Phi},\\
\label{eq:no4} &&\frac{1}{18}\left( \widetilde{R} + 6 \widetilde{\mu}
\right) \left( 3 \widetilde{H} - \widetilde{\Phi} \right) -
\frac{\widetilde{R}^2}{24} -3\widetilde{a}_1 \widetilde{H}^2 = 3 \widetilde{a}_0
\left( a^{-3}+\chi a^{-4} \right),
\end{eqnarray}
where we have used $\overset{em}{T}=3P_M-\rho_M= - \rho_m = - 3 a_0
m^2 a^{-3}$ due to $w_r = p_r / \rho_r = 1/3$ and $w_m = p_m /
\rho_m = 0$.
 From $(\ref{E:w_T})$ and (\ref{eq:no1})--(\ref{eq:no4}),
we have
\begin{eqnarray}
\label{eq:no_w} w_T = \frac{1}{3} \,  \frac{\widetilde{\mu} \left(
\widetilde{R}- \bar{R}/m^2 \right)}{3\widetilde{\mu} \widetilde{H}^2 - \left(
\widetilde{R} + 6 \widetilde{\mu} \right) \left( 3 \widetilde{H} - \widetilde{\Phi}
\right)^2 / 18 + \widetilde{R}^2/24} +\frac{1}{3}.
\end{eqnarray}

To perform the numerical computations, we need to specify two
parameters: $\widetilde{a}_0$ and $\widetilde{a}_1$, along with two initial
conditions: $\widetilde{R}$ and $\widetilde{H}$. Thus, the initial condition
for $\widetilde{\Phi}$ is automatically determined by (\ref{eq:no4}).
The numerical results are shown in Fig.~\ref{fg:2}, where the
initial conditions at $z=0$ are set as $(\widetilde{a}_0, \widetilde{a}_1,
\widetilde{R}_0, \widetilde{H}_0)
 = (2, 1, 14, 2), (2, 1, 13, 2), (3, 1, 8, 2)$ for
solid, dot-dashed, and dashed lines, respectively. Note that
$\chi=3.07 \times 10^{-4}$ originates from the WMAP-5 data, and
$\widetilde{H}=2$ corresponds to $\Omega_m^{(0)} =\widetilde{H}^{-2}_0 =
0.25$.

\begin{center}
\begin{figure}[tbp]
\begin{tabular}{ll}
\begin{minipage}{80mm}
\begin{center}
\unitlength=1mm \resizebox{!}{6.5cm}{\includegraphics{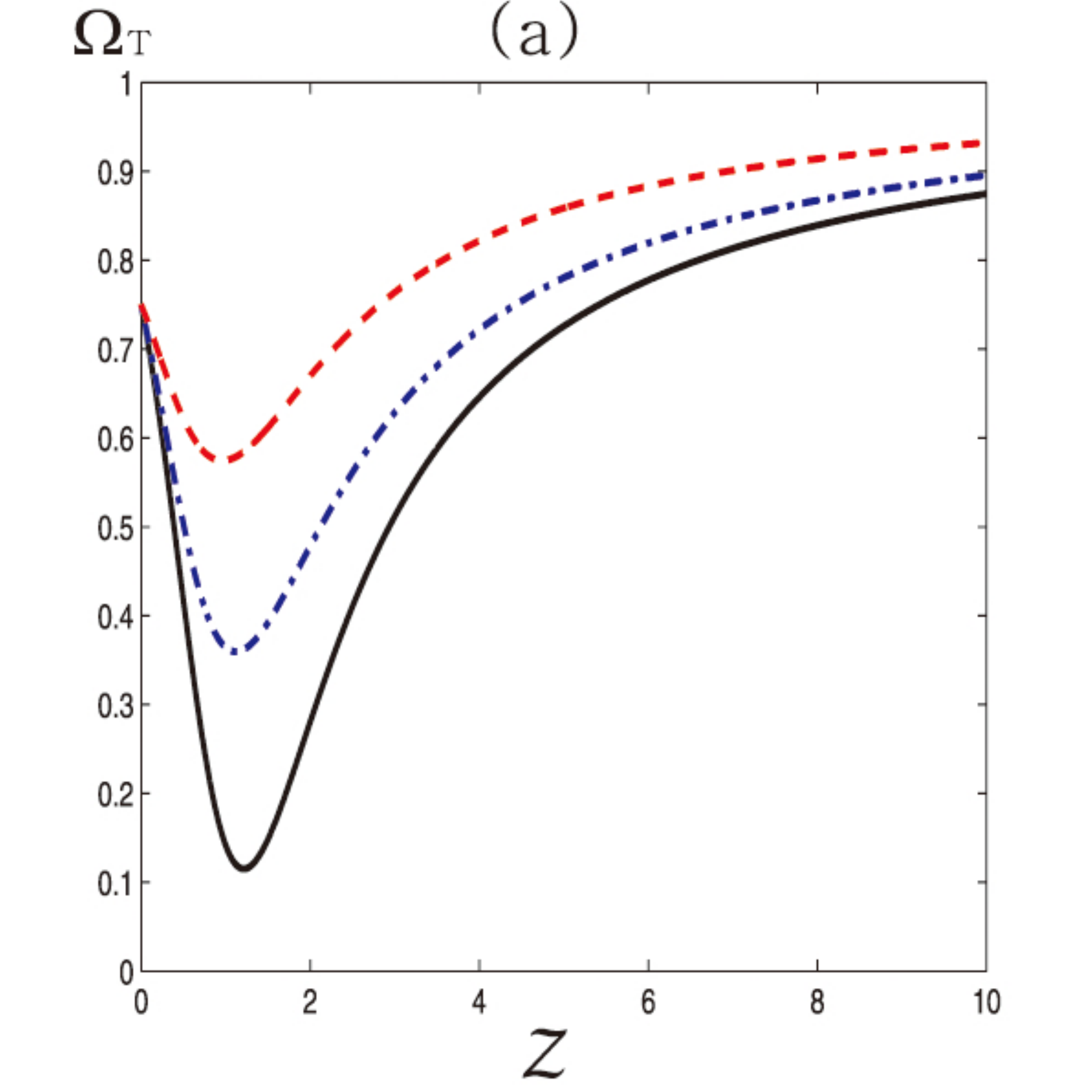}}
\end{center}
\end{minipage}
&
\begin{minipage}{80mm}
\begin{center}
\unitlength=1mm \resizebox{!}{6.5cm}{\includegraphics{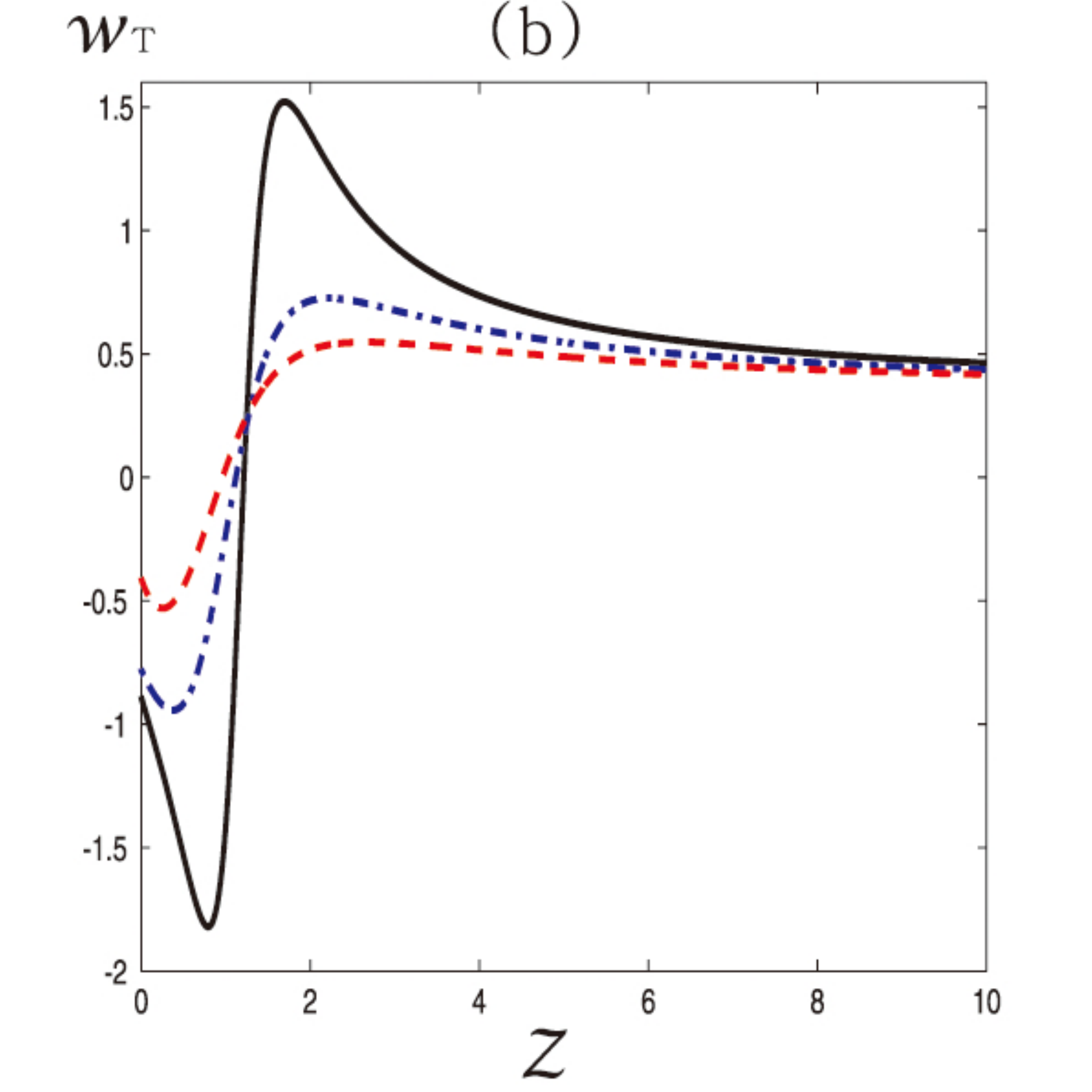}}
\end{center}
\end{minipage}\\[5mm]
\end{tabular}
\caption{Evolutions of (a) the energy density ratio $\Omega_T$ and
(b) the torsion EoS $w_T$ in the universe as functions of the
redshift $z$ with $\Omega_m^{(0)}=25\%$ and $\chi=3.07 \times
10^{-4}$, where the solid, dotted-dashed and dashed lines correspond
to
 $ (\widetilde{a}_0, \widetilde{a}_1, \widetilde{R}_0, \widetilde{H}_0) = (2, 1, 14, 2), (2, 1, 13, 2), (3, 1, 8, 2)$, respectively.}
\label{fg:2}
\end{figure}
\end{center}

In  Fig.~\ref{fg:2}a, we  show the evolution of the density ratio,
$\Omega_T = \rho_T / \rho_c$, as a function of the redshift $z$. The
figure demonstrates that the torsion density $\rho_T$ dominates the
universe in the high redshift regime ($z \gg 1$) with the general
parameter and initial condition selection, while the
matter-dominated regime is reached only within a very short time
interval. In Fig.~\ref{fg:2}b, we show that $w_T$ has an asymptotic
behavior at the high redshift regime, $i.e.$ $ w_{z \gg 0}
\rightarrow 1/3$. Moreover, in the low redshift regime, it may even
have a phantom crossing behavior, $i.e.$, the torsion EoS could
cross the phantom divide line of $w_T=-1$. As a result, the
scalar-torsion mode is able to account for the late-time
accelerating universe. We also notice that from the numerical
illustrations of the normal case one observes asymptotic behavior in
the high red shift regime. Compared to
Refs.~\cite{Shie:2008ms,Ao:2010mg} only oscillating behavior are
indicated. However, we find such asymptotic behavior in the normal case
is not only generic, but rather a property that can be proved
universally. Below we give a further account.


\section{Asymptotic Behavior of High Redshift}

\subsection{Semi-analytical solution in high redshift}\label{sec:AnalyticalSolution}

In the following, we shall provide the semi-analytical solution for
the positive energy case of the scalar-torsion mode in the large scalar affine
curvature limit $R \gg 6\mu/b$ which is commonly achieved in the
high redshift regime $(a \ll 1)$. In such circumstances, we may write
the energy density of ordinary matter and torsion in series
expansions of $a(t)$  as
\begin{eqnarray}
\label{E:rho_exp}
\rho_M &=& \frac{\rho_m^{(0)}}{a^3} +  \frac{\rho_r^{(0)}}{a^4},
\nonumber\\
\frac{\rho_T}{\rho_m^{(0)}} &= &\sum_{k=-c}^{\infty} A_{-k} \,
a^k\,,
\end{eqnarray}
respectively. Before the analysis, we adopt the rescaling
(\ref{eq:norescaling}) and write (\ref{E:EOMrescaled}) as functions
of $\widetilde{t}$ as
\begin{eqnarray}
\label{E:main eq1-n/dim}
&&\frac{d \widetilde{H}}{d\widetilde{t}}  = \frac{\widetilde{\mu}}{6 \widetilde{a}_1}  \widetilde{R} - \frac{\widetilde{a}_0}{2\widetilde{a}_1 \, a^3} -2 \widetilde{H}^2,\\
\label{E:main eq2-n/dim} && \frac{d\widetilde{\Phi}}{d\widetilde{t} } =
\frac{\widetilde{a}_0}{2 \widetilde{a}_1} \left( \widetilde{R} - \frac{3}{a^3}
\right) - 3 \widetilde{H} \widetilde{\Phi} + \frac{1}{3} \widetilde{\Phi}^2,\\
\label{E:main eq3-n/dim} && \frac{d \widetilde{R}}{d\widetilde{t}} \simeq
-\frac{2}{3} \widetilde{R} \, \widetilde{\Phi}, \\
\label{E:main eq4-n/dim} &&\frac{\widetilde{R}}{18} \left( 3\widetilde{H} -
\widetilde{\Phi} \right) - \frac{\widetilde{R}^2}{24} - 3 \widetilde{a}_1
\widetilde{H}^2 = 3\widetilde{a}_0 \left( \frac{1}{a^3} + \frac{\chi}{a^4}
\right),
\end{eqnarray}
In (\ref{E:main eq3-n/dim}), we have taken the approximation of $R
\gg 6\mu/b$ for the high redshift regime. With the above rescaling,
we shall argue that the lowest order of $\rho_T$ does not exceed
$a^{-4}$ in the following discussion. We formulate the statement as
a theorem.

\begin{theorem}
In the high redshift regime $(a \ll 1)$, $\rho_T = O(a^{-4})$.
\end{theorem}

\begin{proof}

First we expand
\begin{equation}\label{E:H^2}
\widetilde{H}^2(t) = \sum^{\infty}_{k=-c} r_k \, a^{k-4},  \qquad ( r_k
< \infty   )
\end{equation}
where $c$ is some integer, so that we have
\begin{equation}\label{E:dH}
\frac{d \widetilde{H}}{d\widetilde{t}} = \sum^{\infty}_{k=-c} \left(
\frac{k-4}{2} \right) r_k \, a^{k-4}.
\end{equation}
Using (\ref{E:main eq1-n/dim}), (\ref{E:main eq3-n/dim}),
(\ref{E:H^2}) and (\ref{E:dH}), we obtain
\begin{eqnarray}
\label{E:tilde_R} \widetilde{R} &=& \frac{3 \widetilde{a}_1 }{\widetilde{\mu}}
\left( \sum^{\infty}_{k=-c} k \cdot r_k \,
a^{k-4} \right) + \frac{3 \widetilde{a}_0 }{\widetilde{\mu} \, a^3},\\
\label{E:tilde_Phi} \widetilde{\Phi} &=& - \frac{3}{2} \widetilde{H}  \cdot
\frac{\widetilde{a}_1 \left( \sum^{\infty}_{k=-c} k(k-4) r_k \, a^{k-4}
\right) - \frac{3\widetilde{a}_0}{a^3} }{\widetilde{a}_1 \left(
\sum^{\infty}_{k=-c} k \, r_k \, a^{k-4} \right) +
\frac{\widetilde{a}_0}{a^3}}\,,
\end{eqnarray}

Substituting (\ref{E:H^2}), (\ref{E:dH}), (\ref{E:tilde_R}) and
(\ref{E:tilde_Phi}) into (\ref{E:main eq2-n/dim}), and  comparing
the lowest power (requiring $c > -1 $, otherwise losing its leading
position) of $a$ in the high redshift, $a \ll 1$, we derive the
following relation
\begin{equation}\label{E:c}
\left( \frac{\widetilde{a}_1}{\widetilde{\mu}} \right)^2 c^2 \cdot r_{-c}^3
\left[ c^2 + (5-\frac{\widetilde{a}_0}{\widetilde{\mu}}) c + 4  \right] =0,
\end{equation}
which leads to $r_{-c}=0$ if  $c \geq 1$
as $0 <\widetilde{a}_0/\widetilde{\mu} < 1$ and $c^2
+(5-\frac{\widetilde{a}_0}{\widetilde{\mu}}) c + 4 \neq 0$.
This is equivalent to saying that (\ref{E:H^2}) has the form
\begin{equation}
\widetilde{H}^2 = \frac{r_0}{a^4} + \frac{r_1}{a^3} + \frac{r_2}{a^2} +
\frac{r_3}{a} + r_4 + \cdots
\end{equation}

Finally, we achieve our claim from (\ref{E:Friedmann s-t mode}) that
\begin{eqnarray}
\label{E:rho_T} \frac{\rho_T}{\rho_m^{(0)}} &= &- \left(
\frac{\chi}{a^4} + \frac{1}{a^3} \right) + \widetilde{H}^2
\nonumber\\
      & &= -\left( \frac{\chi}{a^4} + \frac{1}{a^3} \right) + \left( \frac{r_0}{a^4} + \frac{r_1}{a^3} + \frac{r_2}{a^2} + \frac{r_3}{a} + r_4 + \cdots
       \right) = O\left(\frac{1}{a^4}\right).
\end{eqnarray}
Note that the last equality follows since $r_0\neq \chi$, which will
be explained later.
\end{proof}

We now write the expansion in (\ref{E:rho_T}), by the theorem above,
simply as
\begin{equation}\label{eq:rho_T}
\frac{\rho_T}{\rho_m^{(0)}} = \sum_{k=-4}^{\infty} A_{-k} \, a^k\,.
\end{equation}
We shall only take first few dominating terms for a sufficient
demonstration. By the procedure in the proof of the theorem, we can
as well compare terms of various orders to yield the following
relations:
\begin{eqnarray}
\label{eq:sol_pow3}
O(a^{-10})&:& 3 \left( A_4 + \chi \right) \left(1+ \frac{\widetilde{a}_1}{\widetilde{\mu}} \, A_3  \right)^2 =0, \\
\label{eq:sol_pow4}
O(a^{-9})&:& 2 \left( 1+ \frac{\widetilde{a}_1}{\widetilde{\mu}} A_3 \right) \left[ \frac{\widetilde{a}_0}{\widetilde{\mu}} \left( 1+ \frac{\widetilde{a}_1}{\widetilde{\mu}} A_{3} \right) A_{3} + \frac{4 \widetilde{a}_1}{\widetilde{\mu}} \left( A_4 + \chi \right) A_{2} \right] =0,\\
\label{eq:sol_pow5} O(a^{-8})&:&  \left( 1+
\frac{\widetilde{a}_1}{\widetilde{\mu}} A_3 \right)  \left[ 4
\frac{\widetilde{a}_0}{\widetilde{\mu}} \left( 1+ 3
\frac{\widetilde{a}_1}{\widetilde{\mu}} A_3 \right) A_2 \right. \nonumber\\
&&\left.- \left( 3 A_2 +\frac{\widetilde{a}_1}{\widetilde{\mu}} \left( A_2
\left( 2 + 5 A_3 \right) -18 A_1 \left( A_4+ \chi \right) \right)
\right) \right] = 0.
\end{eqnarray}
 From  (\ref{eq:sol_pow3}), (\ref{eq:sol_pow4}) and
(\ref{eq:sol_pow5}), one concludes a relation,
\begin{equation}\label{eq:sol_a3}
A_3 = -\frac{\widetilde{\mu}}{\widetilde{a}_1} = - \frac{\left( \widetilde{a}_0
+ \widetilde{a}_1 \right)}{ \widetilde{a}_1} <-1,
\end{equation}
with $A_1,A_2$ and $A_4$ left as arbitrary constants to be
determined by initial conditions and (\ref{E:main eq4-n/dim}).
 Note that (\ref{eq:sol_a3}) implies $r_1= -\widetilde{a}_0/ \widetilde{a}_1 < 0$
 in (\ref{E:H^2}). However, due to the
observational data that  $a=1$ at the current stage, the radiation
density is much smaller than the dust density ($\chi \ll 1$),
whereas the
torsion density is the same order as the dust density, as seen from
(\ref{E:rho_T}),
\begin{equation}
\frac{\rho^{(0)}_T}{\rho^{(0)}_m} = \left[ (r_0 - \chi) + r_2 +
\cdots \right] - ( 1+ |r_1|) \simeq O(1)\,.
\end{equation}
Subsequently, we have that $ [(r_0 - \chi) + r_2 + \cdots] \leq
\max\{ O(1), O(|r_1|) \}$, along with the assumption $r_k < \infty $
for each $k$. As a result, we conclude that $r_k$, for all $k \neq
1$, should not be too large, which forbids the possibility $r_0 =
\chi$. This argument shows the validity of the last equality in
(\ref{E:rho_T}) with the non-vanishing $O(1/a^4)$ coefficient.

From (\ref{E:w_T}), via the continuity equation~\cite{Tseng:2012hz},
we obtain
\begin{eqnarray}
\label{eq:sol_wt} w_T=-1-\frac{\rho^{\prime}_T}{3\rho_T} \simeq -1 +
\frac{1}{3}  \left( \frac{4 A_4 a^{-4} + 3 A_3 a^{-3}}{A_4 a^{-4} +
A_3 a^{-3}} \right) \simeq \frac{1}{3}\left( 1 - \frac{A_3}{A_4}a
\right),
\end{eqnarray}
where the prime ``$\prime$'' stands for  $d/d\ln a$ and we have used
(\ref{eq:rho_T}) for $a \ll 1$.

\subsection{Numerical computations}\label{sec:Numerical}

In this subsection, we perform numerical computations to support the
analysis above. As an illustration, we take the parameters
$\widetilde{a}_0=2$ and $\widetilde{a}_1=1$ and initial conditions
\[
\widetilde{H}(z=0)=\widetilde{H}_0=2, \qquad \widetilde{R}(z=0) =
\widetilde{R}_0=14,
\]
and show the evolutions of $w_T$, $\widetilde{\Phi}$, and $\widetilde{R}$ in
Figs.~\ref{fg:1}, \ref{fg:2}a and \ref{fg:2}b, respectively.
\begin{center}
\begin{figure}[tbp]
\begin{tabular}{ll}
\begin{minipage}{80mm}
\begin{center}
\unitlength=1mm \resizebox{!}{6.5cm}{\includegraphics{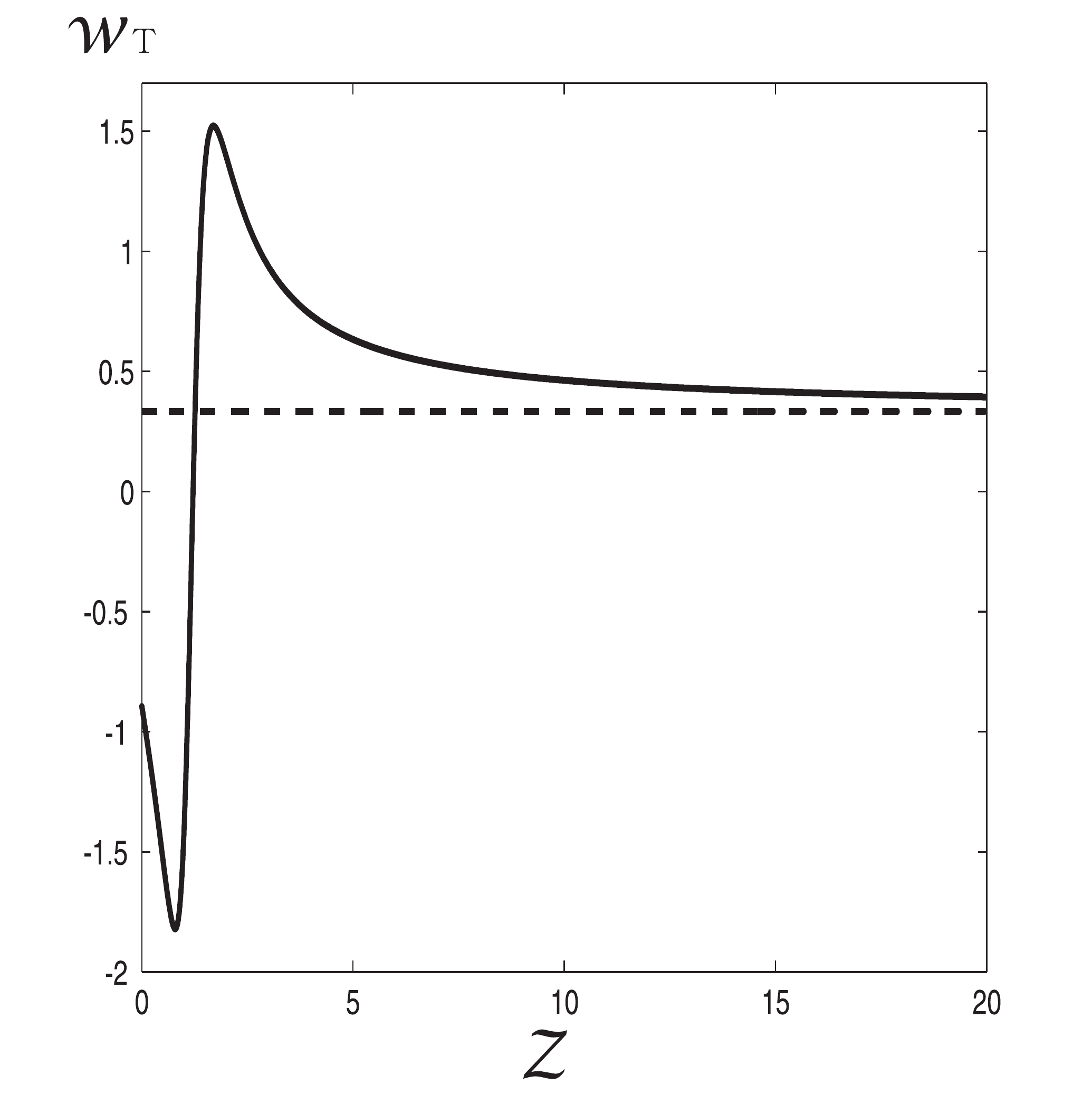}}
\end{center}
\end{minipage}
&
\begin{minipage}{80mm}
\begin{center}
\unitlength=1mm \resizebox{!}{6.5cm}{\includegraphics{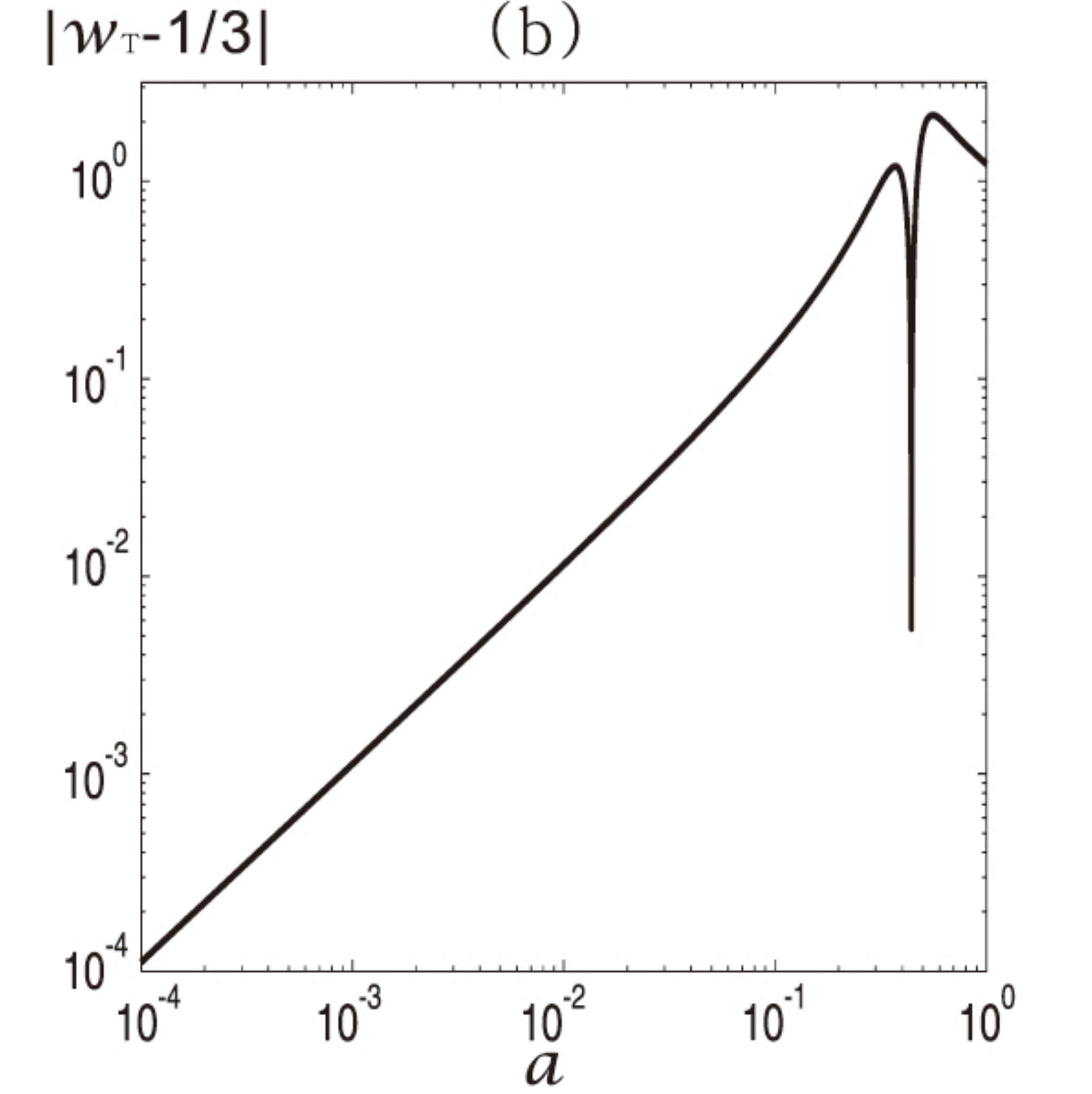}}
\end{center}
\end{minipage}\\[5mm]
\end{tabular}
\caption{Evolutions of (a)  $w_T$ and (b) $\lvert w_T-1/3 \rvert$ as
function of the redshift $z$  and the scale parameter $a$,
respectively, where the parameters and initial conditions are chosen
as $\widetilde{a}_0=2$, $\widetilde{a}_1=1$, $\widetilde{H}_0=2$,
$\widetilde{R}_0=14$ and  $\chi = \rho^{(0)}_r /\rho^{(0)}_m = 3.1
\times 10^{-4}$.} \label{fg:1}
\end{figure}
\end{center}

In Fig.~\ref{fg:1}a, we demonstrate the EoS of  torsion
 as a function of the redshift $z$. As seen from the figure,
in the high redshift regime $w_T$ approaches $1/3$, which indeed
shows an asymptotic behavior.
Fig.~\ref{fg:1}b indicates that $\lvert w_T-1/3 \rvert$ approximates
a straight line in the scale factor $a$ in the $\log$-scaled
coordinate since the slope in the $\log$-scaled coordinates is
nearly $1$. The singularity in the interval $[0.1,1]$ corresponds to
the crossing $1/3$ of $w_T$. Thus, the numerical results concur with
our semi-analytical approximation in (\ref{eq:sol_wt}). In
Fig.~\ref{fg:2}, we observe that the behaviors of $\widetilde{R}$ and
$\widetilde{\Phi} \propto 1/a^2$ in the high redshift regime are
consistent with the results in
(\ref{E:tilde_R}) and (\ref{E:tilde_Phi}), given by
\begin{eqnarray}
\label{eq:sol_R1}
&&\widetilde{R} \simeq  \frac{2 \widetilde{a}_1}{\widetilde{\mu}} \, A_2 \, a^{-2}, \\
\label{eq:sol_phi1} &&\widetilde{\Phi} \simeq 3\widetilde{H} \propto a^{-2},
\end{eqnarray}
respectively, where $\widetilde{H}^2 \simeq \left(\chi + A_4\right)
a^{-4}$ from (\ref{E:H^2}). Note that from (\ref{eq:sol_R1}), the
behavior of the affine curvature $\widetilde{R}$ is highly different
from that of the
 Riemannian scalar  curvature $\bar{R} = -
\mathcal{T}/a_0 = \rho_m/a_0$, which is proportional to $1/a^3$ in
both the matter (dust) and radiation dominated eras.
\begin{center}
\begin{figure}[tbp]
\begin{tabular}{ll}
\begin{minipage}{80mm}
\begin{center}
\unitlength=1mm \resizebox{!}{6.5cm}{\includegraphics{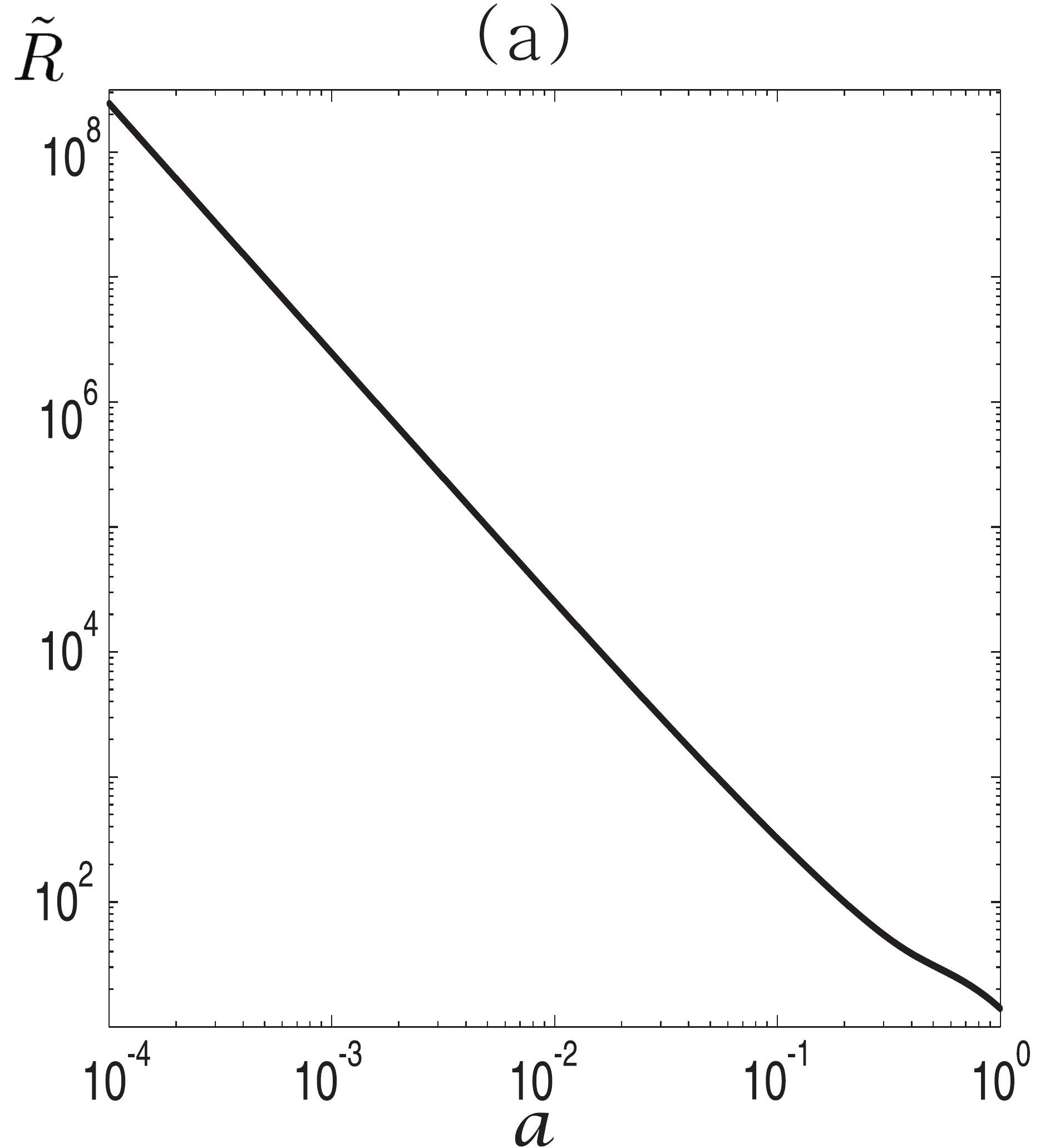}}
\end{center}
\end{minipage}
&
\begin{minipage}{80mm}
\begin{center}
\unitlength=1mm \resizebox{!}{6.5cm}{\includegraphics{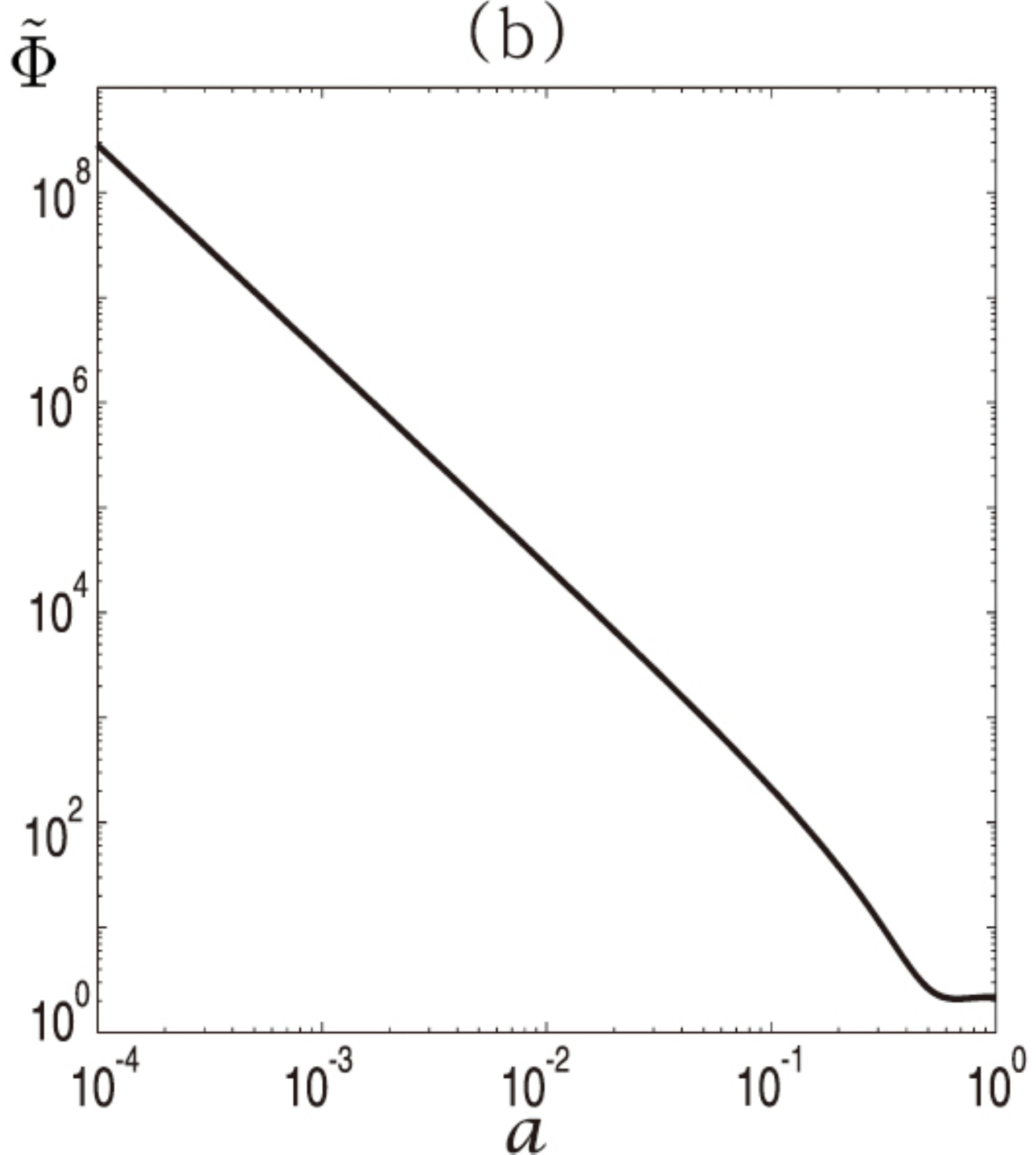}}
\end{center}
\end{minipage}\\[5mm]
\end{tabular}
\caption{ Evolutions of (a) the rescaled affine curvature
$\widetilde{R}$ and (b) the torsion $\widetilde{\Phi}$  as functions of the
scale parameter $a$ in the $\log$ scale with the parameters
and initial conditions taken to be the same as Fig.~\ref{fg:1}.}
\label{fg:2}
\end{figure}
\end{center}

\chapter{Teleparallel Gravity} \label{TEGR}

\begin{chapquote}{\textit{Friedrich W. Hehl}}
\ldots \textit{Torsion is not
just a tensor, but rather a very specific
tensor that is intrinsically related to the
translation group, as was shown by \'{E}lie Cartan in 1923-24.}
\end{chapquote}

\section{A Degenerate Theory of PGT}

From Chapter 2, we have seen that from the affine frame bundle theory where PGT resides, the gauge group $G= \mathbb{R}^{1,3} \rtimes
SO(1,3)$ results in the decomposition of the affine connection
$\g^*\widetilde{\omega} = \omega + \varphi $ according to the Lie algebra
$\mathfrak{g} = \mathbb{R}^{1,3} \oplus \mathfrak{so}(1,3)$, and from (\ref{E:AC curvature}) we see
\begin{equation}\label{E:R+T}
D^{\omega}\g^*\widetilde{\omega} = \underbrace{D^{\omega}\omega
}_{\mathfrak{so}(1,3)}+
\underbrace{D^{\omega}\varphi}_{\mathbb{R}^{1,3}}
\end{equation}
It is then natural to ponder the question whether there could be some
connection $\omega$ such that either $D^{\omega} \omega = 0$ or
$D^{\omega} \varphi = 0$ which takes only a 1-sided Lie-algebra
value.

This na\"{\i}ve question leads to two different gravitational
theories. Of course the case $D^{\omega} \varphi = T^{\omega}
= 0$ (vanishing torsion) is known as Riemannian geometry of GR; the other
one $D^{\omega} \omega = \Omega^{\omega} = 0$ (vanishing
curvature), is called the \textbf{teleparallel geometry} (also known as the \textbf{absolute parallelism}). It is also clear from the
Lie algebra (\ref{E:R+T}) (only the tail of the Lie-algebra value
$\mathbb{R}^{1,3}$ is left) to see the dubbed name
translational gauge theory. Thus GR and teleparallel gravity can be considered as two extreme yet complementary degenerate cases. For
clarity, we summarize as follows:
\begin{definition}{(Teleparallel geometry: bundle, vanishing curvature)}\label{Def:TG1}

On the orthonormal frame bundle $(F(M),\pi , M, SO(1,3))$, if there
exists a connection 1-form $\omega \in \Lambda^1(F(M),
\mathfrak{so}(1,3)) $ such that $D^\omega \omega \equiv 0$, then we
call the tuple $(M,g,\nabla)$ \textbf{teleparallel geometry}, where
the connection $\nabla$ on $M$ is induced by $D^{\omega}$ on $F(M)$.

\end{definition}

Gravitational theorists, who do not necessarily have to know bundle
theories, usually take an alternative version on the (projected)
base manifold.
\begin{definition}{(Teleparallel geometry: base manifold, Weitzenb\"{o}ck connection)}\label{Def:TG2}

Given a semi-Riemannian spacetime $(M,g)$ with a set of (local, not
necessarily) frame $ \{ e_\a \in TM \}$, $\a= 0,1,2,3$, and a
connection $\nabla^W$ on $M$ satisfying
\begin{equation}\label{E:Weitzenbock connection}
\nabla^W_{e_\a} e_\b = 0, \qquad \mbox{(for all $\a$, $\b$)}
\end{equation}
the tuple $(M,g,\nabla^W, \{e_\a\} )$ is defined as teleparallel geometry with respect to $\{ e_\a \}$, also called a \textbf{Weitzenb\"{o}ck spacetime} along with the connection $\nabla^W$ called the \textbf{Weitzenb\"{o}ck connection}
\end{definition}

The definition of a Weitzenb\"{o}ck connection automatically indicates all the
connection 1-forms vanish, $\omega_\a{}^{\b} := \vt^\b(\nabla^W e_\a)
\equiv 0$, and hence the curvature vanishes
\begin{equation}\label{E:zero R}
\Omega^{\mu}{}_{\nu} = d\omega_\m{}^\n + \omega_\n{}^\g \wedge
\omega_\g{}^\m  \equiv 0 ,\quad \Leftrightarrow \quad
R^{\mu}{}_{\nu\a\b} \equiv 0.
\end{equation}

We emphasize that two facts that are neglected sometimes:

\begin{remark}

\begin{enumerate}
\item
The curvedness or flatness is \emph{not} an absolute property of a spcetime
manifold, rather it changes with the connection assigned.

\item
Initially the flatness of teleparallel gravity is with respect to one
chosen frame, say $e_\a$, but since $R^{\mu}{}_{\nu\a\b} \equiv 0$ is a tensor equation, we have $R(X,Y)Z
= 0 $ for all $X,Y,Z \in TM $. In particular, if we choose
another frame $\widetilde{e}_\b $ for computation, we still have
curvature zero. Hence the parallelism is well-defined and independent of
any chosen frame as long as there exists one frame $e_\a$ such that
(\ref{E:Weitzenbock connection}) holds.

\end{enumerate}
\end{remark}

However, one natural question may arise: whether the two definitions
are equivalent or not? A closer look of the question reveals its
essence:
\begin{equation}\label{E:equivalence}
\begin{aligned}
D \omega_{\a}{}^\b := d\omega_\a{}^\b + \omega_\g{}^\b \wedge
\omega_\a{}^\g \equiv 0 \quad &\overset{?}{\Leftrightarrow} \quad
\omega_\a{}^{\b} \equiv 0\\
\mbox{(vanishing curvature, Def \ref{Def:TG1})} \quad
&\overset{?}{\Leftrightarrow} \quad \mbox{(Weitzenb\"{o}ck
connection, Def \ref{Def:TG2})}
\end{aligned}
\end{equation}

Obviously, the direction ($\Leftarrow$) is simply by (\ref{E:zero
R}). The other direction requires some effort to see and involves partial
differential equations in general. First we
consider that there are two frames $\{ \widetilde{e}_\a \}$ and
$\{e_\a \}$, where the latter satisfies the Weitzenb\"{o}ck
condition (\ref{E:Weitzenbock connection}). Since there exists a
linear isomorphism $ A_\a^\b (x) $ for all $x\in M$ such that
$\widetilde{e}_\a (x) = A_\a^\b (x) \, e_\b (x)$, then we may
calculate the Weitzenb\"{o}ck connection 1-form expressed in the
frame $\{ \widetilde{e}_\a \}$. Let $\nabla^W \widetilde{e}_\b :=
\widetilde{\o}_\b^\a \, \widetilde{e}_\a $, then we have
$\widetilde{\o}_\b^\a (X) = \widetilde{\vt}^\a \left( \nabla^W_X
\widetilde{e}_\b \right)$ where $\widetilde{\vt}^\a $ is the dual
basis of $\widetilde{e}_\a$, and thus
\begin{equation}
\widetilde{\o}_\b^\a (\widetilde{e}_{\mu})  = \widetilde{\vt}^\a
\left( \nabla^W_{\widetilde{e}_{\mu}} \widetilde{e}_\b  \right) =
\widetilde{e}_{\mu} (A_\b^{\nu}) (A^{-1})_{\nu}^\a, \quad
\Leftrightarrow \quad \widetilde{\o}_\b^\a = (dA_\b^{\nu}) \cdot(
A^{-1})_{\nu}^\a \neq 0.
\end{equation}
We see that the Weitzenb\"{o}ck connection 1-form
$\widetilde{\o}_\b^\a$ in the frame $\{\widetilde{e}_\a \}$ is in
general nonzero although $\omega_\b{}^\a \equiv 0$ in the frame $\{
e_\a \}$. The computation of changing frames helps answering the
question (\ref{E:equivalence}). Since if we start from any frame
$\{e_\a\} $ in Def. (\ref{Def:TG1}), to acquire the Weitzenb\"{o}ck
condition defined by (\ref{E:Weitzenbock connection}) we need to
solve $A_{\a}^{\b}(x)$ from the coupled partial differential
equations, namely
\begin{equation}
D \widetilde{\o}_{\a}{}^\b = d\widetilde{\o}_\a{}^\b +
\widetilde{\o}_\g{}^\b \wedge \widetilde{\o}_\a{}^\g \equiv 0
\end{equation}
then the existence of $A_{\a}^{\b}(x)$ at all $x\in M$ guarantees
the equivalence of the other direction.

\subsection{Another construction from PGT}

There is another construction that leads to the same teleparallel gravity
directly from PGT. This intuitive derivation can be found in
\cite{Hehl:Erice1979}. To require vanishing curvature
$\Omega^\b{}_\a \equiv 0$ for a gravitational theory, one can put
manually a Lagrangian multiplier two-form $\Lambda_{\a\b} \in
\Lambda^2(M)$ to attain the desired constraint from the Lagrangian such that
\begin{equation}\label{E:multiplier Lagrangian}
L_{\text{tot}} = L_{\text{G}} + L_{\text{mat}} - \frac{1}{2\varrho}
\Lambda_{\a\b} \wedge \Omega^{\a\b}.
\end{equation}
which leads to the PGT field equations (\ref{E:PGT EOM2}),
\begin{equation}\label{E:Teparallel EOM}
\begin{aligned}
D H_\a - e_\a \lr  L_{\text{G}} -  (e_\a \lr T^\b ) \wedge
H_\b &= \mathcal{T}_\a \\
\frac{1}{2\varrho} D \Lambda_{\a\b} + \vt_{[\a} \wedge H_{\b]} &=
\mathcal{S}_{\a\b}
\end{aligned}
\end{equation}
where one observes that $(\ref{E:Teparallel EOM})_2$ helps solve
the Lagrangian multiplier $\Lambda_{\a\b}$ while $(\ref{E:Teparallel
EOM})_1$ is free of $\Lambda_{\a\b}$ and thus defines the field equation
for teleparallel gravity, see \cite{B&H}.

So far in this section we introduced three approaches that define teleparallel geometry. However, we have not yet specified a
geometric Lagrangian to indicate the evolution of
Weitzenb\"{o}ck spacetime.  Below we introduce an interesting
Lagrangian equivalent to GR.


\section{The Teleparallel Equivalent to General Relativity (TEGR)}

The teleparallel equivalent to general relativity (TEGR, or
$\text{GR}_{||}$) is a special type of teleparallel gravity that
manifests equivalence to GR in spacetime and matter evolutions. First
we give a formal definition.
\begin{definition}{(TEGR)}\label{Def:TEGR}

\textbf{TEGR} is defined by a Weitzenb\"{o}ck spacetime
$(M,g,\nabla^W, \{e_\a\} )$ with the gravitational Lagrangian
(4-form)
\begin{equation}\label{E:teleparallel Lagrangian 1}
L_{\text{GR}_{||}} = - \frac{1}{2\k} T^\a \wedge \star \left( -
{}^{(1)}T_\a + 2 {}^{(2)}T_\a + \frac{1}{2} {}^{(3)}T_\a \right)
\end{equation}
and vanishing spin source of matter, $\mathcal{S}_{\a\b} \equiv 0$. Here ${}^{(I)}T_\a$ are the torsion irreducible pieces defined in
(\ref{E:torsion decomp}).

\end{definition}
With this special TEGR Lagrangian (\ref{E:teleparallel Lagrangian
1}), the translational and rotational excitations (\ref{E:excitations}) and gravitational currents (\ref{E:gravitational energy-momentum}) can be computed as
\begin{equation}\label{E:TEGR excitation}
\begin{aligned}
H_\a &= \frac{1}{\k} \star \left( - {}^{(1)}T_\a + 2 {}^{(2)}T_\a +
\frac{1}{2} {}^{(3)}T_\a \right), \quad H_{\a\b} = - \frac{\p \mathcal{L}_{G}}{\p \O^{\a\b}} =0,\\
t_\a &= i_{e_\a}(\mathcal{L}_{G}) + (i_{e_\a}( T^\b)) \wedge H_\b , \qquad  \qquad s_{\a\b} =\frac{\p \mathcal{L}_{G}}{\p \o^{\a\b}} =0,
\end{aligned}
\end{equation}
As a consequence, only field equation $(\ref{E:PGT EOM2})_1$ survives, given by (\ref{E:Teparallel EOM})
\begin{equation}\label{E:Teparallel EOM2}
dH_\a - t_\a = \k \mathcal{T}_\a
\end{equation}
where the second field equation degenerates much as in the Einstein-Cartan theory. In fact, the equivalence of TEGR to GR is reflected by two facts: the first one is given by as follows.
\begin{theorem}
The motion of matter (assumed spinless) in TEGR is
equivalent to GR (the proof demonstrated here is provided by Hehl).
\end{theorem}

\begin{proof}

Recall the contortion 1-form
\begin{equation}\label{E:contortion}
K_\a{}^\b := \widetilde{\o}_\a{}^\b - \o_\a{}^\b
\end{equation}
where $\widetilde{\o}_\a{}^\b$ is the Levi-Civita connection and
$\o_\a{}^\b$ is the Weitzenb\"{o}ck connection. From the fact that
\begin{equation}
T^\b = d \vt^\b + \o_\a{}^\b \wedge \vt^\a , \quad 0=  d \vt^\b +
\widetilde{\o}_\a{}^\b \wedge \vt^\a
\end{equation}
we can rewrite torsion 2-form in terms of the contortion
\begin{equation}\label{E:torsion identity}
T^\b = K^\b{}_\g \wedge \vt^\g
\end{equation}
We also notice the fact that $K^{\a\b} = - K^{\b\a}$ if both
$\widetilde{\o}_\a{}^\b $ and $ \o_\a{}^\b$ are metric compatible.
Now from the energy-momentum conservation law in $U_4$ for PGT, we
have
\begin{equation}\label{E:conserv1}
D\mathfrak{T}_\a= \left( e_\a \lr T^\b \right) \wedge
\mathfrak{T}_\b + \left( e_\a \lr \Omega^{\b\g} \right) \wedge
\mathfrak{S}_{\b\g} = \left( e_\a \lr K^\b{}_\g \right) \wedge
\vt^\g \wedge \mathfrak{T}_\b - K^\b{}_\a \wedge \mathfrak{T}_\b
\end{equation}
where we have used the Weitzenb\"{o}ck connection $\Omega^{\b\g}=0 $ and
(\ref{E:torsion identity}). On the other hand, the LHS of
(\ref{E:conserv1}) follows from the definition that
\begin{equation}\label{E:two D eq}
D\mathfrak{T}_\a :=  d\mathfrak{T}_\a - \o_\a{}^\b \wedge
\mathfrak{T}_\b := d\mathfrak{T}_\a - \left( \widetilde{\o}_\a{}^\b
- K_\a{}^\b \right) \wedge \mathfrak{T}_\b =
\widetilde{D}\mathfrak{T}_\a + K_\a{}^\b \wedge \mathfrak{T}_\b,
\end{equation}
where $\widetilde{D}$ denotes the Levi-Civita connection. When we
equate (\ref{E:conserv1}) and (\ref{E:two D eq}), we have
\begin{equation}
\widetilde{D}\mathfrak{T}_\a + \cancel{K_\a{}^\b \wedge
\mathfrak{T}_\b} = \left( e_\a \lr K^\b{}_\g \right) \wedge \vt^\g
\wedge \mathfrak{T}_\b - \cancel{K^\b{}_\a \wedge \mathfrak{T}_\b}
\end{equation}
where the two terms cancel due to the fact $K_{\a\b} = - K_{\b\a}$,
so that we obtain
\begin{equation}
\widetilde{D}\mathfrak{T}_\a = \left( e_\a \lr K^{\b\g} \right)
\wedge \vt_{[\g} \wedge \mathfrak{T}_{\b]}.
\end{equation}
Recall that in TEGR, we assume only spinless matter $\mathfrak{S}_{\g\b} =0$ which has the
identity $\vt_{[\g} \wedge \mathfrak{T}_{\b]}=0 $. Then we derive
$\widetilde{D}\mathfrak{T}_\a = 0$, which is the conservation law of
GR, indicating that the motion is same as GR. This completes the proof.

\end{proof}

Next, we show the equivalence of the gravitational field evolution to
GR. There are two levels of the proof to certify such an equivalence. One
sets out from the level of the Lagrangian and converts
$L_{\text{GR}_{||}}$ in (\ref{E:teleparallel Lagrangian 1}) into the
Hilbert-Einstein Lagrangian of GR $L_{\text{EH}} =\frac{1}{2\k}
\widetilde{\Omega}^{\a\b} \wedge \eta_{\a\b} $ via a remarkable
identity (see \cite{Muench:1998ay})
\begin{equation}
L_{\text{GR}_{||}}= -\frac{1}{2\k} \widetilde{\Omega}^{\a\b} \wedge
\eta_{\a\b} + d( \vt^\a \wedge \star d\vt_\a )
\end{equation}
where $\widetilde{\Omega}_{\a\b}$ is the curvature 2-form of
Levi-Civita connection $\widetilde{\o}_\a{}^\b$. The other assertion
for the equivalence is shown at the level of the gravitational field
equations as we present below. First we prove some useful
identities.

\begin{lemma}
\begin{align}
\frac{1}{2} K^{\mu\nu} \wedge \eta_{\a\mu\nu}  &= \star \left( -
{}^{(1)}T_\a + 2 {}^{(2)}T_\a +
\frac{1}{2} {}^{(3)}T_\a \right)\label{E:contortion id} \\
\O^\b{}_\a &= \widetilde{\O}_\a{}^\b - \widetilde{D}K_\a{}^\b -
K_\a{}^{\mu} \wedge K_{\mu}{}^\b \label{E:curvature id}
\end{align}
\end{lemma}

\begin{proof}
The proof for the first identity can be found in the Appendix C
\cite{Obukhov:2006gea}. Here we only prove the second. By
definition,
\[
\O^\b{}_\a = d\o_\a{}^\b - \o_\a{}^\g \wedge \o_\g{}^\b
\]
with the contortion 1-form $K_\a{}^\b = \widetilde{\o}_\a{}^\b -
\o_\a{}^\b$. It can be rewritten as
\[
\begin{aligned}
\O^\b{}_\a &= d \left( \widetilde{\o}_\a{}^\b - K_\a{}^\b \right) - \left( \widetilde{\o}_\a{}^\g - K_\a{}^\g \right) \wedge \left( \widetilde{\o}_\g{}^\b - K_\g{}^\b \right)\\
   &= \left( d\widetilde{\o}_\a{}^\b - \widetilde{\o}_\a{}^\g \wedge \widetilde{\o}_\g{}^\b \right) - \left( dK_\a{}^\b - K_\a{}^\g \wedge \widetilde{\o}_\g{}^\b - \widetilde{\o}_\a{}^\g \wedge K_\g{}^\b \right) - K_\a{}^\g \wedge K_\g{}^\b\\
   &= \widetilde{\O}^\b{}_\a - \widetilde{D}K_\a{}^\b - K_\a{}^\g \wedge K_\g{}^\b
\end{aligned}
\]

\end{proof}

\begin{theorem}
The $\text{GR}_{||}$ field equation $(\ref{E:Teparallel EOM})_1$ is
equivalent to Einstein's equation (GR).
\end{theorem}
\begin{proof}
Recall that the $\text{GR}_{||}$ choice $H_\b$ is given by
(\ref{E:TEGR excitation}). Via (\ref{E:contortion id}) it reads
\begin{equation}\label{E:special H_a 2}
H_\b = \frac{1}{2\k} K^{\mu\nu} \wedge \eta_{\b\mu\nu}.
\end{equation}
Thus the $\text{GR}_{||}$ Lagrangian (\ref{E:teleparallel Lagrangian
1}) can be rewritten as:
\begin{equation}
L_{\text{GR}_{||}} =  - \frac{1}{2} T^\b \wedge H_\b
                   =  - \frac{1}{4\k} \left( K^{\b\g} \wedge \vt_\g \right) \wedge \left( K^{\mu\nu} \wedge \eta_{\b\mu\nu} \right)
                   =    \frac{1}{2\k} \left( K^\b{}_{\mu} \wedge K^{\mu\nu} \wedge \eta_{\nu\b} \right)
\end{equation}
where we have used (\ref{E:torsion identity}), (\ref{E:special H_a
2}) and identity (\ref{E:eta basis id}). Then the $\text{GR}_{||}$
field equation $(\ref{E:Teparallel EOM})_1$ reads:
\begin{multline}\label{E:main eq}
D H_\a - e_\a \lr L_{\text{GR}_{||}} -  (e_\a \lr T^\b ) \wedge H_\b = \left( \frac{1}{2\k}  \underbrace{\widetilde{D}K^{\mu\nu} \wedge \eta_{\a\mu\nu}}_{\mathclap{(\bigstar)}} + K_\a{}^\b \wedge H_\b \right) \\
- \frac{1}{2\k}\left[ e_\a \lr \left( K^\b{}_{\mu} \wedge K^{\mu\nu}
\right) \wedge \eta_{\nu\b} + \underbrace{K^\b{}_{\mu} \wedge
K^{\mu\nu} \wedge \eta_{\nu\b\a}}_{\mathclap{(\bigstar)}}  \right] -
(e_\a \lr T^\b ) \wedge H_\b
\end{multline}
where the term
\begin{equation}
D H_\a = \widetilde{D} H_\a + K_\a{}^\b \wedge H_\b = \frac{1}{2\k}
\widetilde{D}K^{\mu\nu}  \wedge \eta_{\a\mu\nu} + K_\a{}^\b \wedge
H_\b
\end{equation}
utilizes (\ref{E:two D eq}), (\ref{E:special H_a 2}) and the fact
$\widetilde{D}\eta_{\a_1 \cdots \a_p} \equiv 0 $ from (\ref{E:eta basis id2}), and
\begin{align}
e_\a \lr L_{\text{GR}_{||}} = \frac{1}{2\k}\left[ e_\a \lr \left(
K^\b{}_{\mu} \wedge K^{\mu\nu} \right) \wedge \eta_{\nu\b} +
K^\b{}_{\mu} \wedge K^{\mu\nu} \wedge \eta_{\nu\b\a}  \right]
\end{align}
One observes that using the important identity (\ref{E:curvature
id}) with the $\text{GR}_{||}$ requirement $R^{\b\nu} \equiv 0 $,
the $(\bigstar)$ terms collect as
\begin{equation}
\frac{1}{2\k} \left( \widetilde{D}K^{\b\nu} + K^\b{}_{\mu} \wedge
K^{\mu\nu} \right) \wedge \eta_{\a\b\nu} = \frac{1}{2\k}
\widetilde{R}^{\b\nu} \wedge \eta_{\a\b\nu} = - \frac{1}{\k} \left(
\widetilde{R}_{\mu\a} - \frac{1}{2} \widetilde{R} \, g_{\mu\a}
\right) \eta^{\mu}
\end{equation}
which is immediately recognized as the Einstein tensor. The
appearance of this term almost claims the completion of the proof;
it remains to prove that all the rest of the terms in (\ref{E:main eq})
cancel, which only takes some more steps to see it. Indeed, since the
last term in (\ref{E:main eq}) can be rewritten as
\begin{align}
\left( e_\a \lr T^\b \right) \wedge H_\b &= \left( e_\a \lr \left( K^{\b\g} \wedge \vt_\g \right)  \right) \wedge H_\b \\
                                         &= \left( \left(e_\a \lr K^{\b\g}\right) \wedge \vt_\g - K^\b{}_\a \right) \wedge H_\b  = - 2 \left( e_\a \lr K^\b{}_{\mu} \right) \wedge K^{\mu\nu} \wedge \eta_{\nu\b} + K_\a{}^\b \wedge H_\b
\end{align}
together, the 3 terms remaining vanish,
\begin{equation}
K_\a{}^\b \wedge H_\b - \frac{1}{2\k}\left( e_\a \lr \left(
K^\b{}_{\mu} \wedge K^{\mu\nu} \right) \wedge \eta_{\nu\b} \right) -
\left( e_\a \lr T^\b \right) \wedge H_\b = 0
\end{equation}
where we have used $e_\a \lr \left( K^\b{}_{\mu} \wedge K^{\mu\nu}
\right) \wedge \eta_{\nu\b} = 2 \left( e_\a \lr K^\b{}_{\mu} \right)
\wedge K^{\mu\nu} \wedge \eta_{\nu\b}$.

\end{proof}

The above two theorems then conclude that TEGR is equivalent to GR
in the matter and geometric evolutions. One observes that the use of
computational tools developed in Sec.(\ref{Sec:Computational Aspects}) largely reduces the complicated calculation, while it
causes more abstraction and indirectness. For most teleparallel
theorists, the component forms of TEGR are then used more often (see \cite{Peireira}). Since the conversion in between is rarely seen, we demonstrate
some sketch below.
\begin{theorem}
\begin{equation}\label{E:teleparallel Lagrangian 2}
\begin{aligned}
L_{\text{GR}_{||}} &:= - \frac{1}{2} T^\a \wedge H_\a \\
                   &= \frac{1}{2\k} \left( \frac{1}{4} T_{\mu\nu\a} \, T^{\mu\nu\a} + \frac{1}{2} T_{\mu\nu\a} \, T^{\a\nu\mu} - T_{\mu\sigma}{}{}^{\sigma} \,
                   T^{\mu\delta}{}{}_{\delta}
                   \right) \, \eta \\
                   &:= - T_{\text{scalar}} \cdot \eta
\end{aligned}
\end{equation}
where the term $T_{\text{scalar}}:= \frac{1}{2\k} \left( \frac{1}{4}
T_{\mu\nu\a} \, T^{\mu\nu\a} + \frac{1}{2} T_{\mu\nu\a} \,
T^{\a\nu\mu} - T_{\mu\sigma}{}{}^{\sigma} \,
T^{\mu\delta}{}{}_{\delta} \right) \in C^{\infty}(M)$ is referred to as
the \textbf{torsion scalar}.
\end{theorem}
\begin{proof}
From (\ref{E:T irr pieces}) we simplify (\ref{E:teleparallel
Lagrangian 1}) as
\begin{equation}
L_{\text{GR}_{||}} = \frac{1}{2\k}  T_\a \wedge \star \left(T^\a -
\vt^\a \wedge i_{e_\b}(T^\b) - \frac{1}{2} i_{e^\a} \left( \vt_\g
\wedge T^\g \right)  \right)
\end{equation}
and then we compute the terms above individually
\begin{equation}
\begin{aligned}
\vt^\a \wedge i_{e_\b}(T^\b) &= - T_{\mu} \,  \vt^\a \wedge
\vt^{\mu}\\
i_{e^\a} \left( \vt_\g \wedge T^\g \right) &= T_{\a} -
T_{\a[\mu\nu]} \, \vt^{\mu} \wedge \vt^{\nu}
\end{aligned}
\end{equation}
where (\ref{E:T and R}) and (\ref{E:eta basis id}) are repeatedly
used and $T_{\mu}:= T_{\mu\nu}{}{}^{\nu}$ is the torsion trace we defined early. Note that here a single notation
$T_{\mu}$ can either denote \emph{torsion trace} or \emph{torsion
2-form}. However the confusion shall not rise since one can tell
from the correct degree of differential forms. After rearrangement,
one is left to compute
\begin{equation}
\begin{aligned}
T_\a \wedge \star T^\a  &= - \frac{1}{2}
T_{\mu\nu\a} \, T^{\mu\nu\a} \, \eta ,\\
T_\a \wedge \star \left( T_{\mu} \, \vt^\a \wedge \vt^\m \right) &=
T_\m \, T^\m \, \eta, \\
T_{\a} \wedge \star \left( -\frac{1}{2} i_{e^\a} (\vt_\g \wedge T^\g
)\right) &= -\frac{1}{2} T_\a \wedge \star T^\a - \frac{1}{2}
T^{\mu\nu\a} \, T_{\a [\mu\nu ]} \eta.
\end{aligned}
\end{equation}
if we put everything together, we obtain the component form in
(\ref{E:teleparallel Lagrangian 2}).

\end{proof}

\begin{remark}

Most of the literatures that studies teleparallel gravity, the component formulation is usually adopted. One of the merits is that the computation made direct. Since both a tetrad and local coordinate vectors are a basis, one can expand a tetrad in terms of the local coordinate vectors, hence
\begin{equation}
e_\a = e^i_\a(x) \, \frac{\partial}{\p x^i}, \qquad \vt^\a = e^{\a}_i (x) \, dx^i
\end{equation}
where $e^i_\a: U \subseteq M \to \mathbb{R} $ is the coefficient matrix. Along with the dual condition
\[
\vt^\b (e_\a) = \left( e^\b_i (x) \, dx^i \right) \,  \left( e^j_\a (x) \, \frac{\partial}{\p x^j} \right) =  e^\b_i(x) \, e^i_\a(x) = \d^\b_\a, \qquad \left( dx^i \left( \frac{\partial}{\p x^j} \right)\Bigr|_{x} = \d^i_j \right)
\]
for all $x\in U$, one obtains the relation $e^\b_i (x) \, e^i_\a (x) = \d^\b_\a  $. Thus the entire teleparallel field equation (\ref{E:Teparallel EOM2}) in terms of $e^\b_i$ and $ e^i_\a$ is only the equation of $e^\b_i$ alone. More often, teleparallelists tend to call the coefficients $ e^i_\a $ and $e^\b_i$ as a tetrad and co-tetrad.

\end{remark}

\section{Local Lorentz violation in TEGR}

Although we have proved that TEGR is equivalent to GR in many
aspects, there is one peculiar property that does not share with GR:
local Lorentz violation. It has been pointed out in \cite{Li:2010cg}, \cite{Sotiriou:2010mv}
that the TEGR Lagrangian $L_{\text{GR}_{||}} $ does not respect local Lorentz symmetry, where the local violation term is in the form of an exact differential. Hence in effect, as a boundary term, it does not affect the action and field equations. However, the story will be different in the $f(T)$ theory, in which violation terms cannot be eliminated, see \cite{Li:2010cg}.

In fact, we can give an account why the TEGR Lagrangian fails to be
Lorentz invariant from the viewpoint of the fibre bundle theory. To
depict the local Lorentz violation first we define it operationally

\begin{definition}{\textbf{(Local Lorentz Invariance and Violation)}}

Let $\{U_i \subseteq M \}$ be an open covering of the spacetime $M$,
and a family of local Lagrangian 4-forms $\{ L_i \in \Lambda^4(U_i, \mathbb{R})\}$ defined on $U_i$. If on nonempty overlaps $U_i \cap U_j$
with transition map $\Psi_{ij}: U_i \cap U_j \to
SO(1,3)$, the two Lagrangians agree $L_i(x) = L_j(x)$ (in the sense of up to a constant),
$\forall x\in U_i \cap U_j$, then we say $L_i$ or
$L_j$ is \textbf{locally invariant} on $U_i \cap U_j$. Moreover, if
$L := \cup_i L_i$ is a locally invariant on any nonempty overlaps,
then $L$ is a locally invariant Lagrangian on the whole $M$. Thus if a
Lagrangian is not locally invariant on some nonempty overlaps, we
call it \textbf{Lorentz violation} otherwise.
\end{definition}

From Section. \ref{Sec:chaning frame}, we see that if we have two orthonormal frames $(e_\a) \in \mathfrak{X}(U_i)$ and $(\widetilde{e}_\b ) \in \mathfrak{X}(U_j)$ defined on $U_i$ and $U_j$
respectively with $U_i \cap U_j \neq \phi $ and $\a$,
$\b= 0,1,2, 3$ with $(\vt^\a) \in \Lambda^1(U_i)$ and $(\widetilde{\vt}^\b) \in \Lambda^1(U_j)$ denoting their dual coframes, then by (\ref{E:frame change}) the frame changing between two sections $\s_{U_i} : U_i \to M$ and $\s_{U_j} : U_j \to M$ with the transition function $\Psi_{U_j U_i} : U_i \cap U_j \to SO(1,3)$ one has the transforms
\begin{equation}\label{E:local Lorentz}
e_\b(x) = (A(x))_\b^\a \, \widetilde{e}_\a(x) \, \Leftrightarrow \,
\widetilde{\vt}^\b(x) = (A^{-1}(x))^\b_\a \, \vt^\a(x)
\end{equation}
where $A_\b^\a := (\Psi_{U_j U_i})_\b^\a$ denotes the matrix form of $SO(1,3)$. According to such transformation rule, the TEGR action (\ref{E:teleparallel Lagrangian 1}) and the corresponding field momenta (\ref{E:TEGR excitation}) transforms as
\begin{equation}\label{E:Lorentz violation}
\begin{aligned}
\widetilde{\mathcal{L}}_G(x) &= \mathcal{L}_G(x)  - \frac{1}{2\kappa} d \left(
(A^{-1}(x))^\a_\b
\, d(A(x))^\b_\g \wedge \eta^\g{}_\a \right)\\
\widetilde{H}_\a(x) &= (A^{-1}(x))^\b_\a \, H_\b(x) - \frac{1}{2\kappa}(
A^{-1}(x))^\b_\a \, (A^{-1}(x))^\n_\g \, d(A(x))^\g_\m \wedge
\eta_{\b}{}^\m_{}{}_\n\\
\widetilde{t}_\a(x) &= (A^{-1}(x))^\b_\a \, t_\b(x) + d(A^{-1}(x))^\b_\a
\wedge H_\b \\
&\qquad \qquad\qquad \qquad - \frac{1}{2\kappa} d \left( (A^{-1}(x))^\b_\a \,
(A^{-1}(x))^\n_\g \, d(A(x))^\g_\m \wedge \eta_{\b}{}^\m_{}{}_\n
\right)(x)
\end{aligned}
\end{equation}
where $x\in U_i \cap U_j$ and the tilde-symbols refer to
quantities in the $U_j$ system.

After moments of thought, one recognizes that the violation of
(\ref{E:Lorentz violation}) follows from the fact that they are not \emph{basic differential
forms} (see Def. (\ref{Def:basic differential form})) from a principal bundle. One sees that if $\a \in \overline{\Lambda}^k(P,V)$ is a basic differential form, by applying the change of sections (\ref{E:change section}) its local form transforms as
\begin{equation}
\s_{U_j}^* \, \a = \rho\left( \Psi^{-1}_{U_i U_j}(x) \right) \cdot \s_{U_i}^* \,\a
\end{equation}
where $\s_{U_i}: U_i \to P$ is a local section. Some well-behaved examples are: the \emph{torsion 1-form} $T^\a$ by (\ref{E:LM change section}) and (\ref{E:frame change}), and the \emph{curvature 2-form} $\O^\o$ by (\ref{E:curvature transition}). An ill-behaved example is the \emph{connection 1-form} $\o$ by (\ref{E:connection transition}) and (\ref{E:connection transition2}). In fact, we have \cite{Steenrod}

\begin{theorem}

Basic differential forms on a principal bundle $\pi:P \to M$ are in 1-1 correspondence to global-defined differential forms on $M$ with
values in the associated bundle $P\times_G V$
\begin{equation}
\overline{\Lambda}^k(P,V) \cong \Lambda^k(M,P\times_G V),
\end{equation}
where the \textbf{associated bundle} $P\times_G V$ is defined as
$(P\times V)/G:=\{[p,v] | \, p \in P, v\in V \}$ by the quotient of
the $G$-action $(p,v)\cdot g := (p \cdot g,\rho(g^{-1})\cdot v )$.
Yet forms on $M$ of value in $P\times_G V$ are equivalently
represented by a family of $\mu_i \in \Lambda^k(U_i,V)$ with the
\emph{gluing condition}
\begin{equation}\label{E:gluing condition}
\mu_i = \rho(\Psi_{ji}^{-1}) \cdot \mu_j
\end{equation}
where $\Psi_{ji} : U_i \cap U_j \to G$ is the transition function.
\end{theorem}
Thus using the theorem, we conclude that

\begin{corollary}{(Local Lorentz violation)}

The reason for the local Lorentz violation (\ref{E:Lorentz violation}) is due to the fact that the TEGR Lagrangian $L_{\text{GR}_{||}}$ is \textbf{not} a basic differential form, i.e, not a global-defined scalar on $M$.
\end{corollary}

In contrast, we know that the Hilbert-Einstein action has a nice transformation behavior\footnote{In fact, before Hilbert gave the Lagrangian (\ref{E:HE Lagrangian}) Einstein took one different from the scalar curvature $R$ by a total differential and was not invariant under changes of the coordinate frame. Nevertheless, the corresponding field equations under the variation can still possess Lorentz invariance property.} under the change of frames. As one now looks back to the Einstein-Hilbert action, or even the more general Einstein-Cartan theory,
\[
\mathcal{L}_G := \frac{1}{2\k} \O^{\a\b} \wedge \eta_{\a\b} = \frac{1}{2\k} \, R \, \sqrt{-g} \, d^4x \in
\Lambda^4(M)
\]
it can be checked that it belongs to $\overline{\Lambda}^0(P,V)$. Therefore it is \emph{well-glued} on each overlap $U_\a \cap U_\b$ and hence globally defined showing no local Lorentz violation.

\begin{remark}
There is usually a misleading concept that a scalar without index, e.g.
a Lagrangian, is always an invariant scalar under change of frames, since it
may only be locally defined. One has to check the gluing condition carefully.
It should be now clear from the theorem that a locally invariant scalar must be originated from
$\overline{\Lambda}^k(P,V) \cong \Lambda^k(M,P\times_G V)$, and hence the word
``tensor" does not mean too much unless it is specified globally or locally.

However it is generally believed that physics law should remain the
same form in every frame which is an essence of a gauge
theory. This is reminiscent of the Yang-Mill's theory of the non-Abelian gauge. Local (gauge) violation occurs if we consider the field strength defined like $F = dA$ for some gauge potential $A\in \Lambda^1(M,\mathfrak{g})$ transforming as
\[
A =
\varphi^{-1} \, d\varphi + \varphi^{-1} \, \tilde{A} \, \varphi
\]
by (\ref{E:connection transition2}). To cure the local violation of $F$, the remedy is to add the term $A \wedge A$ such that $\widetilde{F} := dA + A\wedge A$, which was the Yang-Mills original idea. Thus $\widetilde{F}$ will then transforms as ("a tensor" by an abuse of language) $\widetilde{F} = \Psi^{-1} \, \widetilde{F} \Psi$, which is now globally defined.
\end{remark}

\subsection{Generalized teleparallel theories}

Despite the local Lorentz violation term, the interesting property of TEGR equivalence with GR makes some people to conceive an extension of teleparallel gravity to $f(T)$ theory (here $T=T_{\text{scalar}}$). This proposal mainly replaces the Lagrangian $L_{\text{GR}_{||}}$ with $L = f(T)$ by mimicking the $f(R)$ phenomenological gravity models, while the underlying geometry remain unchanged. Thus one can see that $f(T)$ theory contains nontrivial local
Lorentz violation terms that cannot be treated as boundary terms, see \cite{Li:2010cg}. In effect, the field equations of $f(T)$ deviate from those of GR. However \cite{Li:2013oef} finds that if one considers the weak field limit of the tetrad field on the solar system scale, the geodesic equation of $f(T)$ coincides with that of GR up to the Newtonian limit, which indicates that the current solar system observations hardly distinguish these two theories, regardless of the actual form of $f(T)$. However, $f(T)$ theory does demonstrate a different evolution history on the cosmological scale. Some recent studies show that $f(T)$ has a certain effect in on the sub-horizon scale \cite{Wu:2012hs}, which modifies the effective Newtonian constant and causes different formation history for the large scale structure. In this background, $f(T)$ serves as an alternative theory to compare with $\Lambda$CDM. Several papers are devoted in the subsequent investigations of $f(T)$ models.

In addition to $f(T)$ theory, another minimal extension of teleparallel gravity mimicking the scalar-tensor theory in GR was proposed, called\textbf{ teleparallel dark energy}. Such a theory was shown to provide a contrast to that of GR, see \cite{Geng:2011aj}, \cite{Geng:2011ka}. Some other problems of $f(T)$ were also found by Ong, Izumi, Nester, and Chen \cite{Ong:2013qja} that in a generic $f(T)$ theory there could be super-luminal propagating modes due to the effects of nonlinear constraints via the analysis of the corresponding characteristic equations and the Hamiltonian structure. Nevertheless, $f(T)$ theory still provides useful phenomenological explanations in cosmology problems \cite{Bamba:2010wb}, \cite{Geng:2012vn}.


\chapter{Five Dimensional Theories} \label{TEGR}

In this chapter we study the five dimensional theories of teleparallel
gravity. Kaluza \cite{Kaluza} and Klein \cite{Klein} attempted to apply extra dimension to unify
electromagnetism and gravity.\footnote{Around 1914, Nordstr\"{o}m proposed a similar 5D unified field theory attempt of gravitation and Maxwell's electromagnetism before Kaluza and Klein, see \cite{Nordstrom}.}

Kaluza's ansatz in 1921 was to split the 5-dimensional spacetime
$\overline{M}$ into a 4-dimensional spacetime $M$ and the Maxwell's
electrodynamics we perceive. However, Kaluza's attempt was to
consider the fibre as $\mathbb{R}^1$ while Oskar Klein in 1926
modified his theory with curled fibre $S^1 \cong U(1)$, a compact
Abelian Lie group. Nowadays we can interpret the Kaluza-Klein theory
as a principal fibre bundle $\pi: P = \overline{M}\to M$ with gauge
group $G = U(1)$, where $M$ denotes our 4-dimensional spacetime. The
splitting of $\overline{M}$ into $M$ with electromagnetic force is
then due to the nature of the local trivialization of the principal fibre bundle
in Def.(\ref{Def:PFB}) with an assigned metric $\bar{g} = \pi^* g +
\overline{A} \otimes \overline{A}$, where $g$ is the metric on $M$
and $A $ is a connection form on $M$ such that $\overline A = \pi^*
A $, see \cite{Jost2}.

In the following we are motivated by the Kaluza-Klein theory to
construct five-dimensional theories in teleparallel gravity.

\section{Five-dimensional teleparallel gravity construction}

As before, we keep our setting as general as possible. Let $f:M \to V$ be an isometric embedding of a 4-dimensional Lorentzian manifold $M$ (hypersurface) into a 5-dimensional Lorentzian manifold (the bulk) $(V,\bar{g})$.

Consider a tetrad $\{e_0,\ldots,
e_3\}$ on $M$ and its natural extension as a tetrad
$\{\bar{e}_0,\ldots,\bar{e}_3,\bar{e}_5\}$ on $V$, where
$\bar{e}_a:=f_*(e_a)$ and $\bar{e}_5$ is the unit normal vector
field to $M$. The corresponding dual coframe to $e_a $ is $\{\vt^0,\ldots,\vt^3\}$ and to $\bar{e}_A$ is $\{\bvt^0,\ldots,\bvt^3,\bvt^5\}$, and thus $f^*(\bvt^a) =\vt^a$. We shall identify $M$ with $\overline{M}:=f(M)\subset V$,
$X\in \mathfrak{X}(M)$ with $\bar{X}\in \mathfrak{X}(\overline{M})$, and
$\vt^a \in T^*M$ with $\bvt^a \in T^*\overline{M}$...,etc
interchangeably.

Indices of the tetrad on $M$ are labeled by small Latin letters
$a,b,c,\ldots,= 0,1,2,3$. For the local coordinate on $M$ the index is denoted by
Greek letters $\mu,\nu,\ldots =0,1,2,3$. Both the spatial components of the local coordinates and tetrad on $M$ share the labelling by middle Latin letters
$i,j,k, \ldots =1,2,3$; indices of the tetrad in $V$ are labelled by capital Latin letters
$A,B,C,\ldots,= 0,1,2,3,5$; indices of the local coordinate, e.g, $dx^M=(dx^{\mu},dx^5)$, in $V$ are denoted by capital Latin letters $M, N= 0, 1, 2, 3, 5$, where $5$ denotes the extra dimension ($5^{th}$ dimension). Quantities with a bar, e.g,
$\bar{e}_A$, are used to mean objects viewed in $V$.

\begin{figure}[h]
\begin{center}
\scalebox{0.6}{\includegraphics{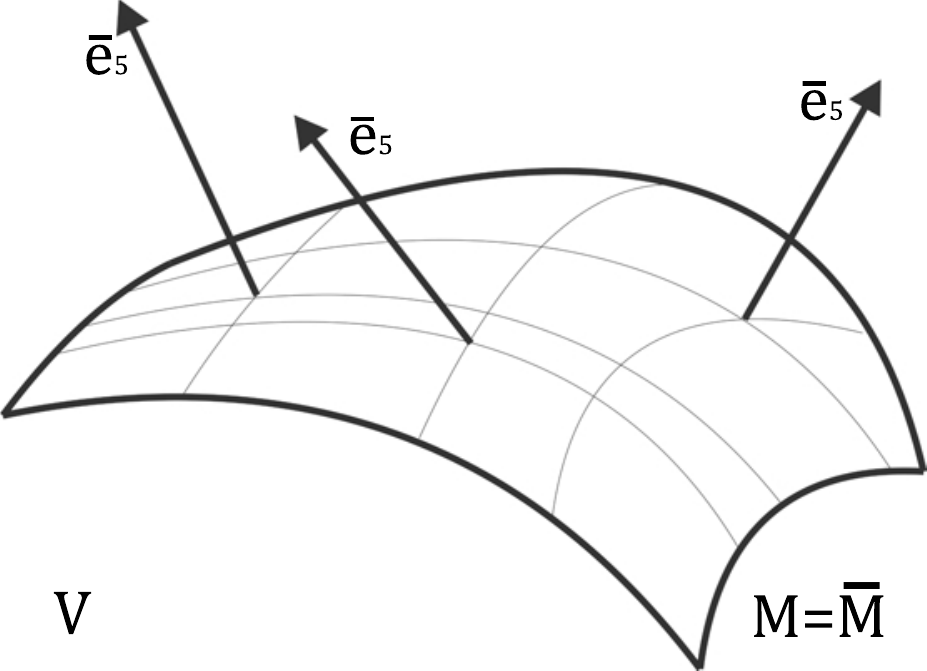}} \caption{The isometric embedding $f:M \to V$.}
\end{center}
\label{fig:embedding}
\end{figure}

The 4-dimensional teleparallel gravity on $M$, as introduced in Chapter 4, is formulated by a tetrad $\vt^a$, and the
Weitzenb\"{o}ck connection defined by (\ref{E:Weitzenbock connection}) with respect to $e_a$, denoting $ \nabla^W e_b := \omega_b{}^a \, e_a$. The metric $g$ of $M$ is written as
\begin{equation}
ds^2 = g_{\mu\nu}\,dx^{\mu} \otimes dx^{\nu} = \e_{ab}\,\vt^a
\otimes \vt^b\,,
\end{equation}
where $\e_{ab}$ denotes the Minkowski metric. The metric signature is fixed as $(-,+,+,+,\veps)$ where $\veps:= \bar{g}(\bar{e}_5,\bar{e}_5)= \pm 1$ is the sign of the $5^{th}$ dimension.

With the TEGR Lagrangian (\ref{E:teleparallel Lagrangian 1})
\begin{equation}\label{E:torsion 4-form}
\mathcal{T} := T_a \wedge \star
\left[ {}^{(1)}T^a - 2 \, {}^{(2)}T^a - \frac{1}{2} \, {}^{(3)}T^a \right],
\end{equation}
or the component form (\ref{E:teleparallel Lagrangian 2})
\begin{subequations}\label{E:4D Torsion scalar}
\begin{eqnarray}
T &=& \frac{1}{4}\,T_{abc} \, T^{abc} +
      \frac{1}{2}\,T_{abc} \,T^{cba} -
      \,T^{b}{}_{ba} \, T^{c}{}_{c}{}^{a}\label{E:torsion scalar 1}\\
  &=& \frac{1}{4}\,T_{\mu\nu\sigma} \, T^{\mu\nu\sigma} +
      \frac{1}{2}\,T_{\mu\nu\sigma} \,T^{\sigma\nu\mu} -
      \,T^{\nu}{}_{\nu\mu} \, T^{\s}{}_{\s}{}^{\mu}   \,.\label{E:torsion scalar 2}
\end{eqnarray}
\end{subequations}
so that $\mathcal{T} = T \star 1$. Teleparallel gravity on the bulk $V$ is similarly defined
with the Weitzenb\"{o}ck connection 1-form $\{\bar{\omega}^B_A\in \Lambda^1(V)\}$ on $V$ corresponding to $\{\bar{e}_0,\ldots,\bar{e}_3,\bar{e}_5 \}$, and the 5-form TEGR Lagrangian $\bar{\mathcal{T}} $ for $V$
\begin{equation}\label{E:5D Torsion scalar}
\bar{\mathcal{T}} =  \bar{T}_A \wedge \bar{\star}
\left[ {}^{(1)}\bar{T}^A - 2 \, {}^{(2)}\bar{T}^A
- \frac{1}{2} \, {}^{(3)}\bar{T}^A \right]\,,
\end{equation}
where $\bar{\star}$ is the Hodge dual operator in $(V,\bar{g})$ and
$\bar{T}^A:= \nabla^W \bvt^A = \bar{d}\bvt^A + \bar{\omega}^A_B
\wedge \bvt^B = \bar{d} \bvt^A$ is the torsion 2-form on $V$. Here there are two Cartan differentials $d:\Lambda^k(M) \to
\Lambda^{k+1}(M)$ and $\bar{d}:\Lambda^k(V) \to \Lambda^{k+1}(V)$ to be carefully distinguished, along with the requirement
$\bar{d}\bigr|_{\Lambda^k(M)} = d\bigr|_{\Lambda^k(M)} $. The gravitational action on $V$ is given by
\begin{equation}
{}^{(5)}S= - \frac{1}{2\kappa_5} \int \bar{\mathcal{T}} = - \frac{1}{2\kappa_{5}}
\int \bar{T} \,  \bar{\star} 1 \, ,
\end{equation}
where $\kappa_{5}=8\,\pi\,G^{(5)}$ represents the 5-dimensional gravitational
coupling, $\bar{T}$ stands for the torsion scalar of $V$, and
$\bar{\star} 1 = {}^{(5)}e \, d^5x =\det (e_{M}^A) \, d^5x $
is the volume form of $V$.

In order to understand 5-dimensional teleparallel gravity, it is necessary to look back to the analysis in that of GR.


\section{Five-dimensional gravity of GR}

Here we provide an approach, which is rarely found in the literature, using Cartan's moving frame to derive the Gauss-Codazzi Theorem that provides a useful basis for the later construction. The Gauss-Codazzi equation has the importance of connecting relations between the 5-dimensions and the 4-dimensions, and hence gives the projected information of the five-dimensional spacetime down to the 4-dimensional spacetime.

Keeping the notations for the setting from the last section, and suppose now that $(V,\overline{\nabla} ,\bar{g})$ and $(M,g,\nabla) $ are both semi-Riemannian manifolds. From Cartan's structure equation (vanishing torsion),
\begin{equation}\label{E:Cartan structure}
\begin{aligned}
\bar{d} \, \bvt^A &+ \bar{\o}_B{}^A \wedge \bvt^B =0 , \qquad \mbox{(on $\overline{M}$)}\\
d \vt^a &+ \o_b{}^a \wedge \vt^b =0 , \qquad \mbox{(on $M$)}
\end{aligned}
\end{equation}
in particular, from $(\ref{E:Cartan structure})_1$ with $A=5$ and $A=a$ we have
\begin{equation}
\begin{aligned}
\bar{d} \, \bvt^5 &+ \bar{\o}_b{}^5 \wedge \bvt^b =0 ,\\
\bar{d} \, \bvt^a &+ \bar{\o}_5{}^a \wedge \bvt^5 + \bar{\o}_b{}^a \wedge \bvt^b =0
\end{aligned}
\end{equation}
also from $(\ref{E:Cartan structure})_2$, we may replace $\bar{d}\bvt^a = -\o_b{}^a \wedge \vt^b$ since we require $\bar{d}\bigr|_{\Lambda^k(M)} = d\bigr|_{\Lambda^k(M)} $. Then we derive
\begin{equation}\label{E:extrinsic connection}
\left( \bar{\o}_b{}^a - \o_b{}^a \right) \wedge \vt^b + \bar{\o}_5{}^a \wedge \bvt^5 =0.
\end{equation}

Recall that the extrinsic curvature (or the second fundamental form) is defined by
\begin{equation}
\nabla_{\overline{X}} \overline{Y} -\nabla_X Y = K(X,Y) \, \bar{e}_5
\end{equation}
or alternatively
\begin{equation}\label{Def:extrinsic curv}
K(X,Y) : = \veps \bar{g} \left( \nabla_{\overline{X}} \overline{Y} -\nabla_X Y  , \bar{e}_5  \right) = \veps \bar{g} \left( \nabla_{\overline{X}} \overline{Y} , \bar{e}_5  \right).
\end{equation}
Taking $X = e_a$, $Y= e_b$, one obtains
\begin{equation}\label{E:extrinsic curv}
\bar{\o}_a{}^5(e_b) :=\bar{\o}_{ba}{}^5 = K(e_a, e_b) := K_{ab}.
\end{equation}
Note that here $K_{ab}$ denotes the extrinsic curvature 2-tensor, not the contortion 1-form (\ref{E:contortion}). On the other hand, the metric compatible condition
\begin{equation}
0  = \bar{d} \bar{g} \left(\bar{e}_a, \bar{e}_5 \right) =  \bar{g} \left( \overline{\nabla} \bar{e}_a, \bar{e}_5 \right) + \bar{g} \left( \bar{e}_a, \overline{\nabla} \bar{e}_5 \right)
\end{equation}
helps us to derive
\begin{equation}
\bar{\o}_5{}^b (e_a):= -\eta^{bc} \veps \, \bar{\o}_c{}^5 (e_a) =  -\eta^{bc} \veps \,  \bar{\o}_{ac}{}^5  = - \veps K^b{}_a
\end{equation}
where (\ref{E:extrinsic curv}) is used. Finally combined with (\ref{E:extrinsic connection}) we obtain the relation of connection 1-forms on $V$ and on $M$, given by
\begin{equation}\label{E:extrinsic connection 2}
\bar{\o}_b{}^a = \o_b{}^a - \veps K_b{}^a \bvt^5.
\end{equation}
Therefore only with some simple computation, one has attained
\begin{theorem}{(\cite{Bartnik:2002cw})}
\begin{equation}\label{E:Gauss eq}
\begin{aligned}
\overline{R}^a{}_{bcd} &= R^a{}_{bcd} - \veps \,  K^a{}_c \, K_{db} + \veps \, K^a{}_d \, K_{cb}, \qquad \mbox{(Gauss)}\\
\overline{R}  &= R + 2 \left( \veps \mathcal{L}_{\bar{e}_5} tr K - K^{ab} \, K_{ab} \right), \qquad \mbox{(Mainardi)}
\end{aligned}
\end{equation}
\end{theorem}
\begin{proof}
Simply use the relation (\ref{E:extrinsic connection 2}) on the curvature 2-form
\begin{equation}
\overline{\Omega}^A{}_B := \bar{d} \bar{\o}_B{}^A + \bar{\o}_C{}^A \wedge \bar{\o}_B{}^C = \frac{1}{2} \overline{R}^A{}_{BCD} \, \bvt^C \wedge \bvt^D.
\end{equation}
With $A=a$ and $B=b$, one can derive the first equation; with $A=a$, $B=5$ one can derive
\begin{equation}\label{E:U(1) action}
\overline{R}^a{}_{5b5} = \veps \mathcal{L}_{\bar{e}_5} K^a{}_b - K^{ac} \, K_{cb} , \quad \overline{R}^a{}_{5a5} = \veps \mathcal{L}_{\bar{e}_5} tr K - K^{ab} \, K_{ab},
\end{equation}
where $\mathcal{L}_{\bar{e}_5}$ is the Lie derivative with respect to the normal direction. The second identity of the curvature scalar is obtained from
\begin{equation}\label{E:R+U(1) action}
\overline{R} = \overline{\O}^{AB} \wedge \overline{\eta}_{AB} = \overline{\O}^{ab} \wedge \overline{\eta}_{ab} + 2 \overline{\O}^{a5} \wedge \overline{\eta}_{a5}.
\end{equation}
\end{proof}

One notices that if we put the so-called cylindrical condition (see (\ref{E:Cylindrical condition})), namely $\mathcal{L}_{\bar{e}_5} K^a{}_b \equiv 0$, then $(\ref{E:Gauss eq})_2 $ reads $\overline{R}  = R - 2 K^{ab} \, K_{ab} $, where the last term is reminiscent of the Maxwell's electrodynamics action $-\frac{1}{4} F_{ab} \, F^{ab}$ with a suitable reparametrization \cite{Ryder}. If we interpret the extrinsic curvature of $M$ in $V$ as the electrodynamics we perceive, then the five-dimensional gravity shall induce four-dimensional gravity plus Maxwell's electrodynamics. Also notice that there is no extra coupling term between gravity and electromagnetism in this theory.

With these understanding in GR, next we turn to teleparallel gravity.


\section{Effective gravitational action on $M$}

To derive the effective teleparallel gravity from $V$ onto $M$, a crucial ingredient is to find a relation analogous to the Gauss equation (\ref{E:Gauss eq}). However one notices in TEGR that the Weitzenb\"{o}ck connection leads to vanishing extrinsic curvature (\ref{Def:extrinsic curv}) of $M$ in $V$, $K_{ab} \equiv 0 $ , since $\overline{\nabla}^W \overline{e}_5 = \bar{\o}_5{}^a \, e_a \equiv 0$. Thus the Gauss equation (\ref{E:Gauss eq})
\[
\overline{R}^a{}_{bcd} = R^a{}_{bcd} - \veps \,  K^a{}_c \, K_{db} + \veps \, K^a{}_d \, K_{cb}
\]
simply becomes a zero identity due to $\overline{R}^a{}_{bcd} = 0$ and $R^a{}_{bcd} = 0 $ in teleparallel gravity, which indicates that the brane is like a
flat-paper of $\mathbb{R}^2$ put into an Euclidean space $\mathbb{R}^3$.

The extra degree of freedom in TEGR actually lies in the torsion 2-form.
From (\ref{E:T and R}) we may decompose the torsion $\bar{T}^a$ of
$V$ into normal and parallel components with respect to $M$
\begin{equation}\label{E:torsion decompose}
\bar{T}^a = T^a + \bar{T}^a{}_{b 5} \, \bvt^b \wedge \bvt^5 \,,
\end{equation}
where $T^a=\frac{1}{2} T^a{}_{bc}\, \vt^b \wedge \vt^c$ is the
torsion 2-form on $M$. Therefore the second term in (\ref{E:torsion decompose}) plays a
role like the extrinsic curvature $K_{ab}$ in GR.

Since $f:M \to V$ is an isometric embedding, for any chart
$(\chi,U=(x^{\mu})\subseteq \mathbb{R}^4)$ of $M$ we can always find
a local normal form $(\Phi,W=(x^{\mu},y)\subseteq  \mathbb{R}^5)$
such that $F=\Phi^{-1} \circ f\circ \chi :\chi(U) \to \Phi(W)$ is
given by $F(x^{\mu}) = (x^{\mu},y=0)$ and the 5D metric $\bar{g}$ is in the form
\begin{equation}
\bar{g}_{MN} =
\begin{pmatrix}
    g_{\mu\nu}(x^{\mu},y)  &  0  \\
    0  &  \veps \phi^2(x^{\mu},y) \,,
\end{pmatrix}
\end{equation}
where $y=x^5$. The choice of the notation $\phi$ is in order for the extra dimensional effect to mimic a scalar field. Within such coordinates, we obtain a preferred frame for
$V$ with
\begin{equation}\label{E:special frame}
\bar{e}_A= \left( e_a, \frac{1}{\phi} \, \frac{\partial}{
\partial y} \right), \quad \bvt^A = \left( \bvt^a, \phi  \, dy \right)\,,
\end{equation}
such that $\bar{g} = \eta_{ab} \, \vartheta^a\otimes\vartheta^b+
\varepsilon \bvt^5\otimes \bvt^5$. In this case,
\begin{equation}
\bar{T}^n = \bar{d} \bvt^n + \bar{\omega}^n_A \wedge \bvt^A =
\bar{d}\bvt^n = \frac{e_a(\phi)}{\phi} \, \bvt^a \wedge \bvt^n
\end{equation}
one reads $\bar{T}^{5}{}_{b 5}= \frac{1}{\phi} \, e_{b}(\phi)$ where $e_{b}(\phi) := (d\phi)(e_b)$. Thus we conclude that in the frame (\ref{E:special frame}), the
nonvanishing torsion components of $V$ are $T^{a}{}_{bc}$,
$\bar{T}^{a}{}_{5b}$ and
$\bar{T}^{5}{}_{b 5}= \frac{1}{\phi} \, e_{b}(\phi)$.

So far we have not yet specified the type of our five-dimensional spacetime. If we now let the ambient space $V$ be a local product of $U \times W$,
where $U \subseteq M$ is open in $M$ and $W$
corresponds to the extra spatial dimension. Utilize (\ref{E:5D Torsion scalar}) and the local product
structure of $V$, we may compute the integration of the five-dimensional
action over the base space $U$ of $M$,
\begin{equation}\label{E:5D action}
S_{\text{bulk}} = \frac{-1}{2\k_5} \int_U \int_W \left(T  +
\frac{1}{2} \left( T_{ab5} \, T^{ab5} + T_{a5b}\, T^{b5a} \right) +
\frac{2}{\phi} e_a(\phi) \, t^a -  t_5 \cdot t^5 \right) \phi dy \,
dvol^4 ,
\end{equation}
where $T$ is the (induced) 4-dimensional torsion scalar
defined in (\ref{E:4D Torsion scalar}) and let $t^a := T_{ab}{}{}^{b}$ denote the torsion trace instead of $T_a$ to avoid confusion. The equation (\ref{E:5D action}) then provides us with a general
effective action on the hypersurface $M$ in TEGR theory. Next, we concentrate on
two specific theories of braneworld and Kaluza-Klein scenarios.

\begin{figure}[h]
\begin{center}
\scalebox{0.17}{\includegraphics{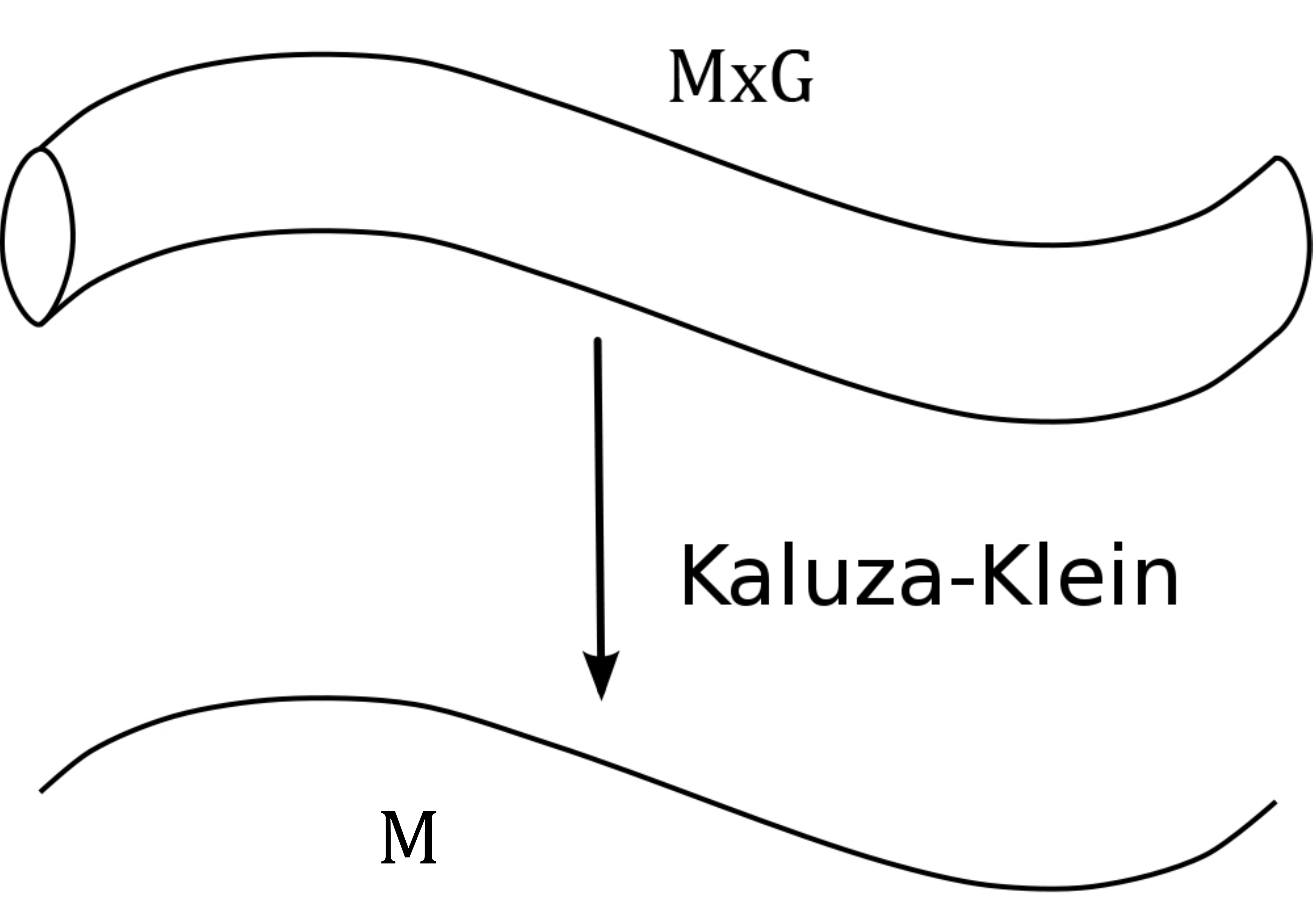}} \caption{The space $M \times G $ is compactified over the compact Lie group $G$, and after the Kaluza-Klein decomposition we have an effective field theory over $M$ (Image courtesy of Wikipedia).}
\end{center}
\label{fig:KK}
\end{figure}

\subsection{Braneworld Scenario}
In the braneworld scenario, we set the hypersurface $M$ located at
$y=0$ as a brane and specify the fibre $W = \mathbb{R}$ such that $V = M \times \mathbb{R}$.  From (\ref{E:5D action}), the general action on the bulk reads
\begin{equation}\label{E:bulk action}
S_{\text{bulk}} = \frac{-1}{2\k_5} \int_M \int_{\mathbb{R}}
\left\{\phi T + \phi \left( \frac{1}{2} \left( T_{ab5} \, T^{ab5} +
T_{a5b}\, T^{b5a} \right) + \frac{2}{\phi} e_a(\phi) \, t^a -  t_5
\, t^5 \right)\right\} \, dy \, dvol^4
\end{equation}
The first term of the parentheses $\int_M \int_{\mathbb{R}}\phi \, T \, \sqrt{-g} \, dy
\, d^4x$ in (\ref{E:bulk action}) is recognized as the usual TEGR Lagrangian with a non-minimally
coupled scalar field $\phi$ on the brane localized in the
fifth-dimension, which is analogous to the non-minimally
coupled Hilbert action $\int_M
\int_{\mathbb{R}} \phi \, R \, \sqrt{-g} \, dy \, d^4x$ of
4-dimension in GR. The second term arises from the fifth-dimensional
component.

According to the \emph{induced-matter theory}, the fifth-dimensional
component and the flow along the 5th-dimension of the second term in (\ref{E:bulk action})
can be regarded as the induced-matter
from the geometry.
It is the projected effect due to the extra spatial dimension.
We note that the \emph{mathematically equivalent formulation} between
the induced-matter and braneworld theories has been demontsrated by Ponce de Leon
in~\cite{PonceDeLeon:2001un}.

\subsection{Kaluza-Klein Theory}

If we identity the space $V$ locally as $U
\times S^1$ where topologically $S^1 \cong U(1)$ and consider the 4-dimensional effective low-energy
theory to require the \emph{Kaluza-Klein ansatz} in TEGR,
\begin{equation}
\label{E:Cylindrical condition}
e_5(g_{\mu\nu}) =0 \quad \mbox{or}\qquad \frac{\partial}{\partial
y}g_{\mu\nu} =0, \quad \mbox{(Kaluza-Klein ansatz)}
\end{equation}
which is also called the \textbf{cylindrical condition} to indicate that the function is independent of the fibre level. With this condition the theory allows only the \emph{massless Fourier mode} \cite{Overduin:1998pn}.
The metric is then reduced to
\begin{equation}
\bar{g}_{MN} =
\begin{pmatrix}
    g_{\mu\nu}(x^{\mu})  &  0  \\
    0                    &  \phi^2(x^{\mu})
\end{pmatrix}
\end{equation}
with $\veps=+1$. By the \emph{KK ansatz},
we see $T^a{}_{b5} =0$ and $t^5=0$ so that the extra-dimensional integration is trivial, i.e,
$\int_{S^1} \phi(x^{\mu})\,dy=2\pi r\, \phi(x^{\mu})$,
where $r$ is the radius of the fifth-dimension. As a result, we obtain
\begin{equation}
\label{E:effective KK action}
S_{\text{KK}} = \frac{-1}{2\kappa_{4}} \int_U \left( \phi\, T
+ 2\,\partial_{\mu}\phi \, t^{\mu} \right) e \, d^4x ,
\end{equation}
where $\kappa_{4}:= \k_5 /2\pi r$ is the \emph{effective}
4-dimensional gravitation coupling constant.
We point out that our result of (\ref{E:effective KK action})
disagrees with that given in~\cite{Bamba:2013fta}.
One can adopt a simple case with $F(T)=T$ in Eq.~(5) of~\cite{Bamba:2013fta} to check that the resultant equation
differs from ours in~(\ref{E:effective KK action}).

In the next section, we examine our five-dimensional theory of TEGR in a cosmological background.


\section{Friedmann Equation of Braneworld Scenario in TEGR}

Before the computation for the TEGR theory, it is necessary to reformulate the FLRW cosmology of GR in differential forms.

\subsection{FLRW cosmology in GR}

In GR, the FLRW universe of sectional curvature $k=0$ is described by
\begin{equation}
g_{\mu\nu } = diag \left( -1 , a^2(t), a^2(t), a^2(t) \right)
\end{equation}
a canonical choice of coframe field is given by
\begin{equation}
\vt^0 = dt, \quad  \vt^i = a(t) \, dx^i
\end{equation}

First we need to compute the connection 1-form of the Levi-Civita connection $\widetilde{\nabla} e_a = \tdo_a{}^b \, \vt^b$. Start with Cartan's structure equation $(\ref{E:Cartan structure})_2$
\begin{equation}
0 = T^0  = \cancel{d\vt^0} + \tdo_i{}^0 \wedge \vt^i
\end{equation}
on the isotropic condition, one can choose  $\widetilde{\o}_i{}^0 = B(t,\mathbf{x}) \, dx^i$ for some temporal function $B(t,\mathbf{x})$, and the spatial components yield
\begin{equation}\label{E:FRW Cartan struc2}
0 = T^i  = d\vt^i + \tdo_j{}^i \wedge \vt^j + \tdo_0{}^i \wedge \vt^0.
\end{equation}
Again, on the isotropic consideration, we may assume $\tdo_j{}^i = C_{kj}{}{}^i (t,\mathbf{x}) \, dx^k$, then (\ref{E:FRW Cartan struc2}) tells us that
\[
0 = \left( \dot{a}(t) - B(t,\mathbf{x}) \right) \, dt \wedge dx^i + C_{kj}{}{}^i(t,\mathbf{x})\, dx^k \wedge a(t ) \, dx^j.
\]
By the linear independence of $\{dt, dx^i\}$, one arrives at $B(t,\mathbf{x}) = \dot{a}$ and $C_{kj}{}{}^i(t,\mathbf{x}) = 0$, which is
\begin{equation}\label{E:Riemann connection}
\tdo_i{}^0  =  \eta_{ij} \, \tdo_0{}^j = \frac{\dot{a}}{a} \, \vt^i, \quad \mbox{and} \quad  \tdo_j{}^i  \equiv 0.
\end{equation}
where we have used the relation $\tdo_0{}^i = - \eta_{00} \, \eta^{ij} \, \tdo_j{}^0$ from the metric compatibility. With (\ref{E:Riemann connection}), the curvature 2-forms (\ref{E:T and R}) are
\begin{equation}
\begin{aligned}
\Omega^{jk} &= \eta^{km} \, \O^j{}_m = \frac{\dot{a}^2}{a^2} \, \vt^j \wedge \vt^k \\
\Omega^{0j} &= \eta^{jm} \, \O^0{}_m = \frac{\ddot{a}}{a} \, \vt^0 \wedge \vt^j
\end{aligned}
\end{equation}
one can readily verify the Friedmann equation from (\ref{E:EC EOM})
\[
\begin{aligned}
\frac{1}{2} \, \O^{\b\g} \wedge \eta_{0\b\g}  &= \k t_0 , \quad \Leftrightarrow  \quad \left( \frac{\dot{a}}{a} \right)^2 = \k \rho, \\
\frac{1}{2} \, \O^{\b\g} \wedge \eta_{i\b\g}  &= \k t_i , \quad \Leftrightarrow  \quad 2\frac{\ddot{a}}{a} + 4 \left( \frac{\dot{a}}{a} \right)^2 = \k (\rho - p)
\end{aligned}
\]
where the energy-momentum 3-form is defined by $t_0 = \rho \, \eta_0$ and $t_i =   p \, \d_{ij} \, \eta^j$. Again, one observes that the use of differential forms reduces tedious Christoffel symbol and curvature components computation. With these helpful computations in the Riemannian case, we may return to TEGR.

\subsection{FLRW Brane Universe}

Now we apply the teleparallel braneworld effect in
cosmology. Assume that the brane $M$ at $y=0$ is a
homogeneous and isotropic universe. The bulk metric $\bar{g}$
is further assumed to be a maximally symmetric 3-space with spatially flat ($k=0$) by
\begin{equation}\label{E:bulk metric}
\bar{g}_{MN} = diag \left(-1, a^2(t,y), a^2(t,y), a^2(t,y),
\veps\,\phi^2(t,y)\right)
\end{equation}
by choosing the canonical coframe field
\begin{equation}
\bvt^0 = dt, \quad \bvt^i = a (t,y)\, dx^i, \quad and \quad
\bvt^5 = \phi(t,y) \, dy
\end{equation}

Subsequently, the torsion 2-forms are
\begin{equation}
\bar{T}^0= \bar{d} \, \bvt^0 = 0,\qquad \bar{T}^i = \bar{d} \bvt^i =
\frac{\dot{a}}{a} \, \bvt^0 \wedge \bvt^i + \frac{a'}{a\phi} \,
\bvt^5 \wedge \bvt^i, \qquad \bar{T}^5 =  \frac{\dot{\phi}}{\phi} \,
\bvt^0 \wedge \bvt^5\,,
\end{equation}
where the \emph{dot} and  \emph{prime} stand for the partial
derivatives respect to $t$ and $y$, respectively. Thus the non-vanishing torsion components of five-dimension are read off
\begin{equation}
T_{0i}{}{}^i = \frac{\dot{a}}{a}, \quad T_{ni}{}{}^i = \frac{1}{\phi} \frac{a'}{a}, \quad  T_{na}{}{}^n = \frac{e_a(\phi)}{\phi}
\end{equation}
With these data, the bulk Lagrangian (\ref{E:bulk action}) in the FRLW cosmology background has the form
\begin{equation}
\bar{\mathcal{T}} = \left[ T + \left(\frac{3-9\,\veps}{\phi^2}\,
\frac{a'^2}{a^2} + 6\, \frac{\dot{a}}{a} \, \frac{\dot{\phi}}{\phi}
\right) \right] dvol^5
\end{equation}
where $\veps=+1$ and $T = 6\dot{a}^2/a^2$ is the usual 4-dimensional
scalar torsion.

\subsection{Equations of Motion}
The gravitational field equations on the bulk can be derived from the
formulation given
in~\cite{Gronwald:1997bx,Obukhov:2002tm}. The equations
of motion on $V$ are 4-forms
\begin{equation}\label{E:4-form EOM}
\bar{D}\bar{H}_A - \bar{E}_A = - 2\,\k_5\,{}^{(5)} \bar{\Sigma}_A\,,
\end{equation}
with
\begin{eqnarray}\label{Def:H_A,E_A}
\bar{H}_A &=& (-2) \bar{\star} \left( {}^{(1)}\bar{T}_A
  - 2 \, {}^{(2)}\bar{T}_A
  - \frac{1}{2} \, {}^{(3)}\bar{T}_A \right)\,, \nonumber\\
\bar{E}_A &:=& i_{\bar{e}_A}(\bar{\mathcal{T}})
  + i_{\bar{e}_A}(\bar{T}^B) \wedge \bar{H}_B\,,
  \nonumber\\
\bar{\Sigma}_A &:=& \frac{\delta \bar{L}_{mat}}{\delta \bvt^A}\,,
\end{eqnarray}
where $\bar{\Sigma}_A$ is the canonical energy-momentum 4-form of
matter fields, and   $\bar{H}_A$ can be
simplified as~\cite{Obukhov:2002tm}
\begin{equation}\label{E:H_A}
\bar{H}_A = \left( \bar{g}^{BC}\bar{K}^D_C \right) \wedge
\bar{\star}\left( \bvt_A \wedge \bvt_B \wedge \bvt_D \right)\,,
\end{equation}
with $\bar{K}^D_C := \bar{\omega}^D_C - \widetilde{\omega}^D_C$ being the
contortion 1-form.

Following the same procedure as in (\ref{E:Riemann connection}), we are able to derive the unique Levi-Civita
connection 1-form $\widetilde{\omega}^D_C$ with respect to the coframe in $V$, given by
\begin{equation}\label{E:Levi-Civita connection}
\begin{aligned}
\widetilde{\omega}^0_i &= \,\, \frac{\dot{a}}{a} \, \bvt^i,  \quad
\widetilde{\omega}^i_0 = \,\, \widetilde{\omega}^0_i ,\quad
\widetilde{\omega}^0_5 = \veps \frac{\dot{\phi}}{\phi} \, \bvt^n ,
\quad \widetilde{\omega}^5_0 = \veps \, \widetilde{\omega}^0_5\, , \\
\widetilde{\omega}^5_j &= - \veps \frac{a'}{\phi a} \, \bvt^j, \quad
\widetilde{\omega}^j_5 = -\veps \widetilde{\omega}^i_j,\quad
\widetilde{\omega}^i_j \equiv 0\,.
\end{aligned}
\end{equation}
From (\ref{E:H_A}) and (\ref{E:Levi-Civita connection}), some computation yields
the equations of  motion for the bulk:
\begin{equation} \label{E:Friedmann eqn}
\begin{aligned}
\bar{D} \bar{H}_0 - \bar{E}_0 &= 3 \left[ \left(
    \frac{\dot{a}^2}{a^2} + \frac{\dot{a}}{a} \frac{\dot{\phi}}{\phi}
    \right) - \frac{\veps}{\phi^2} \left( \frac{a''}{a} -
    \frac{a'}{a}\frac{\phi'}{\phi}\right) - \left(\frac{1+\veps}
    {2\phi^2}\right)\frac{a'^2}{a^2} \right] \bar{\star}\bvt_0 \\
    &\qquad \qquad \qquad \qquad \qquad \qquad \qquad \qquad+ \frac{3\veps}{\phi} \left( \frac{\dot{a}'}{a} - \frac{a'}{a}
    \frac{\dot{\phi}}{\phi}\right) \bar{\star}\bvt_5 = - \k_5 \,\bar{\Sigma}_0 \, , \\
\bar{D} \bar{H}_5 - \bar{E}_5
&= \frac{3}{\phi}\left(
    \frac{a'}{a}\frac{\dot{\phi}}{\phi} - \frac{\dot{a}'}{a}\right)
    \bar{\star}\bvt_0 + 3\left[ \left( \frac{\ddot{a}}{a} +
    \frac{2\dot{a}^2}{a^2}\right) - \left(\frac{1+\veps}{2\phi^2}\right)
    \frac{a'^2}{a^2} \right] \bar{\star}\bvt_5 \\
    &= - \k_5 \, \bar{\Sigma}_5\,.
\end{aligned}
\end{equation}
Thus $(\ref{E:Friedmann eqn})_1$ is the Friedmann equation of the bulk. Moreover, if we let $\veps = + 1$ and expand the energy-momentum 4-form $\bar{\Sigma}_A = \bar{T}_A^B \, \bar{\star} \bvt_B$, we obtain
\begin{equation}\label{E:Friedmann eqn 2}
\left( \frac{\dot{a}^2}{a^2} + \frac{\dot{a}}{a}
\frac{\dot{\phi}}{\phi} \right) - \frac{1}{\phi^2} \left(
\frac{a''}{a} - \frac{a'}{a}\frac{\phi'}{\phi}\right) -
\frac{1}{\phi^2} \frac{a'^2}{a^2} = \frac{\k_5}{3} \bar{T}_{00}
\end{equation}
Furthermore, if we consider our matter field as a perfect fluid, one can decompose the energy-momentum
tensor into bulk and brane parts as~\cite{Binetruy:1999ut}
\begin{eqnarray}
\bar{T}_A^B(t,y) &=& \left(\bar{T}_A^B\right)_{\text{bulk}}+
\left(\bar{T}_A^B\right)_{\text{brane}}\,,
\nonumber\\
\left(\bar{T}_A^B\right)_{\text{brane}} &=& \frac{\delta(y)}{\phi} \,
diag(-\rho(t),P(t),P(t),P(t),0)\,,
\end{eqnarray}
where $\left(\bar{T}_A^B\right)_{\text{bulk}}$ represents the
\emph{vacuum energy-momentum tensor} or the \emph{cosmological constant}
$(\Lambda_5/\k_5)\e_A^B$ in the bulk, and $\rho(t)$ and $P(t)$ are
the energy density and the pressure of the normal matter localized
on the brane, respectively.

If the first discontinuity appears in the first derivative of the bulk metric $\bar{g}$, or $\bar{g}\in C^0(M)\setminus\bigcup_{k= 1}C^{k}(M)$ to be precise,
the Dirac delta function would appear in the second derivative of the bulk metric.
The FLRW metric leads to the equation of the scale factor with the form at $y=0$
\begin{equation}
a''(t,y) = \delta(y) \, [a'](t,0) + \widetilde{a}''(t,y)\,,
\end{equation}
where $\widetilde{a}''$ denotes the non-distributional part of $a''$
and the definition of the \emph{jump} is
\begin{equation}
[f](0):= \lim_{\delta \to 0^+}  f(\delta) - f(-\delta) \qquad (f: M
\to \mathbb{R})\,,
\end{equation}
which measures the discontinuity of a real-valued function $f$ across the
brane. With the form of the scale factor, (\ref{E:Friedmann eqn 2})
yields the junction condition
\begin{equation}\label{E:Jump of scale factor}
[a'](t,0) = \frac{\k_5}{3\veps} \rho \, a_0(t) \, \phi_0(t)
\end{equation}
where $a_0(t):= a(t,0)$ and $\phi_0(t):=\phi(t,0)$ are considered as
the scalar factor and a scalar field on the brane, respectively.
Furthermore, if we impose the so-called $\mathbb{Z}_2$ symmetry
~\cite{Horava:1995qa} for the scale factor in Eq.~(\ref{E:Jump of
scale factor}) on the bulk as a real-valued quantity $f$ must be an
odd function $f(x) = -f(-x)$ across the brane, we obtain the
Friedmann equation on the brane to be
\begin{equation}
\label{E:Brane Friedmann eqn} \frac{\dot{a}_0^2(t)}{a_0^2(t)} +
\frac{\ddot{a}_0(t)}{a_0(t)} = - \frac{\k_5^2}{36}\, \rho(t) (\rho(t)
+ 3P(t)) - \frac{k_5}{3\phi_0^2(t)} \,
\left(\bar{T}_{55}\right)_{\text{bulk}}\,,
\end{equation}
which is, unexpectedly, found to be the same as
the braneworld theory of GR shown in~\cite{Binetruy:1999ut}.
Hence, this indicates that the cosmological braneworld scenario in TEGR coincides
with that of GR, i.e,, there is no distinction between TEGR and GR in
the braneworld FLRW cosmology, which again justifies the name of TEGR.

The physical consequence of the cosmological brane scenario here
then follows from the discussions in~\cite{Binetruy:1999ut}. In particular, if the extra
5th-dimension is compact, one can check if the solutions of
$a(t,y)$ and $\phi(t,y)$ derived from (\ref{E:Friedmann eqn}) are
well-defined  ones, as given in \cite{Binetruy:1999ut}. In effect, we established the diagram
\[
\]
\xymatrix{
\text{5D TEGR} \ar@{.>}[d]_{\text{reduction}} \ar@2{<->}[r] \ar@{-->}[rd] & \text{5D GR}  \ar@{.>}[d]^{\text{reduction}} \\
\text{4D TEGR} \ar@2{<->}[r] & \text{4D GR}  }
\[
\]

Although the result may appear to be obvious, in fact there exists some non-triviality within the reduction. The 5-dimensional reduction to the 4-dimensional of the TEGR and GR involves different connections, the Levi-Civita and the Weitzenb\"{o}ck connection, and different geometric projections. Thus after some moment of thoughts, one realizes it is not that transparent as one though it was.

Finally, we remark that the Friedmann equation (\ref{E:Friedmann
eqn}) in the bulk can be identified as $G_{00} = - \k_5 \bar{T}_{00}$
and $G_{05}=0$, which are the same as those
in~\cite{Binetruy:1999ut}. This result implies that a
\emph{radiating} contribution of the universe can be generated in
TEGR due to the extra spatial dimension. It can be viewed as a
generic property that there exists a component of \emph{dark
radiation} in the braneworld scenario. We have to mention that there
is no extrinsic curvature in TEGR since the projected effects of the
dark radiation and discontinuity property of the brane come from
torsion itself, which is clearly beyond the expectations of
GR~\cite{PonceDeLeon:2001un,Maartens:2010ar} as already pointed out in~\cite{Nozari:2012qi}.

\chapter{Conclusion} \label{conclusion}

In this thesis, we have given a complete view of Poincar\'{e} gauge gravity starting from the affine frame bundle $(\mathbb{A}(M),\widetilde{\pi}, M,GL(\mathbb{R}^{1,3}),\widetilde{\o})$ to Riemann-Cartan spacetime $(M,g,\nabla)$. It can be seen from such a viewpoint that the torsion is a natural byproduct of gauging the Poincar\'{e} group $\mathcal{P} = \mathbb{R}^{1,3} \rtimes SO(1,3) $ into gravity, simply due to the existence of the canonical 1-form $\varphi\in \overline{\Lambda}^1(L(M),\mathbb{R}^{1,3})$ of the frame bundle $L(M)$ and the decomposition of the Lie algebra $\mathfrak{P} = \mathbb{R}^{1,3} \oplus \mathfrak{so}(1,3)$.

We have also discussed one of the interesting models in PGT, the torsion-scalar mode that does support the late-time universe acceleration, and investigated its features in a cosmological model. In particular, we have studied that the energy density ratio $\Omega_T$ and the torsion EoS
$w_T$ in two main cases of the scalar-torsion mode in PGT.

For the first case of the negative energy matter density with negative constant affine curvature $R<0$, the torsion EoS $w_T$ demonstrates the same behavior as
the background fluid in the high redshift regime: in the case of the radiation-dominated era we have $w_T=1/3$, in the case $w_T=0$ of the matter-dominated era it is $w_T = 0$ and in a later stage of the de-Sitter point we also have $w_T = -1$. We also observe that the torsion density ratio of $\O_T$ in the high redshift regime becomes a \emph{negative} constant.

In the second case of the scalar-torsion mode where the positive kinetic energy
condition holds, the numerical solution of the field equations shows that in general $w_T$ has an asymptote to $1/3$ in the high redshift
regime, while it could cross the phantom divide line in the low redshift regime. With a further analysis we find that such asymptotic behavior of $w_T$ and $\O_T$ in fact can be resolved by a semi-analytical solution. In principle, we apply a Laurent series expansion in the scale factor $a(t)$ for the torsion density
$\rho_T$ and find that the series in fact has a cut-off as a the lower bound at $O(a^{-4})$ in the high redshift, corresponding to
a radiation-like behavior. By a
comparison of the next leading-order term of $a^{-3}$ in the field
equations, we are able to extract the coefficient $A_3 = -\mu / a_1$, which results
in the vanishing of the $a^{-3}$ term in the affine curvature $R$,
such that $R$ is only proportional to $a^{-2}$, and this result is consistent with the
numerical demonstration.


In the last two chapters, we regard another special geometry descendant from PGT, teleparallel gravity. A clear definition for the Weitzenb\"{o}ck spacetime and its parallelism are provided. In particular, using the basic differential forms on the principal fibre bundle, we are able to understand the reason of local Lorentz violation in teleparallel gravity that generally occurs.

We also construct the extra dimension theory for TEGR in five-dimensional spacetime. With the use of Cartan's moving frame by means of differential forms, we can find the torsion relations between the brane and the bulk, analogous to the Gauss-Codazzi equation in GR. In particular, from the extra dimension theory rigorously constructed, we can show that the Kaluza-Klein theory in
teleparallel gravity does not generate a Brans-Dicke type of the
effective 4-dimensional Lagrangian as in GR. This result disagrees with the one given in~\cite{Bamba:2013fta}.

We further apply our theory as the braneworld theory of teleparallel gravity and investigate its FLRW cosmology solution. To our surprise, it provides equivalent field equations and hence the same solutions as Einstein's general relativity. We thus conclude that the additional radiation of the universe can arise from the extra dimension, which is a generic feature of the branworld theory.



\appendix

\chapter{The Spin Current}\label{APP:spin current}

Let $\phi \in \Gamma(U)\otimes V$ be a particle field, a vector-valued local section on spacetime $U \subseteq M$, where $V$ is a vector space under consideration (in general $\mathbb{R}^n$ or $\mathbb{C}^n$ ). Consider the infinitesimal variation of $\phi$, denoted by $\delta \phi$
\begin{align}
\delta \int \mathcal{L}_{\text{M}} \left( x , \vt^\a(x) ,\phi^a (x), \phi^a_{;k}(x) \right) \, dvol^4 &= \int \left( \frac{\p \mathcal{L}_{\text{M}}}{\p \phi^a}\, \delta \phi^a + \frac{\p \mathcal{L}_{\text{M}}}{\p \phi^a_{;k}} \, \delta \phi^a_{;k} \right) \, dvol^4 , \\
&=\int \left[ \frac{\p \mathcal{L}_{\text{M}}}{\p \phi^a} \, \delta \phi^a + D_{\p_k} \left( \frac{\p \mathcal{L}_{\text{M}}}{\p
\phi^a_{;k}}\delta \phi^a \right) - D_{\p_k} \left( \frac{\p L}{\p \phi^a_{;k}} \right)\delta
\phi^a \right] \, dvol^4\\
&= \int \left[ \left( \frac{\p \mathcal{L}_{\text{M}}}{\p \phi^a}- D_{\p_k} \left( \frac{\p \mathcal{L}_{\text{M}}}{\p
\phi^a_{;k}} \right) \right) \, \delta \phi^a  + D_{\p_k} \left( \frac{\p \mathcal{L}_{\text{M}}}{\p \phi^a_{;k}} \, \delta
\phi^a \right) \right] \, dvol^4
\end{align}
where $dvol^4 = \sqrt{-g} \, d^4x $, $D$ is the exterior covariant derivative defined in (\ref{Def:Exterior covariant derivative}), $\delta \phi^a_{;k} :=  \delta (D \phi^a )( \p_k ) =  D (\delta \phi^a )( \p_k ) $ due to the fixed connection, and we have applied the integration by parts in the second equality. If we define a vector field $J^k:= \frac{\p \mathcal{L}_{\text{M}}}{\p \phi^a_{;k}} \, \delta
\phi^a$ (\textbf{general conserved current}) and assume $\phi$ satisfies the \textit{Euler-Lagrange equation} (on-shell)
\begin{equation}
 \frac{\p \mathcal{L}_{\text{M}}}{\p \phi^a}- D_{\p_k} \left( \frac{\p \mathcal{L}_{\text{M}}}{\p
\phi^a_{;k}} \right) = 0,
\end{equation}
with the infinitesimal invariant (under a Lie group $G$) of the Lagrangian, we obtain,
\begin{equation}\label{E:int div}
\delta \int \mathcal{\mathcal{L}_{\text{M}}} \left( x ,\phi^a (x), \phi^a_{;k}(x) \right) \, dvol^4 = \int D_{\p_k} \left( \frac{\p \mathcal{L}_{\text{M}}}{\p
\phi^a_{;k}}\delta\phi^a \right) \, dvol^4
\end{equation}
which will lead to Noether's first theorem for the $G$-symmetry. Now we carefully define the meaning of infinitesimal variation $\delta \phi$. Let $\mathfrak{g}$ be the Lie algebra of $G$, $\rho: G \to GL(V)$ be a Lie group representation, and $E \in \mathfrak{g}$ is an element. Let $ \phi(t):= \rho \left( e^{tE} \right) \cdot \phi $ and define $\d \phi$ as the variation of $\phi(t)$ in a 1-parameter subgroup of $G$,
\begin{equation}
\delta\phi := \frac{d}{dt}\phi(t)\bigr|_{t=0} = \left( \rho_* E \right) \cdot \phi
\end{equation}
where in components we write
\begin{equation}
\delta\phi^a=\left( \rho_* E \right)^a_b \, \phi^b
\end{equation}
where $\rho_*: \mathfrak{g} \to \mathfrak{gl}(V)$ is the induced Lie algebra representation, $(\xi_a)\in V$ is a set of chosen basis for $V$, and we expand $\phi := \phi^a \, \xi_a$ so that $\left( \rho_* E \right)^a_b$ denotes the matrix representation of $\rho_* E$ under the basis $\xi_a$. Thus (\ref{E:int div}) in view of the invariance of the Lagrangian under $G$ leads (on shell) to
\begin{equation}
Div (J) = D_{\p_k} \, J^k \equiv 0
\end{equation}
In particular, if we take the Lorentz symmetry with $G=SO(1,3)$, we obtain a conserved current corresponding to each $E_{ij}$ generator in $\mathfrak{so}(1,3)$, i.e,
\begin{equation}
J^k = \frac{\p \mathcal{L}}{\p
\phi^a_{;k}}( \rho_* E_{ij})^a_b \, \phi^b
\end{equation}
which is called the \textbf{canonical spin angular momentum} or (simply \textbf{spin current}).
\begin{equation}\label{E:spin current2}
S_{ij}{}{}^k= \left( \frac{\p \mathcal{L}_{\text{M}}}{\p \phi^a_{;k}} \right) \, \rho_{[ij]}{}{}_b{}{}^a \cdot \phi^b = - S_{ji}{}{}^k , \qquad \left( \mbox{where} \, \, \rho_{[ij]}{}{}_b{}{}^a := \left( \rho_* E_{ij} \right)^a_b \right)
\end{equation}

It follows that we have $\rho_{[ij]}{}{}_b{}{}^a := \left( \rho_* E_{ij} \right)^a_b$ due to the property $E_{ij} = - E_{ji}$ ($\forall \,\, i,j$) for generators of $\mathfrak{so}(1,3)$. Notice that the indices $a,b$ and $i,j$ are indicating different objects, which should be distinguished. Indices $a,b$ correspond to the vector basis $(e_a)\in V$, hence the matrix indices, while $i,j$ specify which Lie algebra basis of $\mathfrak{g}$ we are referring to, not the spacetime coordinate.


\newpage

\backmatter

\newpage

\addcontentsline{toc}{section}{Appendix}  
\appendix

\end{document}